\documentclass[12pt,paper=a4,twosided]{book}
                      
\newlength{\dinwidth}
\newlength{\dinmargin}

\usepackage{amsmath, amsfonts, amssymb, amscd}
\usepackage{url}
\usepackage[pdfborder={0 0 0}, colorlinks=true, linkcolor=black, citecolor=black, urlcolor=gray, breaklinks=false]{hyperref}
\usepackage{graphicx}
\usepackage{mathrsfs}
\usepackage{amsthm}
\usepackage{braket}
\usepackage{accents}
\usepackage{oldgerm}
\usepackage{color}
\usepackage{hyperref}
\usepackage{yfonts}
\usepackage{ctable}
\usepackage{wrapfig}
\usepackage{longtable}

\definecolor{gray}{rgb}{.17,.21,.21}
\newcommand{\unit}{\mathbf{1}}
\newcommand{\pp}{{\boldsymbol p}}
\newcommand{\qq}{{\boldsymbol q}}
\newcommand{\hhs}{\mathfrak{h}}
\newcommand{\xx}{{\boldsymbol x}}

\newcommand{\ee}{{\boldsymbol e}}
\newcommand{\vveps}{{\boldsymbol \veps}}

\newcommand{\Aa}{\mathfrak{A}}
\newcommand{\Euc}{\mathrm{E}_4(3)}
\renewcommand{\to}{\mapsto}
\newcommand{\M}{\frak{M}}

\newcommand{\gST}{\mathrm{\bf g}}
\newcommand{\Ff}{\mathscr{F}}

\newcommand{\Ibb}[1]{ {\rm I\ifmmode\mkern -3.6mu\else\kern -.2em\fi#1}}
\newcommand{\ibb}[1]{\leavevmode\hbox{\kern.3em\vrule
     height 1.2ex depth -.3ex width .2pt\kern-.3em\rm#1}}
\newcommand{\Cl}{{\ibb C}}
\newcommand{\Rl}{{\Ibb R}}
\newcommand{\Nl}{{\Ibb N}}
\newcommand{\Zl}{\mathbb{Z}}
\newcommand{\Quat}{\mathbb{H}}

\newcommand{\ra}{\rightarrow}
\newcommand{\lra}{\longrightarrow}

\newcommand{\uhr}{\upharpoonright}
\newcommand{\ua}{\uparrow}

\newcommand{\Bl}{\biggl}
\newcommand{\Br}{\biggr}
\newcommand{\bl}{\bigl}
\newcommand{\br}{\bigr}

\newcommand{\Aut}{\text{\normalfont\textrm{Aut}}}

\newcommand{\GL}{\textrm{GL}}

\newcommand{\id}{\mathrm{id}}

\newcommand{\bs}{\boldsymbol}
\renewcommand{\cal}{\mathcal}

\newcommand{\cc}{\frak{C}}

\newcommand{\diag}{\mathrm{diag}}
\newcommand{\Tr}{\mathrm{Tr}}
\newcommand{\im}{\mathrm{Im}}

\newcommand{\supp}{\text{\normalfont\textrm{supp}}}

\newcommand{\BH}{\cal{B(H)}}

\newcommand{\bj}{{\boldsymbol j}}

\newcommand{\End}{\mathrm{End}}

\newcommand{\gdS}{\mathrm{\bf g}}

 \newcommand{\CAR}{\mathrm{CAR}}

\newcommand{\del}{\partial}

\newcommand{\Mat}{\mathrm{Mat}}
\newcommand{\G}{\cal{G}}
\newcommand{\F}{\cal{F}}

\newcommand{\SO}{\mathrm{SO}}

\newcommand{\Uone}{\mathrm{U(1)}}
\newcommand{\Spin}{\mathrm{Spin}}
\newcommand{\Sp}{\mathrm{Sp}}
\newcommand{\SL}{\mathrm{SL}}

\renewcommand{\L}{\cal{L}}
\newcommand{\tLo}{\widetilde{\L}_0}

\newcommand{\Lpo}{\cal{L}^\ua_+}

\newcommand{\vtheta}{\vartheta}
\newcommand{\veps}{\varepsilon}

\newcommand{\A}{\frak{A}}
\newcommand{\B}{\cal{B}}

\newcommand{\K}{\cal{K}}
\renewcommand{\P}{\cal{P}}
\renewcommand{\B}{\cal{B}}

\newcommand{\spec}{\mathrm{spec}}

\newcommand{\HS}{\cal{H}}

\newcommand{\hs}{\frak{H}}

\newcommand{\OO}{\mathcal{O}}

\newcommand{\W}{\cal{W}}

\newcommand{\U}{\mathrm{{\bf U}}}

\newcommand{\C}{\cal{K}}
\renewcommand{\S}{\cal{S}}

\newcommand{\flow}{\varphi}
\newcommand{\iso}{{\rm Iso}}

\newcommand{\Om}{\Omega}
\newcommand{\om}{\omega}

\newcommand{\la}{\lambda}

\newcommand{\eps}{\varepsilon}

\newcommand{\Hil}{\mathcal{H}}

\newcommand{\DD}{\mathcal{D}}
\newcommand{\Ss}{\mathscr{S}}   % Schwartz space

\newcommand{\hti}{\tilde{h}}

\newcommand{\Wti}{\tilde{W}}

\newcommand{\pti}{\tilde{p}}
\newcommand{\xiti}{\tilde{\xi}}

\newcommand{\fF}{\frak{F}}

\newcommand{\frA}{\frak A}

\def\bs{{\mbox{\boldmath{$s$}}}}

\newcommand{\CCR}{\mathrm{CCR}\,}
\theoremstyle{plain}
\newtheorem{theorem}{Theorem}[section]
\newtheorem{proposition}[theorem]{Proposition}
\newtheorem{lemma}[theorem]{Lemma}
\newtheorem{corollary}[theorem]{Corollary}

\theoremstyle{definition}
\newtheorem{definition}[theorem]{Definition}
\newtheorem{remark}[theorem]{Remark}

\theoremstyle{remark}

\numberwithin{equation}{section}
\newcommand{\vphi}{\varphi}
\newcommand{\Iso}{\mathrm{Iso}}

\DeclareMathAlphabet{\mathpzc}{OT1}{pzc}{m}{it}

\definecolor{lightgray}{rgb}{0.8,0.8,0.8}
\usepackage{fancyhdr} 
\allowdisplaybreaks

\begin{document} 
%===========================================
%===========================================
% Title page
\thispagestyle{empty}
${}$
\vspace{3cm}
\begin{center}

%-------------------------------------------
\textbf{\LARGE Deformations of Quantum Field Theories}

\vspace{0.2cm}
\textbf{\LARGE on Curved Spacetimes}

%-------------------------------------------
\vspace{1.5cm}

by

\vspace{0.05cm}
Eric Morfa-Morales
%-------------------------------------------

%-------------------------------------------
\vspace{10cm}
A thesis presented to obtain the degree

\vspace{0.1cm}
{\it Doktor der Naturwissenschaften}

\vspace{0.1cm}
from the

\vspace{0.1cm}
University of Vienna, Faculty of Physics
%-------------------------------------------

%-------------------------------------------
\vspace{2cm}
Vienna 2012
%-------------------------------------------
\end{center}

%===========================================
\thispagestyle{empty}

\cleardoublepage

%===========================================
% Abstract
%===========================================
\thispagestyle{empty}
${}$
\vspace{0.0cm}

\begin{center}
 {\bf Abstract}
\end{center}
The construction and analysis of deformations of quantum field theories by warped convolutions is extended to a class of globally hyperbolic spacetimes.
First, we show that any four-dimensional spacetime which admits two commuting and spacelike Killing vector fields carries a family of wedge regions with causal properties analogous to the Minkowski space wedges. Deformations of quantum field theories on these spacetimes are carried out within the operator-algebraic framework $-$ the emerging models share many structural properties with deformations of field theories on flat spacetime. In particular, deformed quantum fields are localized in the wedges of the considered spacetime. As a concrete example, the deformation of the free Dirac field is studied. 
Second, quantum field theories on de Sitter spacetime with global $\Uone$ gauge symmetry are deformed using the joint action of the internal symmetry group and a one-parameter group of boosts. The resulting theories turn out to be wedge-local and non-isomorphic to the initial one for a class of theories, including the free charged Dirac field. The properties of deformed models coming from inclusions of $\CAR$-algebras are studied in detail.
Third, the deformation of the scalar free field in the Araki-Wood representation on Minkowski spacetime is discussed as a motivating example.
%===========================================
\vspace{1cm}

%===========================================
\begin{center}
 {\bf Zusammenfassung}
\end{center}
Die Konstruktion und Analyse von deformierten Quantenfeldtheorien durch warped convolutions wird auf eine Klasse von global hyperbolischen Raumzeiten verallgemeinert. Erstens, es wird gezeigt, dass jede vierdimensionale Raumzeit, welche zwei kommutierende und raumartige Killing-Vektorfelder erlaubt, eine Familie von Keilregionen tr\"agt, deren kausale Eigenschaften analog zu Minkowski-Keilen sind. Deformationen von Quantenfeldtheorien auf solchen Raumzeiten werden im operator-algebraischen Rahmen durchgef\"uhrt $-$ die strukturellen Eigenschaften der entstehenden Modelle sind \"ahnlich wie in Falle von flachen Raumzeiten. Insbesondere sind die deformierten Quantenfelder in den Keilen der jeweiligen Raumzeit lokalisiert. Als konkretes Beispiel wird die Deformation des freien Dirac-Feldes diskutiert.
Zweitens, Quantenfeldtheorien auf der de-Sitter-Raumzeit mit globaler $\Uone$-Eichsymmetrie werden mithilfe der gemeinsamen Wirkung der inneren Symmetrien und einer Ein-Parameter-Gruppe von Boosts deformiert. Die resultierenden Theorien sind Keil-lokal und nicht-isomorph zu der urspr\"unglichen Theorie im Falle einer Klasse von Theorien, welche das freie geladene Dirac-Feld enth\"alt. Die Eigenschaften deformierter Modelle, welche aus Inklusionen von $\CAR$-Algebren entstehen, werden im Detail untersucht. Drittens, die Deformation des freien Skalarfeldes in der Araki-Woods-Darstellung auf der Minkowski-Raumzeit wird als motivierendes Beispiel behandelt.
%===========================================
\thispagestyle{empty}
\cleardoublepage
%===========================================
\tableofcontents
\newpage
%===========================================
%%%%%%%%%%%%%%%%%%%%%%%%%%%%%%%%%%%%%%%%%%%%%%%%%%%%%%%%%%%%
\chapter{Introduction}
%%%%%%%%%%%%%%%%%%%%%%%%%%%%%%%%%%%%%%%%%%%%%%%%%%%%%%%%%%%%
Deformations of quantum field theories arise in different contexts and have been studied from different points of view in recent years. One motivation for considering such models is a possible noncommutative structure of spacetime at small scales, as suggested by combining classical gravity and the uncertainty principle of quantum physics \cite{DoplicherFredenhagenRoberts:1995}. Quantum field theories on such noncommutative spaces can then be seen as deformations of usual quantum field theories, and it is hoped that they might capture some aspects of a still elusive theory of quantum gravity ({\em cf.} \cite{Szabo:2003} for a review). By now there exist several different types of deformed quantum field theories (see \cite{GrosseWulkenhaar:2005,BalachandranPinzulQureshiVaidya:2007,Soloviev:2008,GrosseLechner:2008,BlaschkeGieresKronbergerSchwedaWohlgenannt:2008,BahnsDoplicherFredenhagenPiacitelli:2010} and references cited therein for some recent papers).

Certain deformation techniques arising from such considerations can also be used as a device for the construction of new models in the framework of usual quantum field theory on commutative spaces  \cite{GrosseLechner:2007,BuchholzSummers:2008,GrosseLechner:2008,BuchholzLechnerSummers:2010,LongoWitten:2010}, independent of their connection to the idea of noncommutative spaces. From this point of view, the deformation parameter plays the role of a coupling constant which changes the interaction of the model under consideration, but leaves the classical structure of spacetime unchanged.

Deformations designed for either describing noncommutative spacetimes or for constructing new models on ordinary spacetimes have been studied mostly in the case of a flat manifold, either with a Euclidean or Lorentzian signature. In fact, many approaches rely on a preferred choice of Cartesian coordinates in their very formulation, and do not generalize directly to curved spacetimes. The analysis of the interplay between spacetime curvature and deformations involving noncommutative structures thus presents a challenging problem. As a first step in this direction, we study in this thesis how certain deformed quantum field theories can be formulated in the presence of external gravitational fields, {\em i.e.} on curved spacetime manifolds (see also \cite{AschieriBlohmannDimitrijevicMeyerSchuppWess:2005,OhlSchenkel:2009} for other approaches to this question). We will not address here the fundamental question of dynamically coupling the matter fields with a possible noncommutative geometry of spacetime  \cite{PaschkeVerch:2004, Steinacker:2007}, but rather consider, as an intermediate step, deformed quantum field theories on a fixed Lorentzian background manifold.

\newpage
%%%%%%%%%%%%%%%%%%%%%%%%%%%%%%%%%%%%%%%%%%%%%%%%%%%%%%%%%%%%
\section{The algebraic approach to quantum field theory}
\label{sec:AQFT}
%%%%%%%%%%%%%%%%%%%%%%%%%%%%%%%%%%%%%%%%%%%%%%%%%%%%%%%%%%%%
In this section we describe the framework for our discussion of quantum field theories on Minkowski spacetime \cite{Haag:1996, Araki:1999} and its extension to Lorentzian manifolds \cite{Dimock:1980}. 
\\
\\
Algebraic quantum field theory (local quantum physics) is a model-independent approach to relativistic quantum physics which uses techniques from operator algebras.
The fundamental objects in its formulation are the basic building blocks for the experimental description of quantum systems, namely, measuring instruments (observables) and systems being measured (states). In physical terms, states are (equivalence classes of) preparation devices producing ensembles of quantum systems and observables are (equivalence classes of) measuring devices which are applied to the quantum system and yield a ``result'' \cite{Araki:1999}.

Within the algebraic framework, observables are represented by selfadjoint elements in a (unital) C$^*$-algebra. Since measurements are performed over a finite time and over a finite spatial extent, each observable carries an intrinsic localization. Hence, one assumes that for each open and bounded subset $\OO\subset \Rl^4$ of Minkowski spacetime\footnote{See appendix \ref{ch:NotationConventions} for our conventions and notation concerning Minkowski spacetime.} $(\Rl^4,\eta)$, there exists a C$^*$-algebra $\Aa(\OO)$ containing all the (bounded) observables which are measurable in $\OO$. We refer to the local algebras of observables $\Aa(\OO)$ as \emph{local algebras}. 
The map 
\begin{equation*}
 \OO\longmapsto \Aa(\OO)
\end{equation*}
is required to fulfill a number of conditions, which every physically meaningful theory should satisfy: Since larger regions should contain more observables, one requires
\begin{enumerate}
 \item[$1)$] (Isotony). If $\OO_1\subset\OO_2$, then $\Aa(\OO_1)\subset\Aa(\OO_2)$,
\end{enumerate}
where $\Aa(\OO_1)\hookrightarrow\Aa(\OO_2)$ is a unital embedding. This gives $\OO\mapsto\Aa(\OO)$ the structure of a net, {\it i.e.} an inclusion-preserving mapping. Consequently $\{\Aa(\OO):\OO\subset\Rl^4\}$ is a directed system of C$^*$-algebras, 
so there exists a (up to isomorphism) unique C$^*$-algebra $\Aa$ (the inductive limit of the directed system), such that $\bigcup_\OO\Aa(\OO)$ is dense in $\Aa$ \cite{KadisonRingrose:1986}.
We refer to $\Aa$ as \emph{quasilocal algebra}.
Since measurements in spacelike separated regions should not interfere with each other by Einstein causality, one demands that the corresponding observables commute:
\begin{enumerate}
 \item[$2)$] (Locality). If $\OO_1\subset \OO_2\, '$, then $\Aa(\OO_1)\subset\Aa(\OO_2)'$,
\end{enumerate}
where $\OO'$ denotes the causal complement of $\OO$ and the commutant is understood as relative commutant in $\Aa$. Relativistic covariance is expressed by demanding that there exists a continuous representation $\alpha:\P_0\ra\Aut(\Aa)$ of the identity component $\P_0$ of the Poincar\'e group by automorphisms on $\Aa$, which acts in a geometrically compatible way on the net:
\begin{enumerate}
 \item[$3)$] (Covariance). $\alpha_h(\Aa(\OO))=\Aa(h\OO),\quad h\in\P_0$,
\end{enumerate}
where $h\OO:=\{hx:x\in\OO\}$ denotes the transformed spacetime region. We refer to a map $\Aa:\OO\mapsto\Aa(\OO)$ satisfying conditions 1), 2) and 3) as a \emph{local net}\label{localnet}.

\label{state}
States are represented in the algebraic framework by continuous linear maps $\omega:\Aa\ra\Cl$ satisfying $\omega(\boldsymbol{1})=1$ and $\omega(A^*A)\ge 0$ for all $A\in\Aa$, {\em i.e.} states are particular elements in the dual space of $\Aa$. Depending on the physical context, one demands further properties ({\it e.g.} Poincar\'e-invariance, spectrum condition, KMS condition) to select the states of actual interest. 

The familiar Hilbert space formulation of quantum field theory is recovered by means of the the Gelfand-Naimark-Segal (GNS) construction: The pair $(\Aa,\omega)$ yields a (up to unitary equivalence) unique triple $(\HS_\omega,\pi_\omega,\Omega_\omega)$, where $\HS_\omega$ is a Hilbert space, $\Omega_\omega\in \HS_\omega$ is a unit vector and $\pi_\omega:\Aa\ra \cal{B}(\HS_\omega)$ is a C$^*$-homomorphism (representation), such that
\begin{equation*}
 \omega(A)=\left<\Omega_\omega,\pi_\omega(A)\Omega_\omega\right>_{\HS_\omega},\quad A\in\Aa
\end{equation*}
and $\pi_\omega(\Aa)\Omega_\omega$ is dense in $\HS_\omega$. Each $\omega(A)$, $A\in\Aa$ is physically interpreted as the expectation value of the observable $A$ in the state $\omega$. If the representation $\alpha$ is implementable, as it is the case for Poincar\'e-invariant states, there exists a unitary representation $U_\omega:\P_0\ra \End(\HS_\omega)$, such that
\begin{equation*}
	\pi_\omega(\alpha_h(A))=U_\omega(h)\pi_\omega(A)U_\omega(h)^{-1},\quad h\in\P_0,\;A\in\Aa.
\end{equation*}
The representation $\OO\mapsto \pi_\omega(\Aa(\OO))''$ of the local net by von Neumann algebras also satisfies conditions $1)$, $2)$ and $3)$, since algebraic relations are preserved under the homomorphism $\pi_\omega$. The quasilocal algebra $\Aa$ possesses, in general, an abundance of inequivalent representations $-$ this fact is used to explain superselection rules \cite{DoplicherHaagRoberts:1969}.

To summarize, a quantum field theory in the algebraic sense is a Hilbert space representation of a local net, {\em i.e.} a collection of objects $\{\{\Aa(\OO)\}_{\OO\subset\Rl^4},\{\alpha_h\}_{h\in\P_0},\omega\}$ satisfying conditions $1)$, $2)$ and $3)$, where the state $\omega$ is required to meet certain physically-motivated regularity conditions.
\\
\\
The algebraic formulation of quantum field theory was generalized to Lorentzian spacetime\footnote{See Appendix \ref{ch:NotationConventions} for our conventions and notation concerning curved spacetimes.} manifolds $(M,\gST)$ in \cite{Dimock:1980}. Again, one demands that there exists a C$^*$-algebra $\Aa$ containing C$^*$-subalgebras $\Aa(\OO)$ which are associated with open and bounded spacetime region $\OO\subset M$, such that $\OO\mapsto \Aa(\OO)$ satisfies certain requirements. Conditions $1)$ and $2)$ carry over directly to the curved setting, whereas the notion of spacelike separation is determined by the causal structure of the underlying spacetime (see appendix \ref{ch:NotationConventions}). For a spacetime $(M,\gdS)$ which has a non-trivial global symmetry group, the Poincar\'e group $\P_0$ in condition $3)$ is replaced by the isometry group $\Iso(M,\gST)$ of $(M,\gST)$. Covariance is then formulated in the following way \cite{Dimock:1980}:
\begin{enumerate}
	\item[$3')$] (Covariance). There exists a continuous representation $\alpha:\Iso(M,\gST)\ra\Aut(\Aa)$, such that
	\begin{equation*}
		\alpha_h(\Aa(\OO))=\Aa(h\OO)
	\end{equation*} 
	holds for all $h\in\Iso(M,\gST)$ and every $\OO\subset M$.
\end{enumerate}

%-----------------------------------------------------------------------
\begin{remark}
\label{rem:BFV}
 (1) A generic curved spacetime has a trivial isometry group, so this condition is empty in this case. However, the spacetimes which we consider in this thesis ({\em e.g.} Bianchi  type models I-VII) do have a non-trivial global isometry group and these isometries will play an important role in the following, in particular for our definition of deformations of quantum field theories. Hence we adopt condition $3')$ as the notion of covariance.

 (2) Note that condition $3')$ is of a \emph{global} type, since it relies on the presence of a global isometry group. A concept of covariance which uses \emph{local} conditions is provided by the principle of general local covariance \cite{BrunettiFredenVerch:2003}. A framework for quantum field theories on {\em generic} globally hyperbolic spacetimes which makes use of this concept can be found in \cite{HollandsWald:2008}.
\end{remark}
%-----------------------------------------------------------------------

For quantum field theories on curved spacetimes, the notion of state is the same as on flat spacetime, {\em i.e.} a state is a continuous linear map $\omega:\Aa\ra\Cl$ satisfying $\omega(\boldsymbol{1})=1$ and $\omega(A^*A)\ge 0$ for all $A\in\Aa$. But restrictive  selection criteria for states of physical interest, such as Poincar\'e-invariance or positivity of energy, are not available in this setting and constitute one of the main difficulties in its formulation. We note that the microlocal spectrum condition \cite{Radzikowski:1996, BrunettiFredenhagenKohler:1996} provides, to a certain extent, a replacement for the relativistic spectrum condition on flat spacetime. However, in this thesis we will not make use of the microlocal spectrum condition. Since we do not work with generic curved spacetimes, but rather with spacetimes $(M,\gST)$ which possess a sufficiently large isometry group $\iso(M,\gST)$, we consider states which are $\iso(M,\gST)$-invariant, such as the Bunch-Davies state in the case of full de Sitter spacetime.

%%%%%%%%%%%%%%%%%%%%%%%%%%%%%%%%%%%%%%%%%%%%%%%%%%%%%%%%%%%%
\section{Constructive algebraic quantum field theory}
\label{sec:CAQFT}
%%%%%%%%%%%%%%%%%%%%%%%%%%%%%%%%%%%%%%%%%%%%%%%%%%%%%%%%%%%%
In this section we review a constructive technique for quantum field theories on Minkowski spacetime called warped convolution \cite{BuchholzSummers:2008, BuchholzLechnerSummers:2010}, which uses a certain deformation procedure for C$^*$-algebras. For a much more comprehensive overview of the current status of constructive algebraic quantum field theory we refer to \cite{Summers:2010}.
\\
\\
Once a local net (in a vacuum representation\footnote{The corresponding states appear in the description of systems in elementary particle physics.}) is given, an enormous amount of physical information can be extracted, such as the superselection structure, the particle content and collision cross sections (see \cite{BuchholzHaag:2000, BuchholzSummers:2005} and references cited therein). However, the task of explicitly {\em constructing} non-trivial examples has turned out to be particularly difficult, especially in four dimensions. Constructive quantum field theory has celebrated great successes in two and three dimensions by using functional integral methods in Euclidean space \cite{GlimmJaffe:1981}, but the physical case of four spacetime dimensions remains an open challenge to mathematical physics up to this day.

A novel approach towards the construction of models within the algebraic framework has emerged in the last few years \cite{Schroer:1997, Schroer:1999, SchroerWiesbrock:2000, Lechner:2003, Lechner:2005, Lechner:2008, BuchholzSummers:2007, BuchholzSummers:2008, BuchholzLechnerSummers:2010}. A common theme in these works is that one tries to avoid the direct construction of strictly localized quantities in a first step, and rather considers auxiliary non-local fields, which are associated with certain unbounded wedge-shaped spacetime regions called {\em wedges}. A wedge in four-dimensional Minkowski space is a region which is bounded by two non-parallel null planes. More specifically, one starts with a reference wedge 
\begin{equation*}
W_0:=\{x\in\Rl^4:x^1>|x^0|\} 
\end{equation*}
and defines the family of wedges $\W:=\{hW_0:h\in\P_0\}$ as the Poincar\'e-orbit of $W_0$. In a general relativity context these regions also go under the name of Rindler spacetime.
In the following, we refer to a map $\W\ni W\mapsto \Aa(W)$ satisfying condition $1)$, $2)$ and $3)$ from Section \ref{sec:AQFT} as {\em wedge-local net}\label{wedgelocalnet} and to the algebras $\Aa(W)$ as {\em wedge algebras}. The main advantage of considering nets over wedges is that a wedge $W$ is large enough to make the construction of observables which are localized in $W$ doable, but also small enough to make a complete description of the causal structure of Minkowski spacetime feasible \cite{ThomasWichmann:1997} (see below).
Once a wedge-local net is given, one can proceed to a local net over double cones $\OO_\diamond$ by taking suitable intersections:
\begin{equation}
\label{eq:LNfromWLN}
 \OO_\diamond\longmapsto\Aa(\OO_\diamond):=\bigcap_{W\supset \OO_\diamond}\Aa(W).
\end{equation}
That this indeed defines a local net relies on the peculiar properties of wedges \cite{ThomasWichmann:1997}: (1) every double cone $\OO_\diamond$ can be written as $\OO_\diamond=\bigcap_{W\supset\OO_\diamond}W$ and (2) the family $\W$ is causally separating for double cones, {\em i.e.} for every pair of spacelike separated double cones $\OO_\diamond,\widetilde{\OO}_\diamond$ there exists a wedge $W\in\W$, such that
\begin{equation*}
 \OO_\diamond\subset W\subset (\widetilde{\OO}_\diamond)'.
\end{equation*}
Since double cones form a base for the topology of $\Rl^4$, one can define a local net over arbitrary open and bounded spacetime regions $\OO$ via
\begin{equation*}
 \OO\longmapsto \Aa(\OO):=\overline{\bigcup_{\OO_\diamond\subset\OO}\Aa(\OO_\diamond)}^{\|\cdot\|},
\end{equation*}
where the superscript denotes the norm closure.

Wedge-local nets can by constructed by means of so-called {\em wedge triples} \cite{Lechner10} (see also \cite{BuchholzLechnerSummers:2010} for the closely related notion of a {\em causal Borchers triple}). A wedge triple $(\Aa_0,\Aa,\alpha)$, relative to the wedge $W_0$, consists of an inclusion of two C$^*$-algebra $\Aa_0\subset\Aa$ and a strongly continuous action $\alpha:\P_0\ra\Aut(\Aa)$, such that: $h\in \P_0$
\begin{enumerate}
 \item[$i)$] If $hW_0\subset W_0$, then $\alpha_h(\Aa_0)\subset \Aa_0$ 
 \item[$ii)$] If $hW_0\subset (W_0)'$, then $\alpha_h(\Aa_0)\subset (\Aa_0)'$,
\end{enumerate}
where the commutant is understood as the relative commutant in $\Aa$. Given a wedge-triple $(\Aa_0,\Aa,\alpha)$, then
\begin{equation}
\label{eq:WLNfromWT}
 W:=hW_0\longmapsto\alpha_h(\Aa_0)=:\Aa(W)
\end{equation}
defines a wedge-local net. Isotony holds by condition $i)$, locality by condition $ii)$ and covariance by the very definition of the net (\ref{eq:WLNfromWT}). Conversely, every wedge-local net $W\mapsto \Aa(W)$ gives rise to a wedge-triple: Define $\Aa_0:=\Aa(W_0)$ and let $\Aa$ be the quasilocal algebra of the theory. Then conditions $i)$ and $ii)$ hold by the isotony and locality of the net. Hence, a wedge-triple can be viewed as the basic building block of a quantum field theory and the task is to construct examples such that the resulting local net $\OO\mapsto\Aa(\OO)$ describes non-trivial interaction. 
We summarize the constructive problem in the following three steps:
\begin{enumerate}
 \item[I)] Specify a wedge triple $(\Aa_0,\Aa,\alpha)$.
 \item[II)] Construct a wedge-local net $W\mapsto\Aa(W)$ from $(\Aa_0,\Aa,\alpha)$.
 \item[III)] Prove that the double cone algebras (\ref{eq:LNfromWLN}) are non-trivial.
\end{enumerate}
Step II) is canonical in the case of Minkowski spacetime, since the Poincar\'e group acts, by definition, transitively on the set of wedges. However, in Chapters \ref{ch:Thermal} and \ref{ch:Cosmological} we will encounter situations where this is not the case anymore.
Step III) is the hardest task in this program. In two spacetime dimensions the split property for wedges and the modular nuclearity condition provide sufficient conditions for the non-triviality of intersections of wedge algebras \cite{BuchholzLechner:2004}. In four dimensions it is known that the scalar free field does not satisfy the split property for wedges \cite{Buchholz:1974}, so a condition which implies the non-triviality of intersections of wedge algebras is missing so far in this case. 
\\
\\
In this thesis we will be concerned exclusively with steps I) and II) of this program. One possibility to construct new examples of wedge triples from known ones ({\em e.g.} the one which is provided by a free field theory) is by means of a deformation procedure called {\em warped convolution} (see \cite{BuchholzSummers:2008, BuchholzLechnerSummers:2010} and \cite{GrosseLechner:2007, GrosseLechner:2008} for precursors of this work). The authors consider a C$^*$-dynamical system $(\A,\alpha,\Rl^n)$ which is covariantly represented on a Hilbert space $\HS$, {\em i.e.} the C$^*$-algebra $\A$ is a norm-closed $*$-subalgebras of some $\BH$ and $\alpha:\Rl^n\ra\Aut(\A)$ is a strongly continuous automorphic action which is implemented by a family of unitary operators $U:\Rl^n\ra\End(\HS)$.\footnote{Note that this is only a slight loss of generality since we can either work in the GNS representation of a translation-invariant state, or we work in the universal covariant representation which exists for every C$^*$-dynamical system \cite[Prop.7.4.7, Lem.7.4.9]{Pedersen:1979}. In the latter case, we assume that $\HS$ is separable, as it is the case in a variety of concrete models in quantum field theory.} The warped convolution $A_\theta$ of an operator $A\in\A$ is defined as
\begin{equation}
\label{eq:IntroWarpedConvolution}
 A_\theta:=(2\pi)^{-n}\int_{\Rl^n\times \Rl^n} dx\, dy\, e^{-ixy}U(\theta x)AU(y-\theta x),
\end{equation}
where $\theta$ is an antisymmetric $n\times n$ matrix, which plays the role of a deformation parameter. This integral can be defined in an oscillatory sense, if $A$ is an element of a certain ``smooth'' subalgebra $\A^\infty\subset\A$ (see Section \ref{sec:deformation} for details). For deformations of a single algebra, the mapping $A\to A_\theta$ has many features in common with deformation quantization and the Weyl-Moyal product, and in fact was recently shown \cite{BuchholzLechnerSummers:2010} to be equivalent to specific representations of Rieffel's deformed C$^*$-algebras with $\Rl^n$-action \cite{Rieffel:1992}. In applications to field theory models, however, one has to deform a whole family of algebras, corresponding to subsystems localized in spacetime, and the parameter $\theta$ has to be replaced by a family of matrices adapted to the geometry of the underlying spacetime. At this point the formulation of wedge-local nets in terms of wedge triples is convenient, since the deformation of the net amount to the deformation of the algebra $\Aa_0$. The form of the matrix $\theta$ has to be adjusted such that the deformed triple is again a wedge triple. In four spacetime dimensions the class of ``admissible'' matrices is of the form
\begin{equation*}
 \theta=
\begin{pmatrix}
0 & \kappa & 0 & 0\\
-\kappa & 0 & 0 & 0\\
0 & 0 & 0 & \kappa'\\
0 & 0 & -\kappa' & 0
\end{pmatrix},\quad
\kappa\ge 0,\, \kappa'\in\Rl.
\end{equation*}
Buchholz, Lechner and Summers proved the following key statement.
%---------------------------------------------------------------------------------------
\begin{theorem}{\bf (\cite[Thm.4.2]{BuchholzLechnerSummers:2010}).}\\
 Let $(\A_0,\A,\alpha)$ be a covariantly represented wedge triple such that $U\uhr\Rl^4$ satisfies the spectrum condition and $\theta$ an admissible matrix. Then the warped triple $((\Aa_0)_\theta,\BH,\alpha)$ with $(\Aa_0)_\theta:=\{A_\theta:A\in\Aa_0\cap \Aa^\infty\}''$ is also a wedge triple.
\end{theorem}
%---------------------------------------------------------------------------------------
\noindent
The proof of the locality condition $ii)$ relies on a subtle interplay between the special form of the matrix $\theta$, the relativistic spectrum condition and the geometric form of wedges.

For theories which describe massive particles, the two-particle scattering for the deformed theory has been computed in \cite{GrosseLechner:2007, GrosseLechner:2008} within the framework of \cite{BorchersBuchholzSchroer:2001} (see \cite{DybalskiTanimoto10} for a similar analysis in the massless case). It turns out that  the S-matrix changes under the deformation and that it is non-trivial even if it was trivial for the initial theory. The (improper) scattering states in the deformed theory ($\theta\neq 0$) and the (improper) scattering states in the undeformed theory ($\theta=0$) are related by
\begin{equation*}
 \ket{p,q}_\theta^{\text{in}}=e^{i|p\cdot\theta q|}\ket{p,q}_0^{\text{in}},\qquad
 \ket{p,q}_\theta^{\text{out}}=e^{-i|p\cdot\theta q|}\ket{p,q}_0^{\text{out}}
\end{equation*}
Since $\theta$ depends on $W_0$, the scattering states break Lorentz invariance. For the kernels of the elastic scattering matrices in the deformed and undeformed theory holds
\begin{equation*}
 {}^{\text{out}}_{\;\;\;\theta}\!\braket{p,q|p',q'}^{\text{in}}_\theta
=e^{i|p\cdot\theta q|+i|p'\cdot\theta q'|}\cdot {}^{\text{out}}_{\;\;\;0}\!\braket{p,q|p',q'}^{\text{in}}_0.
\end{equation*}
This shows that the deformed and undeformed theories are not unitarily equivalent and that also the deformed theories among each other are not equivalent. A similar result about unitary inequivalence in the generic case of deformations of causal Borchers triples was obtained in \cite[Prop.4.7]{BuchholzLechnerSummers:2010}.

The models which were constructed by these methods provide the first examples of fully Poincar\'e-covariant and wedge-local quantum field theories in four-dimensional Minkowski spacetime which describe non-trivial elastic scattering processes. However, from a physical point of view these models are too simple so far and should be considered as toy models, since (1) the S-matrix breaks Lorentz invariance, (2) the collision cross sections do not change under the deformation and (3) the associated double cone algebras are trivial \cite{BuchholzLechnerSummers:2010}. Nevertheless, the idea of deforming wedge triples sheds  new light on the constructive problem in quantum field theory and it is the hope that similar deformation methods lead to physically relevant theories in four dimensions and eventually to a better understanding of interacting quantum field theories. 
\\
\\
As we see, the warped convolutions deformation procedure heavily relies on the peculiar structure of Minkowski spacetime ({\em e.g.} Poincar\'e symmetry, spectrum condition, existence of wedges). These structures are not present on a generic curved spacetime: Since the isometry group is trivial, one cannot impose invariance of the state or covariance of the fields by referring to symmetries. There is no preferred choice of vacuum state and, in particular, there is no analogue of the spectrum condition. The notion of wedges relies on the presence of a Cartesian coordinate system and it is not obvious how to describe these regions in a covariant way. It is therefore certainly interesting and a priori not clear if a similar deformation procedure is also feasible on curved spacetimes.

%%%%%%%%%%%%%%%%%%%%%%%%%%%%%%%%%%%%%%%%%%%%%%%%%%%%%%%%%%%%
\section{Overview of the thesis}
%%%%%%%%%%%%%%%%%%%%%%%%%%%%%%%%%%%%%%%%%%%%%%%%%%%%%%%%%%%%
In this thesis we extend the construction of quantum field theories by warped convolution to a class of four-dimensional globally hyperbolic spacetimes with a sufficiently large isometry group, which gives rise to a representation of $\Rl^2$ as required by (\ref{eq:IntroWarpedConvolution}). The spacetimes that we consider are the following:
\begin{itemize}
 \item Chapter \ref{ch:Thermal}: Minkowski spacetime
 \item Chapter \ref{ch:Cosmological}: Cosmological spacetimes ({\em e.g.} Bianchi type models I-VII)
 \item Chapter \ref{ch:deSitter}: full de Sitter spacetime 
\end{itemize}

As a motivating example we study in Chapter \ref{ch:Thermal} a deformation of the scalar free field on Minkowski spacetime in the GNS representation of a KMS state (Araki-Woods representation). The spectrum condition is not fulfilled in this representation. Using a warped convolution along the edge of a wedge, we define deformed creation and annihilation operators in the thermal representation along the same lines as in \cite{GrosseLechner:2007}. The deformed fields turns out to be wedge-local and covariant with respect to the (extended) Euclidean group $\Euc:=\Rl^4\rtimes\SO(3)_0$. The associated nets of bounded operators for the deformed and undeformed theory are unitarily inequivalent.
This chapter contains several techniques and methods which will reappear in Chapter \ref{ch:Cosmological} about cosmological spacetimes: (1) we see that wedge-locality can be established without the spectrum condition by making an appropriate choice for the subgroup by which the deformation is defined (translations along the edge of a wedge); (2) the group $\Euc$ does not act transitively on the set of wedges $\W$, so a priori it is not clear how to generate a wedge-local net from a wedge triple (step II in Section \ref{sec:CAQFT}) $-$ we solve this problem by dividing $\W$ in certain ``coherent'' subfamilies and prescribing one initial algebra for each subfamily. 

In Chapter \ref{ch:Cosmological} we study quantum field theories on globally hyperbolic spacetimes which admit two spacelike and commuting Killing vector fields. This class contains various spacetimes which are of interest in cosmology. Using the flow of these Killing vector fields we generate two-dimensional submanifolds which play the role of edges. A wedge is then defined as a connected component of the causal complement of an edge. This construction is fully covariant and reduces to the usual notion in the case of flat spacetime. We show that these regions have inclusion and covariance properties which are analogous to the Minkowski space wedges. Next, we consider quantum field theories on these spacetimes and deform them by warped convolution using the flow of the Killing vector fields. The emerging models share many structural properties with deformed field theories on flat spacetime. In particular, they are localized in the wedges of the considered spacetime. Further aspects of deformed theories are discussed in the concrete example of the free Dirac field, where we prove that the deformed and undeformed theories are not unitarily equivalent.

In Chapter \ref{ch:deSitter} we consider a spacetime which does not fit into the framework of Chapter \ref{ch:Cosmological}, namely, full de Sitter spacetime. A family of wedges is already available here \cite{BorchersBuchholz:1999}. We apply the warped convolution deformation to quantum field theories with global $\Uone$ gauge symmetry within the algebraic setting (field nets). For the deformation we use a combination of internal and external symmetries (boosts associated with a wedge). The resulting theories are wedge-local and de Sitter-covariant. Then we study the deformation of a particular class of field nets in more detail, namely, nets which arise from inclusions of $\CAR$-algebras $-$ among this class is the free charged Dirac field. For these theories we determine the fixed-points of the deformation map and prove inequivalence of the deformed and undeformed models. Finally, we comment on warped convolutions in terms of purely internal or external symmetries using other Abelian subgroups of the de Sitter and gauge group.

In Chapter \ref{ch:outlook} we summarize our results and indicate open problems and perspectives for future work.

In Appendix \ref{ch:NotationConventions} we collect our conventions and a table of frequently used symbols.
\\
\\
The results in Chapter \ref{ch:Cosmological} are based on a joint paper \cite{DappiaggiLechnerMorfa-Morales10} with C. Dappiaggi and G. Lechner. The content of Chapter \ref{ch:deSitter} was published in \cite{Morfa-Morales:2011}.

%%%%%%%%%%%%%%%%%%%%%%%%%%%%%%%%%%%%%%%%%%%%%%%%%%%%%%%%%%%%
\chapter{Thermal Quantum Field Theories on Flat Spacetime}
\label{ch:Thermal}
%%%%%%%%%%%%%%%%%%%%%%%%%%%%%%%%%%%%%%%%%%%%%%%%%%%%%%%%%%%%
This chapter contains a simple but instructive example for a deformation of a quantum field theory on flat spacetime $-$ many of the structures and techniques which we encounter here will reappear in later chapters about deformations of quantum field theories on curved spacetimes. More specifically, we study a deformation of the scalar free field in the representation which is associated with a KMS state (Araki Woods representation). The spectrum condition is not fulfilled in this representation. Physically this can be traced back to the fact that an arbitrary amount of energy can be extracted from the ambient heat bath, so that the thermal Hamiltonian is not bounded from below. As positivity of energy was crucial in the wedge-locality proof for the deformed fields in \cite{GrosseLechner:2007}, it is not obvious if a similar deformation is also feasible in the thermal setting.

Despite this difference, we show that a deformation of the thermal scalar free field can be defined, so that the resulting theory is wedge-local and inequivalent to the initial one. 
In contrast to \cite{GrosseLechner:2007}, the deformation is not defined with the entire subgroup of translations of the Poincar\'e group, but rather by translations along the edge of a wedge.

This chapter is structured as follows. In Section \ref{sec:VacRepSFF} we collect the basic properties of the (abstract) Weyl algebra over a real symplectic vector space and outline its construction for the scalar free field together with its vacuum representation on Fock space. In Section \ref{sec:ThermalRepSFF} we consider KMS states on the Weyl algebra for the scalar free field, describe their GNS representations (Araki-Woods representation) and discuss the properties of the associated unbounded thermal field operators. In Section \ref{sec:Deformationsofthermalscalarfreefields} we introduce deformed thermal field operators by means of a deformation of the thermal creation and annihilation operators. We study the Wightman properties of these fields and it turns out that they are wedge-local and covariant with respect to the (extended) Euclidean group $\Euc=\Rl^4\rtimes\SO(3)_0$. Then, the wedge-local net of bounded operators which is generated by the deformed thermal fields is defined and it is proven that this net is unitarily inequivalent to the initial thermal net. 
\\
\\
In the presentation of the scalar free field and its thermal representation we follow \cite{Jakel:2004}.

\newpage
%===========================================================
\section{Minkowski spacetime and wedges}
\label{sec:MinkowskiSpacetime}
%===========================================================
Minkowski spacetime is the real manifold $\Rl^4=\Rl\times\Rl^3$ equipped with the metric tensor
\begin{equation*}
 \gST=dx^0\otimes dx^0-\sum_{k=1}^3dx^k\otimes dx^k,
\end{equation*}
where $\{x^\mu:\mu=0,\dots,3\}$ is the global chart, which is given by the Cartesian coordinates on $\Rl^4$. It has the structure of a Lorentzian manifold and we fix a time orientation, such that $e_0=(1,0,0,0)$ is future-directed. As $T_x(\Rl^4)\cong \Rl^4$ for all $x\in\Rl^4$, we identify these spaces. For $x,y\in\Rl^4$ there holds
\begin{equation*}
 \gST(x^\mu\del_\mu,y^\mu\del_\nu)=x^0y^0-\sum_{k=1}^3 x^ky^k=x^T\eta y=:x\cdot y,\qquad
\eta=\diag(1,-1,-1,-1),
\end{equation*}
where we summed over repeated indices and $xy$ denotes the Euclidean inner product of $x,y\in\Rl^4$. We use the shorthand notation $(x)^2:=x\cdot x$. In the following, we refer to the metric $\gST$ and its matrix representation $\eta$ simply as Minkowski metric and we write $(\Rl^4,\eta)$ to denote Minkowski spacetime. The Minkowski metric induces a causal structure on $\Rl^4$: we call points $x,y\in\Rl^4$ timelike, spacelike or null related, if $(x-y)^2>0$, $(x-y)^2<0$ or $(x-y)^2=0$, respectively. The causal complement $\OO'$ of a set $\OO\subset\Rl^4$ is defined as the interior of $\{x\in\Rl^4:(x-y)^2<0,y\in\overline{\OO}\}$, so that we always work with open sets. Two spacetime regions $\OO_1,\OO_2\subset\Rl^4$ are called spacelike separated if $\OO_1\subset\OO_2\,'$.

The isometry group of Minkowski space is the Poincar\'e group $\P=\Rl^4\rtimes \L$, where
\begin{equation*}
 \L=\mathrm{O}(1,3)=\{\Lambda\in\Mat(4,\Rl):\Lambda^T\eta\Lambda=\eta\}
\end{equation*}
is the Lorentz group. The semidirect product is defined with the natural action of $\L$ on $\Rl^4$, so that the group operation is the following:
\begin{equation*}
 (y,\Lambda)(y',\Lambda')=(\Lambda y'+y,\Lambda\Lambda'),\quad 
(y,\Lambda),(y',\Lambda')\in\P.
\end{equation*}
The Poincar\'e group is a ten-dimensional, non-compact, non-connected and real Lie group which has four connected components. The component which contains the identity (proper orthochronous Poincar\'e group) is denoted $\P_0:=\Rl^4\rtimes\L_0$, where $\L_0:=\SO(1,3)_0$ consists of those Lorentz transformations which preserve the orientation and time-orientation of $(\Rl^4,\eta)$.
\\
\\
Inside Minkowski spacetime there are special regions called wedges, which will play an important role in the following $-$ they turn out to be the typical localization regions of the deformed quantum fields from Section \ref{sec:Deformationsofthermalscalarfreefields}. They can be visualized as regions which are bounded by two non-parallel characteristic planes.\label{wedgeMinkowski} More precisely, we specify a reference wedge
\begin{equation*}
 W_0:=\{x\in\Rl^4:x^1>|x^0|\}
\end{equation*}
and define the family of all wedges $\W$ as the set of all Poincar\'e transforms of $W_0$:\label{familywedgesMinkowski}
\begin{equation*}
 \W:=\{(y,\Lambda)W_0:(y,\Lambda)\in\P_0\},\quad (y,\Lambda)W_0:=\Lambda W_0+y.
\end{equation*}
The action of the Lorentz group on a region is understood as the pointwise action.
By definition, $\P_0$ acts transitively on $\W$. Each wedge $W\in\W$ has an attached edge\label{edgeMinkoski} $E_W$, which is a two-dimensional plane. We have $E_{W_0}=\{x\in\Rl^4:x^0=x^1=0\}$ and $E_W=(y,\Lambda)E_{W_0}$ for $W=(y,\Lambda)W_0$. The wedge $W$ coincides with a connected component of the causal complement of $E_{W}$. 

The properties of these regions and their potential use in quantum field theory were investigated in \cite{ThomasWichmann:1997}. The authors show that each wedge $W\in\W$ is causally complete, {\it i.e.} $W''=(W')'=W$, and that the family $\W$ is causally separating for double cones in the following sense: for every pair of spacelike separated double cones $\OO_1,\OO_2$ there exists a wedge $W\in\W$, such that
\begin{equation}
\label{eq:Wedges_CausallySeparating}
 \OO_1\subset W\subset \OO_2\,'.
\end{equation}
Moreover, every double cone $\OO$ can be written as a suitable intersection of wedges:
\begin{equation}
\label{eq:Wedges_DoubleCones}
\OO=\bigcap_{W\supset \OO}W. 
\end{equation}

%-----------------------------------------------------------
\begin{remark}
Equations (\ref{eq:Wedges_CausallySeparating}) and (\ref{eq:Wedges_DoubleCones}) are the  two key properties of wedges which enable us to define a local net in terms of a wedge-local net (see Section \ref{sec:CAQFT}).
\end{remark}
%-----------------------------------------------------------

In the subsequent discussion of the scalar free field we will see that Lorentz transformations are broken in the thermal representation. The thermal fields transform covariantly only with respect to the extended Euclidean group $\Euc$. Clearly, this group does not act transitively on the set $\W$. For the definition of the deformed thermal net in Section \ref{sec:Deformationsofthermalscalarfreefields} it is convenient to decompose $\W$ into certain ``coherent'' subfamilies, namely, subsets where $\Euc$ does act transitively. This decomposition is defined in the following way. First, we split $\W$ into a Lorentz part and a translational part:
\begin{equation*}
 \W=\{W+y:W\in\W_\L,\,y\in\Rl^4\},\quad \W_\L:=\{\Lambda W_0:\Lambda\in\L_0\}.
\end{equation*}
Then we define an equivalence relation on $\L_0$: Two Lorentz transformations $\Lambda,\Lambda'\in\L_0$ are called equivalent, if there exists an $R\in\SO(3)_0$ such that $\Lambda=R\Lambda'$. One easily verifies that this defines indeed an equivalence relation. The set $\W_\L$ decomposes into subsets
\begin{equation*}
 \W_\L=\bigcup_{[\Lambda]}\W_{[\Lambda]},\quad
\W_{[\Lambda]}=\{\Lambda'W_0:\Lambda'\sim \Lambda\}
=\{R\Lambda W_0:R\in\SO(3)_0\}.
\end{equation*}
By definition, $\Euc$ acts transitively on each $\{W+y:W\in \W_{[\Lambda]},\,y\in\Rl^4\}$.
Next we make use of the fact that any Lorentz transformation $\Lambda\in\L_0$ can be written as
\begin{equation*}
 \Lambda=\Lambda_{\ee}\Lambda_1(s)\Lambda_{\ee'},
\end{equation*}
where $\Lambda_1(s)$, $s\in\Rl$ is a boost in the $x^1$-direction and $\Lambda_{\ee},\Lambda_{\ee'}$ are spatial rotations around axes $\ee,\ee'\in\Rl^3$ and angles $|\ee|$, $|\ee'|\in\Rl$. Hence, every $\W_{[\Lambda]}$ can be parametrized by a number $s\in\Rl$ and a single vector $\ee\in\Rl^3$, {\it i.e.}
\begin{equation*}
 \W_{[s,\ee]}:=\W_{[\Lambda]}=\{R\Lambda_\ee(s) W_0:R\in\SO(3)_0\},\quad \Lambda_\ee(s):=\Lambda_1(s)\Lambda_\ee.
\end{equation*}
From this follows that every wedge $W\in\W$ can be written as $W=R\Lambda_\ee(s) W_0+y$, for suitable $R\in\SO(3)_0$, $\ee\in\Rl^3$, $s\in\Rl$ and $y\in\Rl^4$.

%-----------------------------------------------------------
\begin{remark}
 In Chapter \ref{ch:Cosmological} about cosmological spacetimes a similar decomposition of wedges into coherent subfamilies will be used.
\end{remark}
%-----------------------------------------------------------

%===========================================================
\section{Vacuum representation of the scalar free field }
%===========================================================
\label{sec:VacRepSFF}
%===========================================================
\subsection{The Weyl algebra for the scalar free field}
%===========================================================
Before we come to the construction of the Weyl algebra for the scalar free field, we collect the basic properties of the (abstract) Weyl algebra over a real symplectic space.
\label{realsymplecticvectorspaceWeylAlgebra}
%-----------------------------------------------------------
\begin{definition}
\label{abstractweylalgebra}
 Let $\hhs$ be a real vector space and $\sigma$ a non-degenerate symplectic form on $\hhs$.  The \emph{Weyl algebra} $\CCR(\hhs,\sigma)$ is the C$^*$-algebra which is generated by symbols $V(f),f\in \hhs$ satisfying
\begin{equation}
\label{eq:WeylCCR}
 V(f)^*=V(-f),\quad 
 V(f)V(g)=e^{-\frac{i}{2}\sigma(f,g)}V(f+g)
\end{equation}
for all $f,g\in \hhs$.
\end{definition}
%-----------------------------------------------------------
%-----------------------------------------------------------
\begin{proposition}[\cite{BratteliRobinson:1997}]
	$\CCR(\hhs,\sigma)$ is unique up to isomorphism and simple, so all its representations are faithful or trivial. Moreover, there holds
	\begin{itemize}
	\item[i)] $V(0)=\unit$.
	\item[ii)] Each $V(f)$, $f\in \hhs$ is unitary.
	\item[iii)] $\|V(f)-\unit\|=2$, if and only if $f\neq 0$.
	\end{itemize}
\end{proposition}
%-----------------------------------------------------------
From part $iii)$ follows that the one-parameter groups $\Rl\ni\lambda\mapsto V(\lambda f)$, $f\in\hhs$ are never norm-continuous. Hence, there do not exist non-trivial strongly continuous actions of $\Rl^n$ on $\CCR(\hhs,\sigma)$. This is essentially the reason why we do not consider deformations of the Weyl algebra itself, but rather of generators of Weyl unitaries (field operators) later on.
For the purpose of defining field operators from Weyl operators it is useful to consider the following class of representations.
%-----------------------------------------------------------
\begin{definition}
	A representation $\pi:\CCR(\hhs,\sigma)\ra\cal{B}(\HS_\pi)$ is called \emph{regular}, if the unitary groups $\Rl\ni\lambda\mapsto \pi(V(\lambda f))$ are strongly continuous for all $f\in\hhs$.
\end{definition}
%-----------------------------------------------------------
From Stone's theorem about strongly continuous unitary groups follows that for these representations there exist unique selfadjoint operators $\phi_\pi(f)$ on $\HS_\pi$, such that $\pi(V(\lambda f))=\exp(i\lambda \phi_\pi(f))$ holds for all $\lambda \in\Rl$, $f\in \hhs$. As $\lambda \mapsto \pi(V(\lambda f))$ is not norm-continuous, this generator is unbounded. The action of the field operators is defined by differentiation of Weyl operators: 
\begin{equation}
\label{eq:SFFregularrep}
	\phi_\pi(f)\Psi
:=\frac{1}{i}\frac{d}{d\lambda }\pi(V(\lambda f))\Psi\br|_{\lambda =0}
\end{equation}
with $\Psi\in D(\phi_\pi(f)):=\{\Phi\in\HS_\pi:\lim_{\lambda \ra 0}[\pi(V(\lambda f))\Phi-\Phi]/\lambda \text{ exists}\}$. These operators satisfy the canonical commutation relation 
\begin{equation*}
 [\phi_\pi(f),\phi_\pi(g)]\Psi=i\sigma(f,g)\Psi,\quad \Psi\in D(\phi_\pi(f))\cap D(\phi_\pi(g))
\end{equation*}
by equation (\ref{eq:WeylCCR}) and since $\pi(V(f))D(\phi_\pi(g))=D(\phi_\pi(g))$ holds for all $f,g\in\hhs$ as the representation $\pi$ is regular \cite[Prop.5.2.4]{BratteliRobinson:1997}.
\\
\\
We proceed with the construction of $(\hhs,\sigma)$ for the (neutral massive) scalar free field.  
This construction is based on Wigner's classification of irreducible unitary representations of the Poincar\'e group \cite{Wigner:1939}. Associated with massive ($m>0$) and spinless ($s=0$) particles is the Hilbert space\label{oneparticlespaceSFF} $\HS_1:=L^2(H_m^+,d\mu_m)$, where $d\mu_m=d^3\pp/2\veps_\pp$\label{lorentzinvariantmeasure} is the (up to a constant) unique Lorentz-invariant measure on the upper mass hyperboloid\label{masshyperboloid} $H_m^+=\{p\in\Rl^4:p^2=m^2,\,p_0>0\}$ and $\veps_\pp=\sqrt{\pp^2+m^2}$ is the energy of a particle with momentum $\pp\in\Rl^3$. In the following, we use the letters $p,q$ for on-shell momenta and $\pp,\qq$ for their spatial components. 

The Hilbert space $\HS_1$ carries a unitary strongly continuous positive energy representation $U_1$ of $\P_0$, which is given by
\begin{equation}
\label{eq:SFFOneparticlerep}
 (U_1(y,\Lambda)\vphi)(p)=e^{ip\cdot y}\vphi(\Lambda^{-1}p),
\quad (y,\Lambda)\in\P_0,
\;\vphi\in \HS_1.
\end{equation}
An element\footnote{or more precisely a ray, {\it i.e.} an equivalence class of vectors in the Hilbert space, where elements are identified if they differ by a phase.} in the representation space is physically interpreted as the state vector (wave function) of a single massive spinless and non-interacting particle which satisfies the dispersion relation $E(\pp)=\sqrt{\pp^2+m^2}$.

Consider now the space $\Ss(\Rl^4,\Rl)$ of real-valued Schwartz functions on Minkowski space. Let $f\in\Ss(\Rl^4,\Rl)$ and define $\tilde{f}_+:=\tilde{f}\uhr H_m^+$. For $m>0$ the range of $f\mapsto\tilde{f}_+$ is contained in $\HS_1$, so we have a well-defined map from test functions to one-particle vectors. By direct computation one verifies that
\begin{equation*}
(\widetilde{f_{(y,\Lambda)}})_+=U_1(y,\Lambda)\tilde{f}_+,\quad   f_{(y,\Lambda)}(x):=f(\Lambda^{-1}(x-y)).
\end{equation*} 
For the construction of the Weyl algebra for the scalar free field we equip $\Ss(\Rl^4,\Rl)$ with the following symmetric positive semidefinite bilinear form
\begin{equation}
\label{eq:SFFscalarproduct}
 \left<f,g\right>_m
:=\int d\mu_m(p)\overline{\tilde{f}(p)}\tilde{g}(p)
=\int d\mu_m(p)\tilde{f}(-p)\tilde{g}(p),\quad f,g\in \Ss(\Rl^4,\Rl).
\end{equation}
Its kernel $\ker_m:=\{f\in\Ss(\Rl^4,\Rl):\|f\|_m^2:=\left<f,f\right>_m=0\}$\label{massshellnorm} consists of those Schwartz functions whose Fourier transform vanishes on $H_m^+$. We write $\hhs_m:=\Ss(\Rl^4,\Rl)/\ker_m$ for the corresponding quotient. The set $\hhs_m$ has the structure of a real vector space and 
\begin{equation*}
    \sigma_m([f],[g]):=\im\left<f,g\right>_m
\end{equation*}
defines a non-degenerate symplectic form on $\hhs_m$, where $f,g$ are representatives of the equivalence classes $[f],[g]\in\hhs_m$, respectively. Note that the inner product (\ref{eq:SFFscalarproduct}) does not depend on the chosen representative.
%-----------------------------------------------------------
\begin{remark}
	For notational simplicity we will write $f\in\hhs_m$ in the following, where $f$ is a representative of the equivalence class $[f]=f+\ker_m$. 
\end{remark}
%-----------------------------------------------------------
The Weyl algebra of the scalar free field is defined as $\Aa_m:=\CCR(\hhs_m,\sigma_m)$.
\label{weylalgebraSFF} Local algebras are obtained in the obvious manner: Let $\OO\subset \Rl^4$ be an open and bounded region in Minkowski space and define
\begin{equation*}
 \Aa_m(\OO):=\CCR(\hhs_m(\OO),\sigma_m),\quad
 \hhs_m(\OO):=\{f\in\hhs_m:\supp(f)\subset\OO\}.
\end{equation*}
From standard results \cite[Prop.5.2.10]{BratteliRobinson:1997} it follows that the map $\OO\mapsto\Aa_m(\OO)$ is a Poincar\'e-covariant net of C$^*$-algebras with respect to the natural action
\begin{equation}
\label{eq:SFFPoincareaction}
 \alpha:\P_0\ra\Aut(\Aa_m),\quad 
\alpha_{(y,\Lambda)}(V(f)):=V(f_{(y,\Lambda)}).
\end{equation}
Moreover, it is also local, {\it i.e.} $\Aa_m(\OO_1)\subset\Aa_m(\OO_2)'$ if $\OO_1\subset\OO_2\,'$ (see \cite[Lem.I.2.2]{Borchers:1996}). 
\\
\\
For the definition of creation and annihilation operators (from Weyl operators) it is convenient to introduce a complex structure $J$ on $\Ss(\Rl^4,\Rl)$. It is essentially given by multiplication with $i$ but also distinguishes between the positive and negative frequency parts of the Fourier transform. Let $h\in C^\infty(\Rl,\Rl)$ be an odd function, satisfying
\begin{equation}
\label{eq:SFF_Complex_Structure_h}
 h(\lambda)=
\begin{cases}
 1 &\text{if }\lambda\ge m\\
-1 &\text{if }\lambda\le -m
\end{cases}
\end{equation}
and define 
\begin{equation}
\label{eq:SFF_Complex_Structure}
 (\widetilde{Jf})(p_0,\pp):=ih(p_0)\tilde{f}(p_0,\pp),\quad f\in\Ss(\Rl^4,\Rl).
\end{equation}
Clearly, $J\Ss(\Rl^4,\Rl)\subset \Ss(\Rl^4,\Rl)$ and the operator $J$ descends to a well-defined map on the quotient $\hhs_m$ with $(\widetilde{Jf})_+(\veps_\pp,\pp)=i\tilde{f}_+(\veps_\pp,\pp)$, so that $J^2=-1$ on $\hhs_m$. By using the properties (\ref{eq:SFF_Complex_Structure_h}) of the function $h$ one can show that the map $J$ is 
compatible with $\sigma_m$ in the sense that
\begin{equation*}
 \sigma_m(f,Jg)=-\sigma_m(Jf,g).
\end{equation*}
holds for all $f,g\in\hhs_m$. Hence, the pair $(\hhs_m,\sigma_m)$ together with $J$ has the structure of a complex pre-Hilbert space.\label{realsymplecticvectorspaceWeylAlgebraSFF}

%===========================================================
\subsection{Fock representation of the scalar free field}
%===========================================================
Consider a state $\omega:\A\ra \Cl$ on a C$^*$-algebra $\A$ together with an action $\alpha:\P_0\ra\Aut(\A)$ of $\P_0$ by automorphisms on $\A$, such that $(y,\Lambda)\mapsto \omega(A\alpha_{(y,\Lambda)}(B))$ is continuous for all $A,B\in\A$. We write $(\HS_\omega,\pi_\omega,\Omega_\omega)$ for the GNS-triple which is associated with $(\A,\omega)$.
%-----------------------------------------------------------
\begin{definition}
The state $\omega$ is called \emph{Poincar\'e-invariant}, if $\omega\circ\alpha_{(y,\Lambda)}=\omega$ holds for all $(y,\Lambda)\in\P_0$.
\end{definition}
%-----------------------------------------------------------
From the invariance of the state it follows that the action is implementable $-$ the unitaries are given by $U(y,\Lambda)\pi_\omega(A)\Omega_\omega:=\pi(\alpha_{y,\Lambda}(A))\Omega$, $A\in\A$. From the continuity of $\alpha$ follows that $\{U(y,\Lambda):(y,\Lambda)\in\P_0\}$ is a strongly continuous unitary group. For the unitaries which implement spacetime translations we write $U(y):=U(y,1)$. 
%-----------------------------------------------------------
\begin{definition}
\label{vacuumstate}
 The state $\omega$ is called \emph{vacuum state}, if it is Poincar\'e-invariant and if the joint spectrum of the generators $\{P_\mu:\mu=0,\dots,3\}$ of the spacetime translations $U(y)=\exp(iy\cdot P)$ is contained in the closed forward lightcone $\overline{V^+}$ (spectrum condition).
\end{definition}
%-----------------------------------------------------------

%-----------------------------------------------------------
\begin{proposition}[\cite{Borchers:1996}]
\label{prop:SFFvacuumstate}
Consider the Weyl algebra for the scalar free field $\Aa_m$. The map $\omega_0:\Aa_m\ra\Cl$, defined as
\begin{equation*}
 \omega_0(V(f)):=\exp\Bl(-\frac{1}{4}\|f\|^2_m\Br),\quad f\in\hhs_m
\end{equation*}
and extended to all of $\Aa_m$ by linearity and continuity, is a vacuum state. 
\end{proposition}
%-----------------------------------------------------------

In the following we refer to the GNS-triple $(\HS_{\omega_0},\pi_{\omega_0},\Omega_{\omega_0})$ as \emph{vacuum representation} of the scalar free field.  
A concrete realization is given by the Fock representation which we describe in the following (see Proposition \ref{prop:Fockrep}). Let
\begin{equation}
\label{eq:SFFvacuumFockspace}
 \F(\HS_1):=\bigoplus_{n\ge 0}\HS_n,\quad \HS_n:=(\HS_1)^{\otimes_sn},\quad \HS_0:=\Cl
\end{equation}
be the Bosonic Fock space\label{bosefockspace} over the one-particle space $\HS_1$. We write $\Omega:=(1,0,0,\dots)$ for the Fock vacuum\label{fockvacuumvacuum}. The representation (\ref{eq:SFFOneparticlerep}) naturally extends to $\F(\HS_1)$ via second quantization: $(y,\Lambda)\in\P_0$, $\Psi\in\F(\HS_1)$
\begin{equation}
\label{eq:SFFPoincareRep}
 (U(y,\Lambda)\Psi)_n(p_1,\dots,p_n)
=e^{i\sum_{k=1}^np_k\cdot y}\Psi_n(\Lambda^{-1}p_1,\dots,\Lambda^{-1}p_n),\quad n\in\Nl
\end{equation}
and $U(y,\Lambda)\Omega=\Omega$. The map $y\mapsto U(y)$ is a strongly continuous unitary group and by Stone's theorem there exist selfadjoint operators $\{P_\mu:\mu=0,\dots,3\}$, such that $U(y)=\exp(iy\cdot P)$. The joint spectrum of this family of commuting and selfadjoint operators is contained in the closed forward light cone.

Elements in $\F(\HS_1)$ are sequences $\Psi=\{\Psi_n\}_{n=0}^\infty$, where $\Psi_n\in\HS_n$ is an $n$-particle vector. For the dense subspace of vectors with only finitely many components unequal to zero we write $\F_0$. On this subspace we define  annihilation operators $a(\vphi)$, $\vphi\in\HS_1$ and creation operators $a^\dagger(\vphi)$, $\vphi\in\HS_1$, in the following way
\begin{equation*}
\label{vacuumcreationannihilationoperators}
 (a(\vphi)\Psi)_n(p_1,\dots,p_n):=
\sqrt{n+1}\int d\mu_m(p)\overline{\vphi(p)}\Psi_{n+1}(p,p_1,\dots,p_n),\qquad
a^\dagger(\vphi):=a(\vphi)^*.
\end{equation*}
There holds
\begin{equation*}
 (a^\dagger(\vphi)\Psi)_n(p_1,\dots,p_n)=
\begin{cases}
0 & \text{if } n=0\\
\frac{1}{\sqrt{n}}\sum_{k=1}^n\vphi(p_k)\Psi_{n-1}(p_1,\dots,\widehat{p_k},\dots,p_n) & \text{if } n>0,
\end{cases}
\end{equation*}
where $\widehat{p_k}$ denotes omission of $p_k$. Note that $\vphi\mapsto a^\dagger(\vphi)$ is complex linear, while $\vphi\mapsto a(\vphi)$ is anti-linear. The creation and annihilation operators $a^\#(\vphi)$ satisfy the usual canonical commutation relations on $\F_0$: $\vphi,\psi\in\HS_1$
\begin{equation*}
 [a(\vphi),a(\psi)]=[a^\dagger(\vphi),a^\dagger(\psi)]=0,\quad
[a(\vphi),a^\dagger(\psi)]
=\left<\vphi,\psi\right>_m\cdot 1.
\end{equation*}
Under the representation (\ref{eq:SFFPoincareRep}) of the Poincar\'e group they transform according to
\begin{equation*}
 U(y,\Lambda)a^\#(\vphi)U(y,\Lambda)^{-1}=a^\#(\vphi_{(y,\Lambda)}),\quad
\vphi_{(y,\Lambda)}(p)=e^{ip\cdot y}\vphi(\Lambda^{-1}p).
\end{equation*}
Using the usual symbolic notation, we introduce the operator-valued distributions $a^\#(p)$, $p\in H_m^+$ which are associated with $a^\#(\vphi)$, $\vphi\in\HS_1$ via
\begin{equation*}
 a(\vphi)=\int d\mu_m(p)\overline{\vphi(p)}a(p),\qquad 
 a^\dagger(\vphi)=\int d\mu_m(p)\vphi(p)a^\dagger(p).
\end{equation*}
In this notation the CCR relations are
\begin{equation*}
 [a(p),a(q)]=[a^\dagger(p),a^\dagger(q)]=0,\qquad
 [a(p),a^\dagger(q)]=2\veps_\pp\delta(\pp-\qq)\cdot 1
\end{equation*}
and the $a^\#(p)$, $p\in H_m^+$ transform according to
\begin{equation*}
 U(y,\Lambda)a(p)U(y,\Lambda)^{-1}=e^{-i\Lambda p\cdot y}a(\Lambda p),\qquad
 U(y,\Lambda)a^\dagger(p)U(y,\Lambda)^{-1}=e^{i\Lambda p\cdot y}a^\dagger(\Lambda p).
\end{equation*}

%-----------------------------------------------------------
\begin{definition}
The scalar free field in the vacuum representation is defined as
\begin{equation*}
 \phi(f):=\frac{1}{\sqrt{2}}[a(\tilde{f}_+)+a^\dagger(\tilde{f}_+)]
,\quad f\in\Ss(\Rl^4,\Rl).
\end{equation*}
\end{definition}
%-----------------------------------------------------------

The following proposition collects the basic properties of these operators and shows that $f\to \phi(f)$ satisfies the Wightman axioms for a scalar field.

%-----------------------------------------------------------
\begin{proposition}[\cite{ReedSimon:1975}]
\label{prop:SFF_operatorproperties}
	Let $f\in\Ss(\Rl^4,\Rl)$. Then:
	\begin{itemize}
		\item[i)] The subspace $\F_0$ of vectors of finite particle number is contained in the domain of each $\phi(f)$ and stable under the action of these 		
		operators.
		\item[ii)] The map $f\to\phi(f)\Psi$, $\Psi\in\F_0$ is a vector-valued tempered distribution.
		\item[iii)] Each $\phi(f)$ is essentially self-adjoint.
	\end{itemize}
	Moreover:	
	\begin{itemize}	
		\item[iv)] Each $\phi(f)$ satisfies the Klein-Gordon equation: $\phi((\Box+m^2)f)=0$.
		\item[v)] Each $\phi(f)$ transforms covariantly under the representation (\ref{eq:SFFPoincareRep}) of $\P_0$:
		\begin{equation*}
		 U(y,\Lambda)\phi(f)U(y,\Lambda)^{-1}\Psi=\phi(f_{(y,\Lambda)})\Psi,\quad
		(y,\Lambda)\in\P_0,\;\Psi\in \F_0.
		\end{equation*}
		\item[vi)] For $\Psi\in\F_0$ and $f,g\in\Ss(\Rl^4,\Rl)$ there holds
		\begin{equation*}
			[\phi(f),\phi(g)]\Psi=i\im\left<f,g\right>_m \Psi.
		\end{equation*} 
	\end{itemize}
\end{proposition}
%-----------------------------------------------------------

We denote the selfadjoint closure of $\phi(f)$ by the same symbol and by Borel functional calculus $V_F(f):=\exp(i\phi(f))$ is a unitary on $\F(\HS_1)$. There holds $V_F(f)^*=V_F(-f)$ since $f$ is real and 
\begin{equation*}
 V_F(f)V_F(g)
=e^{-\frac{i}{2}\im\left<f,g\right>_m}V_F(f+g)
=e^{-i\im\left<f,g\right>_m}V_F(g)V_F(f),\quad f,g\in\hhs_m
\end{equation*}
by part $vi)$ of Proposition \ref{prop:SFF_operatorproperties}. Hence
\begin{equation}
\label{eq:SFFFockrep}
 \Aa_m\ni V(f)\longmapsto \pi_F(V(f)):= V_F(f)\in\cal{B}(\cal{F}(\HS_1))
\end{equation}
is a representation of the Weyl algebra $\Aa_m$. The Weyl operators $V_F(f)$ act on the Fock vacuum $\Omega$ according to (see \cite{Camassa:2003})
\begin{equation*}
 V_F(f)\Omega=e^{\frac{1}{4}\|f\|_m^2}e^{\tilde{f}_+},\quad e^{\tilde{f}_+}:=\bigoplus_{n=0}^\infty \frac{1}{\sqrt{n!}}(\tilde{f}_+)^{\otimes n},
\quad f\in\hhs_m,
\end{equation*}
so there holds $\left<\Omega,\pi_F(V(f))\Omega\right>_m=\omega_0(V(f))$. Furthermore, since $f\mapsto \tilde{f}_+$ has dense range (see \cite[Ch. X, Ex. 44]{ReedSimon:1975}), it follows that $\{e^{\tilde{f}_+}: f\in\hhs_m\}$ forms a total set in $\F(\HS_1)$, so $\pi_F(\Aa_m)\Omega$ is dense in $\F(\HS_1)$. The following proposition collects these findings.

%-----------------------------------------------------------
\begin{proposition}
\label{prop:Fockrep}
    The GNS representations $(\HS_{\omega_0},\pi_{\omega_0},\Omega_{\omega_0})$ and the Fock representation $(\F(\HS_1),\pi_F,\Omega)$ of the Weyl algebra for the scalar free field $\Aa_m$ are unitarily equivalent. Moreover, (\ref{eq:SFFPoincareRep}) implements the action (\ref{eq:SFFPoincareaction}) in this representation.
\end{proposition}
%-----------------------------------------------------------

%===========================================================
\section{Thermal representation of the scalar free field}
%===========================================================
\label{sec:ThermalRepSFF}
%===========================================================
\subsection{KMS states}
%===========================================================
Thermal equilibrium states are characterized by means of the KMS condition \cite{HaagHugenholtzWinnink:1967}. 
%-----------------------------------------------------------
\begin{definition}
\label{KMSstate}
	Let $\A$ be a C$^*$-algebra and $\tau:\Rl\ra\Aut(\A)$ a one-parameter group of automorphisms. A state $\omega_\beta:\A\ra\Cl$ is called $(\tau,\beta)$-KMS state at inverse temperature $\beta> 0$, if for all $A,B\in\A$ there exists a complex-valued function $F_{A,B}$ on $S(0,\beta):=\{z\in\Cl:0\le\im(z)\le\beta\}$, such that
	\begin{itemize}
	\item[$1)$] $F_{A,B}$ is continuous and bounded in $S(0,\beta)$
	\item[$2)$] $F_{A,B}$ is analytic in the interior of $S(0,\beta)$
	\item[$3)$] its boundary values are given by
	\begin{equation}
	\label{eq:KMScondition}
		F_{A,B}(t)=\omega_\beta(A\tau_t(B)),\quad
		F_{A,B}(t+i\beta)=\omega_\beta(\tau_t(B)A).
	\end{equation}
	\end{itemize}
\end{definition}
%-----------------------------------------------------------
The physical interpretation of KMS states as thermal equilibrium states is justified by their passivity and stability properties \cite{BratteliRobinson:1997}.
As the definition of temperature requires the introduction of a heat bath, the KMS condition distinguishes a rest frame, which is determined by the dynamics $\tau$ in the definition of the state. Hence it is expected that Lorentz transformations are broken, {\it i.e.} not implementable in the GNS representation of a KMS state. The following proposition lists the basic consequences of the KMS condition (\ref{eq:KMScondition}).
%-----------------------------------------------------------
\begin{proposition}[\cite{BratteliRobinson:1997}]
\label{prop:KMS_basicproperties}
	Let $\omega_\beta$ be a $(\tau,\beta)$-KMS state on $\A$. Then
	\begin{itemize}
	\item[$i)$] $\omega_\beta\circ\tau_t=\omega_\beta$ for all $t\in\Rl$.
	\item[$ii)$] $\Omega_{\omega_\beta}$ is (cyclic and) separating for $\pi_{\omega_\beta}(\A)''$.
	\end{itemize}
\end{proposition}
%-----------------------------------------------------------
The cyclicity of $\Omega_{\omega_\beta}$ holds simply by the GNS construction. That the vector $\Omega_{\omega_\beta}$ is also separating uses the KMS condition. From part $i)$ it follows that 
\begin{equation*}
 U_{\omega_\beta}(t)\pi_{\omega_\beta}(A)\Omega_{\omega_\beta}:=\pi_{\omega_\beta}(\tau_t(A))\Omega_{\omega_\beta},\quad A\in\A
\end{equation*}
defines a strongly continuous one-parameter unitary group $\{U_{\omega_\beta}(t):t\in\Rl\}$ which implements the automorphisms $\{\tau_t:t\in\Rl\}$ in the representation $(\HS_{\omega_\beta},\pi_{\omega_\beta},\Omega_{\omega_\beta})$. By Stone's theorem there exists a unique selfadjoint generator $L_\beta$, such that $U_{\omega_\beta}(t)=e^{itL_\beta}$, $t\in\Rl$. The operator $L_\beta$ is called \emph{Liouvillian} and it is the analogue of the Hamiltonian in the thermal case. Its spectrum is not bounded from below, symmetric \cite{Araki:1972} and typically\footnote{More precisely, $\spec(L_\beta)=\Rl$ if $\tau$ acts asymptotically Abelian and $\omega_\beta$ admits a decomposition into extremal $(\tau,\beta)$-KMS states.} the whole real line \cite{tenBrinkeWinnink:1976}.
From this we see that in thermal representations the spectrum condition is generically violated. Physically this can be traced back to the fact that an arbitrary amount of energy can be extracted from the ambient heat-bath.

%===========================================================
\subsection{The Araki-Wood representation}
%===========================================================
Consider the Weyl algebra of the scalar free field $\Aa_m$ together with a one-parameter group of automorphisms which represent time translations:
\begin{equation*}
\tau_t(V(f)):=V(f_t),\quad f_t(x^0,\xx)=f(x^0-t,\xx).
\end{equation*}
Time translations amounts to multiplication with the one-particle energy in momentum space, {\it i.e.} $(\tilde{f_t})_+(\veps_\pp,\pp)
=e^{it\veps_\pp}\tilde{f}_+(\veps_\pp,\pp)$. For the following definition it is convenient to define the corresponding multiplication operator: $(\vveps\tilde{f}_+)(\veps_\pp,\pp):=\veps_\pp\tilde{f}_+(\veps_\pp,\pp)$, $f\in\hhs_m$. 

%-----------------------------------------------------------
\begin{proposition}[\cite{ArakiWoods:1963}]
\label{prop:TFFKMSstate}
 Let $\rho:=(e^{\beta\vveps}-1)^{-1}$ and $\beta>0$. Then
\begin{equation}
\label{eq:Weyl_KMS_state}
 \omega_\beta(V(f)):=\exp\Bl(-\frac{1}{4}\left<f,[1+2\rho]f\right>_m\Br),\quad f\in\hhs_m,
\end{equation}
extended to all of $\Aa_m$ by linearity and continuity is a $(\tau,\beta)$-KMS state.
\end{proposition}
%-----------------------------------------------------------
\begin{remark}
	Since $\pp\mapsto [\exp(\beta\sqrt{\pp^2+m^2})-1]^{-1}$, with $\beta,m>0$, is a  bounded and continuous function which vanishes at infinity, there holds $\rho\vphi\in L^2(H_m^+,d\mu_m)$, if $\vphi\in L^2(H_m^+,d\mu_m)$. 
\end{remark}
%-----------------------------------------------------------
\begin{proposition}[\cite{ArakiWoods:1963}]
\label{propo:TFFAWrepresentation}
 The GNS-triple $(\HS_{\omega_\beta},\pi_{\omega_\beta},\Omega_{\omega_\beta})$ which is associated with $(\Aa_m,\omega_\beta)$ is unitarily equivalent to the (left) Araki-Woods (AW) triple $(\HS_\beta,\pi_\beta,\Omega_\beta)$, where
\begin{align*}
 \HS_\beta&=\,\F(\HS_1\oplus\overline{\HS_1})=\F(\HS_1)\otimes\overline{\F(\HS_1)}\\
 \pi_\beta(V(f))
&=\,V_F(\sqrt{1+\rho}\,f\oplus \overline{\sqrt{\rho}\,f})
=V_F(\sqrt{1+\rho}\,f)\otimes \overline{V_F(\sqrt{\rho}\,f)}\\
 \Omega_\beta&=\,\Omega\otimes\overline{\Omega}.
\end{align*}
Here $\F(\HS_1)$ is the Fock space (\ref{eq:SFFvacuumFockspace}) from the vacuum representation of the scalar free field and $\overline{\F(\HS_1)}$ is its\label{conjugatehilbertspace} conjugate\footnote{The conjugate $\overline{\HS}$ of a Hilbert space $\HS$ is as a set the  the same as $\HS$, but it is equipped with the algebraic operations $(\vphi,\psi)\mapsto\vphi+\psi$ (addition) and $(a,\vphi)\mapsto \bar{a}\vphi$ (scalar multiplication) together with the inner product $(\vphi,\psi)\mapsto\left<\vphi,\psi\right>_{\overline{\HS}}=\left<\psi,\vphi\right>_{\HS}$. Elements in $\overline{\HS}$ are denoted by $\overline{\vphi}$ and if $A$ is a linear operator on $\HS$, then we denote by $\bar{A}$ the linear operator on $\overline{\HS}$ which is defined as $\bar{A}\overline{\vphi}:=\overline{A\vphi}$.
}. The tensor product which appears in $\HS_\beta$ is the unsymmetrized tensor product of Hilbert spaces. The vector $\Omega\in\F(\HS_1)$ is the Fock vacuum and $V_F(f)$,$f\in\hhs_m$ is the Fock representation (\ref{eq:SFFFockrep}) of $\Aa_m$.
\label{fockvacuumthermal}
\end{proposition}
%-----------------------------------------------------------

Similar to the case of the charged scalar free field, the first tensor factor in $\HS_\beta$ corresponds to the usual multiparticle state vectors from the vacuum representation and the second tensor factor can be interpreted as holes (antiparticles) in the ambient heat bath. Note that $\rho\ra 0$ as $\beta\ra\infty$, so in the zero temperature limit the AW representation reduces to the vacuum representation (\ref{eq:SFFFockrep}).
\\\\
The state (\ref{eq:Weyl_KMS_state}) is invariant under translations. The automorphisms $\alpha_y(V(f))=V(f_y)$, $y\in\Rl^4$ are implemented by unitary operators
\begin{equation}
\label{eq:SFFthermaltranslations}
 U_\beta(y)=\exp(iy\cdot P_\beta),\quad (P_\beta)_\mu=P_\mu\otimes \bar{1}-1\otimes\overline{P_\mu},
\end{equation}
where $\{P_\mu:\mu=0,\dots,3\}$ are the energy momentum operators of the vacuum representation (\ref{eq:SFFFockrep}), so that
\begin{equation*}
 \pi_\beta(\alpha_y(V(f)))=U_\beta(y)\pi_\beta(V(f))U_\beta(y)^{-1},\quad y\in\Rl^4,\;f\in\hhs_m.
\end{equation*}
The generator of time translations $(P_\beta)_0=P_0\oplus 1-1\oplus\overline{P_0}$ is the Liouvillian in this model, where $P_0$ is the Hamiltonian from the vacuum representation. The spectrum of $(P_\beta)_0$ is the whole real line, since $\spec(P_0)=\{0\}\cup[m,\infty)$, $m>0$. From (\ref{eq:SFFthermaltranslations}) we see that translations in the thermal representation decompose into a product of translations in the vacuum representation:
\begin{equation*}
 U_\beta(y)=U(y)\otimes\overline{U(y)},\quad
U(y)=\exp(iy\cdot P).
\end{equation*}
Lorentz transformations are generically broken, {\it i.e.} not implementable in the GNS representation of KMS state \cite{Narhofer:1977, Ojima:1986}. More precisely, a Poincar\'e transformation is not implementable if it does not commute with the time translations. In our concrete case one can explicitly see that the state (\ref{eq:Weyl_KMS_state}) is invariant under $\SO(3)_0$ but not under boosts. Hence we consider a representation $U_\beta$ of the extended Euclidean group $\Euc$ on $\HS_\beta$, which is defined as
\begin{equation}
\label{eq:TFFEucRep}
 U_\beta(y,R):=U(y,R)\otimes\overline{U(y,R)},\quad (y,R)\in\Euc
\end{equation}
where $U$ is the representation (\ref{eq:SFFPoincareRep}) of the Poincar\'e group on the vacuum Fock space. As $U$ implements $\alpha\uhr\Euc$ in the vacuum representation, $U_\beta$ implements $\alpha\uhr\Euc$ in the thermal representation.

Since the AW representation is a regular representation we can introduce thermal field operators $\phi_\beta(f)$ by differentiating thermal Weyl operators $\pi_\beta(V(f))$ as in (\ref{eq:SFFregularrep}). There holds
\begin{equation}
\label{eq:TFF_definition}
 \phi_\beta(f)
=\phi(\sqrt{1+\rho}\,f)\otimes\bar{1}+1\otimes\overline{\phi(\sqrt{\rho}\,f)}
\end{equation}
by Proposition \ref{propo:TFFAWrepresentation}. Thermal creation and annihilation operators are defined as
\begin{equation*}
 a_\beta(\tilde{f}_+):=\frac{1}{\sqrt{2}}[\phi_\beta(f)+i\phi_\beta(J f)],\qquad
 a^\dagger_\beta(\tilde{f}_+):=\frac{1}{\sqrt{2}}[\phi_\beta(f)-i\phi_\beta(J f)],
\end{equation*}
where $J=i$ is the complex structure (\ref{eq:SFF_Complex_Structure}) on $\hhs_m$.
For $\vphi\in\HS_1$ we obtain
\begin{equation*}
\label{thermalcreationannihilationoperators}
 a_\beta(\vphi)
=a(\sqrt{1+\rho}\,\vphi)\otimes\bar{1}
+1\otimes \overline{a^\dagger(\sqrt{\rho}\,\vphi)},\qquad
a^\dagger_\beta(\vphi)
=a_\beta(\vphi)^*.
\end{equation*}
Note that $\vphi\mapsto a^\dagger(\vphi)$ is complex linear, while $\vphi\to a_\beta(\vphi)$ is anti-linear and
\begin{equation*}
 a_\beta(\vphi)\Omega_\beta=\Omega\otimes\overline{\sqrt{\rho}\,\vphi},\qquad
 a^\dagger_\beta(\vphi)\Omega_\beta=\sqrt{1+\rho}\,\vphi\otimes\bar{\Omega}
\end{equation*}
so the thermal annihilation operator does not annihilate the thermal vacuum, but rather creates a hole. The operators $a^\#_\beta(\vphi)$ satisfy the canonical commutation relations
\begin{equation}
\label{eq:TFF_CCR}
 [a_\beta(\vphi),a_\beta(\psi)]
=[a^\dagger_\beta(\vphi),a^\dagger_\beta(\psi)]=0,\qquad
[a_\beta(\vphi),a^\dagger_\beta(\psi)]=\left<\vphi,\psi\right>_m(1\otimes\bar{1}).
\end{equation}
and transform according to
\begin{equation}
\label{eq:TFF_Covariance_a}
 U_\beta(y,R)a^\#_\beta(\vphi)U_\beta(y,R)^{-1}=a^\#_\beta(\vphi_{(y,R)}),\qquad
\vphi_{(y,R)}(p)=e^{ip\cdot y}\vphi(R^{-1}p)
\end{equation}
under the representation (\ref{eq:TFFEucRep}) of $\Euc$. Again, we proceed to the operator-valued distributions $a^\#_\beta(p)$, $p\in H_m^+$:
\begin{equation*}
 a_\beta(\vphi)=\int d\mu_m(p)\overline{\vphi(p)}a_\beta(p),\qquad
 a^\dagger_\beta(\vphi)=\int d\mu_m(p)\vphi(p)a^\dagger_\beta(p),
\end{equation*}
with
\begin{equation*}
 a_\beta(p)=\sqrt{1+\rho_p}\,a(p)\otimes\bar{1}+1\otimes\sqrt{\rho_p}\,a^\dagger(p),\quad
 a^\dagger_\beta(p)=\sqrt{1+\rho_p}\,a^\dagger(p)\otimes\bar{1}+1\otimes\sqrt{\rho_p}\,a(p)
\end{equation*}
and $\rho_p:=[\exp(\beta\veps_\pp)-1]^{-1}$.
In this notation the CCR relations read
\begin{equation*}
 [a_\beta(p),a_\beta(q)]=[a^\dagger_\beta(p),a^\dagger_\beta(q)]=0,\qquad
[a_\beta(p),a^\dagger_\beta(q)]=2\veps_\pp\delta(\pp-\qq)(1\otimes \bar{1})
\end{equation*}
and the $a_\beta^\#(p),\,p\in H_m^+$ transform according to 
\begin{equation*}
 U_\beta(y,R)a_\beta(p)U_\beta(y,R)^{-1}=e^{-iRp\cdot y}a_\beta(Rp),\quad
 U_\beta(y,R)a^\dagger_\beta(p)U_\beta(y,R)^{-1}=e^{iRp\cdot y}a^\dagger_\beta(Rp)
\end{equation*}
under $\Euc$.

Since the scalar free field in the thermal representation 
\begin{equation*}
 \phi_\beta(f)
=\phi(\sqrt{1+\rho}\,f)\otimes\bar{1}+1\otimes\overline{\phi(\sqrt{\rho}\,f)}
=\frac{1}{\sqrt{2}}[a_\beta(\tilde{f}_+)+a^\dagger_\beta(\tilde{f}_+)]
\end{equation*}
is merely a superposition of two  scalar free fields in the vacuum representation, it is more or less straightforward to see that the Wightman properties from the vacuum case (see Proposition \ref{prop:SFF_operatorproperties}) carry over to the thermal case.
We give a brief outline of the proofs.
%-----------------------------------------------------------
\begin{proposition}
	\label{prop:TFF_operatorproperties}
		Let $f\in\Ss(\Rl^4,\Rl)$. Then:
		\begin{itemize}
			\item[i)] The subspace $\F_{\beta,0}$ of vectors of finite particle number is contained in the domain of each 
			$\phi_\beta(f)$ and stable under the action of these operators.
			\item[ii)] The map $f\to\phi_\beta(f)\Psi$, $\Psi\in\F_{\beta,0}$ is a vector-valued tempered distribution.
			\item[iii)] Each $\phi_\beta(f)$ is essentially self-adjoint.
		\end{itemize}
		Moreover:	
		\begin{itemize}		
			\item[iv)] Each $\phi_\beta(f)$ satisfies the Klein-Gordon equation: $\phi_\beta((\Box+m^2)f)=0$.
			\item[v)] Each $\phi_\beta(f)$ transforms covariantly under the representation (\ref{eq:TFFEucRep}) of $\Euc$:
			\begin{equation*}
			 U_\beta(y,R)\phi_\beta(f)U_\beta(y,R)^{-1}\Psi
			=\phi_\beta(f_{(y,R)})\Psi,\quad (y,R)\in\Euc,\;\Psi\in\F_{\beta,0}.
			\end{equation*}
			\item[vi)] For $\Psi\in\F_{\beta,0}$ and $f,g\in\Ss(\Rl^4,\Rl)$ there holds
			\begin{equation*}
				[\phi_\beta(f),\phi_\beta(g)]\Psi=i\im\left<f,g\right>_m\Psi.
			\end{equation*} 
		\end{itemize}	
\end{proposition}	
\begin{proof}
 $i)$ Since the multiplication operator $\rho$ does not change the particle number, the statements $\F_{\beta,0}\subset D(\phi_{\beta}(f))$ and $\phi_{\beta}(f)\F_{\beta,0}\subset \F_{\beta,0}$ follow from the corresponding properties of the vacuum representation.

$ii)$ Let $\Psi_{k,l}=\Psi_k\otimes\overline{\Psi_l}\in\HS_k\otimes\overline{\HS_l}$ with $k,l\in\Nl_0$. Using the basic bounds $\|a^\#(\vphi)\Psi_k\|\le \sqrt{k+1}\|\vphi\|\|\Psi_k\|$ for creation and annihilation operators in the vacuum representation, one obtains the following bounds for the creation and annihilation operators in the AW representation:
\begin{align}
 \|a^\#_{\beta}(\tilde{f}_+)\Psi_{k,l}\|
&\le \Bl[\sqrt{k+1}\|\sqrt{1+\rho}\tilde{f}_+\|+\sqrt{l+1}\|\sqrt{\rho}\tilde{f}_+\|\Br]\|\Psi_k\|\|\Psi_l\|\nonumber\\
\label{eq:TFF_bound_a}
&\le C_{\beta,m}\cdot \sqrt{\max(k,l)+1}\|\Psi_k\|\|\Psi_l\|
\Bl[\int d\mu_m(p)|\tilde{f}_+(p)|^2\Br]^{1/2}
\end{align}
where $C_{\beta,m}$ is a positive constant independent of $f$. As the right hand side depends continuously on $f$ in the Schwartz space topology, the assertion follows.

$iii)$ Using the bound (\ref{eq:TFF_bound_a}), one can show along the same lines as in \cite[Prop.5.2.3]{BratteliRobinson:1997} that each $\Psi\in\F_{\beta,0}$ is an entire analytic vector for $\phi_{\beta}(f)$. By Nelson's analytic vector theorem the statement follows.

$iv)$ For $g:=(\Box+m^2)f$, $f\in\Ss(\Rl^4,\Rl)$ holds $\tilde{g}_+(p)=(-\veps_\pp^2+\pp^2+m^2)\tilde{f}(\veps_\pp,\pp)=0$ and the statement follows from the linearity of $\phi_\beta(f)$.

$v)$ This follows from the covariance properties (\ref{eq:TFF_Covariance_a}) of the operators $a^\#_\beta(\vphi)$.

$vi)$ This follows from the CCR relations (\ref{eq:TFF_CCR}).
\end{proof}
%-----------------------------------------------------------

The field $\phi_\beta(f)$ yields the familiar two-point function from thermal field theory:
\begin{align*}
 w_\beta^{(2)}(f,g)
:=&\left<\Omega_\beta,\phi_\beta(f)\phi_\beta(g)\Omega_\beta\right>\\
=&\left<f,(1+\rho)g\right>+\left<g,\rho f\right>\\
=&\int d\mu_m(p)\Bl[\frac{\overline{\tilde{f}_+(p)}\tilde{g}_+(p)}{1-e^{-\beta\veps_\pp}}+\frac{\overline{\tilde{g}_+(p)}\tilde{f}_+(p)}{e^{\beta\veps_\pp}-1}\Br]\\
=&\int dx\, f(x)\int dy\, g(y)\int d\mu_m(p)
\Bl[\frac{e^{-ip\cdot(x-y)}}{1-e^{-\beta\veps_\pp}}+\frac{e^{ip\cdot(x-y)}}{e^{\beta\veps_\pp}-1}\Br].
\end{align*}
We will write $w_{\beta}^{(n)}(x_1,\dots,x_n)$ for the distributional kernels of the higher $n$-point functions $w_\beta^{(n)}(f_1,\dots,f_n)$. For the two-point function we find in momentum space
\begin{equation*}
 \tilde{w}^{(2)}_\beta(p,q)
=\frac{(2\pi)^4\delta(p+q)}{2\veps_\pp}
\Bl[\frac{\delta(p_0-\veps_\pp)}{1-e^{-\beta\veps_\pp}}
+\frac{\delta(p_0+\veps_\pp)}{e^{\beta\veps_\pp}-1}\Br].
\end{equation*}
Since the state $\omega_\beta$ is quasifree, the higher $n$-point functions vanish for $n$ odd, {\it i.e.} $w_\beta^{(n)}(x_1,\dots,x_{2n+1})=0$, $n\ge 0$ and the even $n$-point functions decompose into sums of products of two-point functions in the following way: 
\begin{equation*}
 w_\beta^{(n)}(x_1,\dots,x_{2n})=\sum_{\pi}\prod_{k=1}^n\omega_{\beta}^{(2)}(x_{\pi(k)},x_{\pi(k+n)}),\quad n\in\Nl,
\end{equation*}
where the sum runs over all permutations $\pi$ of the set $\{1,\dots,2n\}$ satisfying
\begin{equation*}
 \pi(1)<\dots<\pi(n),\quad \pi(k)<\pi(k+n),\quad k=1,\dots,n.
\end{equation*}

Next we proceed to the net of bounded operators which is generated by the fields $\phi_\beta(f)$. We denote the selfadjoint closure of $\phi_\beta(f)$ by the same symbol and by Borel functional calculus $\exp(i\phi_\beta(f))$ is a unitary on $\HS_\beta$. By construction there holds
\begin{equation*}
 \pi_\beta(V(f))=\exp(i\phi_\beta(f)),\quad f\in\hhs_m.
\end{equation*}
We will write $V_\beta(f):=\pi_\beta(V(f))$. From the covariance and locality properties of the thermal scalar free field $\phi_\beta(f)$ follows that the map
\begin{equation}
\label{eq:ThermalNet}
 \Rl^4\supset \OO\longmapsto \Aa_\beta(\OO):=\{V_\beta(f):\supp(f)\subset\OO\}''
\end{equation}
is an $\Euc$-covariant and local net of von Neumann algebras.\label{weylalgebraSFFthermal}

%===========================================================
\section{Deformation of the thermal scalar free field}
\label{sec:Deformationsofthermalscalarfreefields}
%===========================================================
In this section we define a deformation of the scalar free field in the AW representation along the same lines as in \cite{GrosseLechner:2007} by means of deformed creation and annihilation operators.
The important difference is that we do not use the action of the entire translation subgroup, but rather an $\Rl^2$-action which amounts to translations along the edge of the wedge $W_0$. This special form of the action appears to be suitable in this setting since it allows us to establish the wedge-locality of the deformed thermal field operators, despite the fact that the spectrum condition is not fulfilled.
%------------------------------------------------------------------------
\label{thermaldeformedcreationannihilationoperators}
\begin{definition}
\label{def:TFF_deformed_a}
For $p\in H_m^+$ define
\begin{equation}
\label{eq:DTF_deformed_a}
a_{\beta,\theta}(p):=e^{-ip\cdot\theta P_\beta}a_\beta(p),\quad
a^\dagger_{\beta,\theta}(p):=e^{ip\cdot\theta P_\beta}a^\dagger_\beta(p),
\end{equation}
with
\begin{equation*}
 \theta\in\Mat_-(4,\Rl):=\Bl\{
\begin{pmatrix}
0 & 0 & 0 & 0\\
0 & 0 & 0 & 0\\
0 & 0 & 0 & \kappa\\
0 & 0 & -\kappa & 0
\end{pmatrix}
:\kappa\in\Rl\Br\}.
\end{equation*}
\end{definition}
%------------------------------------------------------------------------
%------------------------------------------------------------------------
\begin{remark}
Note that the $a^\#_{\beta,\theta}(p)$ are the warped convolution of the $a^\#_{\beta}(p)$ with respect to translations along the edge $E_{W_0}\cong \Rl^2$ of the wedge $W_0$.
\end{remark}
%------------------------------------------------------------------------
From the definition (\ref{eq:DTF_deformed_a}) we see that
\begin{align}
\label{eq:DTF_deformed_aa1}
 a_{\beta,\theta}(p)
&=\sqrt{1+\rho_p}\,a_\theta(p)\otimes e^{ip\cdot\theta \bar{P}}
+e^{-ip\cdot\theta P}\otimes \sqrt{\rho_p}\,a^\dagger_\theta(p)\\
\label{eq:DTF_deformed_aa2}
 a^\dagger_{\beta,\theta} (p)
&=\sqrt{1+\rho_p}\,a^\dagger_\theta(p)\otimes e^{-ip\cdot\theta \bar{P}}
+e^{ip\cdot\theta P}\otimes \sqrt{\rho_p}\,a_\theta(p),
\end{align}
where $a_{\theta}(p)=e^{-i p\cdot\theta P}a(p)$ and $a^\dagger_{\theta}(p)=e^{ip\cdot\theta P}a^\dagger(p)$\label{vacuumdeformedcreationannihilationoperators} are the deformed annihilation and creation operators from \cite{GrosseLechner:2007}. The following lemma collects the covariance properties of these operators.

%------------------------------------------------------------------------
\begin{lemma}
\label{lem:TFF_Covariance_a}
Let $(y,R)\in\Euc$ and $p\in H^+_m$, $\theta\in\Mat_-(4,\Rl)$. Then
\begin{align*}
 U_\beta(y,R)a_{\beta,\theta}(p)U_\beta(y,R)^{-1}
&=e^{-iRp\cdot y}a_{\beta,R\theta R^T}(Rp)\\
 U_\beta(y,R)a^\dagger_{\beta,\theta}(p)U_\beta(y,R)^{-1}
&=e^{iRp\cdot y}a^\dagger_{\beta,R\theta R^T}(Rp).
\end{align*}
\end{lemma}
\begin{proof}
Use
\begin{align*}
U(y,R)a_\theta(p)U(y,R)^{-1}
&=e^{-iRp\cdot y}a_{R\theta R^T}(Rp)\\
U(y,R)a^\dagger_\theta(p)U(y,R)^{-1}
&=e^{iRp\cdot y}a^\dagger_{R\theta R^T}(Rp)
\end{align*}
from \cite[Lem.2.1]{GrosseLechner:2007} together with
\begin{equation*}
 U(y,R)e^{ip\cdot \theta P}U(y,R)^{-1}=e^{iRp\cdot R\theta R^T P}
\end{equation*}
and the fact that the multiplication operator $\rho$ commutes with $U(y,R)$, $(y,R)\in\Euc$.
\end{proof}
%------------------------------------------------------------------------

%------------------------------------------------------------------------
\begin{corollary}
$a_{\theta,\beta}(p)^*=a^\dagger_{\theta,\beta}(p)$.
\end{corollary}
\begin{proof}
Use $e^{iy\cdot P_\beta}a_\beta(p)=e^{-iy\cdot p}a_\beta(p)e^{iy\cdot P_\beta}$ and the antisymmetry of the matrix $\theta$.
\end{proof}	
%------------------------------------------------------------------------

The next lemma collects the exchange relations of the deformed thermal creation and annihilation operators. As a special case, we obtain commutation relations for operators with opposite deformation parameters.
%------------------------------------------------------------------------
\begin{lemma}
\label{lem:DTF_Covariance_a}
Let $p,p'\in H^+_m$ and $\theta,\theta'\in\Mat_-(4,\Rl)$. Then
\begin{align*}
 a_{\beta,\theta}(p)a_{\beta,\theta'}(p')
&=e^{ip\cdot(\theta+\theta')p'}a_{\beta,\theta'}(p')a_{\beta,\theta}(p)\\
 a^\dagger_{\beta,\theta}(p)a^\dagger_{\beta,\theta'}(p')
&=e^{ip\cdot(\theta+\theta')p'}a^\dagger_{\beta,\theta'}(p')a^\dagger_{\beta,\theta}(p)\\
a_{\beta,\theta}(p)a^\dagger_{\beta,\theta'}(p')
&=e^{-ip\cdot(\theta+\theta')p'}a^\dagger_{\beta,\theta'}(p')a_{\beta,\theta}(p)
+2\veps_\pp\delta(\pp-\pp')e^{ip\cdot(\theta-\theta')P_\beta}.
\end{align*}
Hence, for $\theta'=-\theta$ there holds
\begin{equation*}
 [a_{\beta,\theta}(p),a_{\beta,-\theta}(p')]
=[a^\dagger_{\beta,\theta}(p),a^\dagger_{\beta,-\theta}(p')]
=0,\quad
 [a_{\beta,\theta}(p),a^\dagger_{\beta,-\theta}(p')]
=2\veps_\pp\delta(\pp-\pp')e^{2ip\cdot\theta P_\beta}.
\end{equation*}
\end{lemma}

\begin{proof}
 Use the Definition (\ref{eq:DTF_deformed_a}) and Lemma \ref{lem:DTF_Covariance_a}
\end{proof}
%------------------------------------------------------------------------

For the definition of the deformed fields we proceed to smeared creation an annihilation operators:
\begin{equation*}
 a_{\beta,\theta}(\vphi):=\int d\mu_m(p)\overline{\vphi(p)}a_{\beta,\theta}(p),\quad
 a^\dagger_{\beta,\theta}(\vphi):=\int d\mu_m(p)\vphi(p)a^\dagger_{\beta,\theta}(p),
\quad \vphi\in\HS_1.
\end{equation*}

%------------------------------------------------------------------------
\begin{definition}
  Let $\theta\in\Mat_-(4,\Rl)$ and $f\in\Ss(\Rl^4,\Rl)$. Define
\begin{equation}
\label{eq:TFF_DeformedField}
 \phi_{\beta,\theta}(f):=\frac{1}{\sqrt{2}}[
a_{\beta,\theta}(\tilde{f}_+)+a^\dagger_{\beta,\theta}(\tilde{f}_+)].
\end{equation}
\end{definition}
%------------------------------------------------------------------------

%------------------------------------------------------------------------
\begin{remark}
Note that this field is, in contrast to (\ref{eq:TFF_definition}), not merely a superposition of two (deformed) vacuum fields from \cite{GrosseLechner:2007}, due to the presence of the extra translation operators in (\ref{eq:DTF_deformed_aa1}) and (\ref{eq:DTF_deformed_aa2}).
\end{remark}
%------------------------------------------------------------------------

The following proposition collects the basic Wightman properties of (\ref{eq:TFF_DeformedField}). The proofs are in fact very similar to \cite[Prop.2.2]{GrosseLechner:2007}, so we will be brief.
%------------------------------------------------------------------------
\begin{proposition}{\bf (Wightman properties of the field $\phi_{\beta,\theta}(f)$).}\\
\label{prop:TFF_WightmanProperties}
	Let $f\in\Ss(\Rl^4,\Rl)$ and $\theta\in\Mat_-(4,\Rl)$. Then:
	\begin{itemize}
		\item[i)] The subspace $\F_{\beta,0}$ of vectors of finite particle number is contained in the domain of each 
		$\phi_{\beta,\theta}(f)$ and stable under the action of these operators.
		\item[ii)] The map $f\to\phi_{\beta,\theta}(f)\Psi$, $\Psi\in\F_{\beta,0}$ is a vector-valued tempered distribution.
		\item[iii)] Each $\phi_{\beta,\theta}(f)$ is essentially self-adjoint.
	\end{itemize}
\end{proposition}
\begin{proof}
 $i)$ Since the translation operators do not change the particle number, there holds
$\F_{\beta,0}\subset D(\phi_{\beta,\theta}(f))$ and $\phi_{\beta,\theta}(f)\F_{\beta,0}\subset \F_{\beta,0}$ by the corresponding property of the undeformed operators (see Proposition \ref{prop:TFF_operatorproperties}).

$ii)$ Let $\Psi_{k,l}=\Psi_k\otimes\overline{\Psi_l}\in\HS_k\otimes\overline{\HS_l}$ with $k,l\in\Nl_0$. Using the bound (\ref{eq:TFF_bound_a}) for the thermal creation operator and the fact that $|e^{ip\theta q}|=1$, $p,q\in\Rl$, there follows
\begin{equation}
\label{eq:TFF_bound_a_theta}
 \|a^\#_{\beta,\theta}(\tilde{f}_+)\Psi_{k,l}\|
\le C_{\beta,m}\cdot \sqrt{\max(k,l)+1}\|\Psi_k\|\|\Psi_l\|
\Bl[\int d\mu_m(p)|\tilde{f}_+(p)|^2\Br]^{1/2}
\end{equation}
where $C_{\beta,m}$ is a positive constant independent of $f$. As the right hand side depends continuously on $f$ in the Schwartz space topology, the assertion follows.

$iii)$ Using the bound (\ref{eq:TFF_bound_a_theta}), one can show along the same lines as in \cite[Prop.5.2.3]{BratteliRobinson:1997} that each $\Psi\in\F_{\beta,0}$ is an entire analytic vector for $\phi_{\beta,\theta}(f)$. By Nelson's analytic vector theorem the statement follows.
\end{proof}

%------------------------------------------------------------------------

%------------------------------------------------------------------------
\begin{proposition}{\bf (Covariance of the field operators $\phi_{\beta,\theta}(f)$).}\\
\label{prop:TFF_CovarianceDeformedField}
Let $f\in\Ss(\Rl^4,\Rl)$ and $\theta\in\Mat_-(4,\Rl)$. Then:
\begin{equation*}
 U_\beta(y,R)\phi_{\beta,\theta}(f)U_\beta(y,R)^{-1}
=\phi_{\beta,R\theta R^T}(f_{(y,R)}),\quad
f_{(y,R)}(x)=f(R^{-1}(x-y))
\end{equation*}
for all $(y,R)\in\Euc$.
\end{proposition}
\begin{proof}
This follows from Lemma \ref{lem:TFF_Covariance_a}.
\end{proof}
%------------------------------------------------------------------------

%------------------------------------------------------------------------
\begin{theorem}{\bf (Wedge-locality of the field operators $\phi_{\beta,\theta}(f)$).}\\
\label{thm:TFF_WedgeLocality}
 Let $f,g\in\Ss(\Rl^4,\Rl)$ and $\theta\in\Mat_-(4,\Rl)$. Then
\begin{equation*}
 [\phi_{\beta,\theta}(f),\phi_{\beta,-\theta}(g)]\Psi=0,\quad \Psi\in\F_{\beta,0},
\end{equation*}
if $\supp(f)\subset W_0$ and $\supp(g)\subset (W_0)'$.
\end{theorem}
\begin{proof}
By Lemma \ref{lem:DTF_Covariance_a} there holds
\begin{align*}
2[\phi_{\beta,\theta}(f),\phi_{\beta,-\theta}(g)]
&= [a_{\beta,\theta}(\tilde{f}_+),a^\dagger_{\beta,-\theta}(\tilde{g}_+)]
-[a_{\beta,-\theta}(\tilde{g}_+),a^\dagger_{\beta,\theta}(\tilde{f}_+)]\\
&=\int d\mu_m(p)\Bl\{
\tilde{f}_+(-p)\tilde{g}_+(p)e^{2ip\cdot\theta P_\beta}-
\tilde{f}_+(p)\tilde{g}_+(-p)e^{-2ip\cdot\theta P_\beta}
\Br\}.
\end{align*}
Let $\Psi\in\F_{\beta,0}$. Then
\begin{align*}
 (2[\phi_{\beta,\theta}(f),&\phi_{\beta,-\theta}(g)]\Psi)_{k,l}(p_{(k)},q_{(l)})\\
&=\int d\mu_m(p)\Bl\{
\tilde{f}_+(-p)\tilde{g}_+(p)e^{2ip\cdot\theta a}-
\tilde{f}_+(p)\tilde{g}_+(-p)e^{-2ip\cdot\theta a}
\Br\}\cdot\Psi_{k,l}(p_{(k)},q_{(l)})\\
&=\int d\mu_m(p)\Bl\{
\tilde{f}_+(-p)\widetilde{(g_{2\theta a})}_+(p)-
\tilde{f}_+(p)\widetilde{(g_{2\theta a})}_+(-p)
\Br\}\cdot\Psi_{k,l}(p_{(k)},q_{(l)})
\end{align*}
where we used the notation
\begin{equation*}
 p_{(k)}:=(p_1,\dots,p_k),\quad q_{(l)}:=(q_1,\dots,q_l),\quad p_j,q_j\in H_m^+
\end{equation*}
and $a:=\sum_{n=1}^k p_n-\sum_{n'=1}^l q_{n'}$. Since $\sum_{n=1}^k p_n,\sum_{n'=1}^l q_{n'} \in V^+$ we have $a\in\Rl^4$ but $\theta a\in E_{W_0}\subset \overline{W_0}$. Hence
\begin{equation*}
 \supp(g_{2\theta a})\subset W_0\, ',\quad a\in\Rl^4
\end{equation*}
and the wedge-locality follows from the wedge-locality of the undeformed field.
\end{proof}
%------------------------------------------------------------------------
%------------------------------------------------------------------------
\begin{remark}
    This pedestrian way of proving wedge locality can be understood from a more general point of view by using a result from \cite{BuchholzLechnerSummers:2010}. The authors show that the warped operators $A_\theta$ and $B_{-\theta}$ commute, if $[\alpha_{\theta p}(A),\alpha_{-\theta q}(B)]=0$ holds for all $p,q$ in the joint spectrum of the generators of the spacetime translations. In our case this condition is trivially satisfied for all $p,q\in\Rl^4$ due to the special form of the matrix $\theta$.
\end{remark}
%------------------------------------------------------------------------

Next we compute the $n$-point functions
\begin{equation*}
 w^{(n)}_{\beta,\theta}(f_1,\dots,f_n)
=\left<\Omega_\beta,\phi_{\beta,\theta}(f_1)\dots\phi_{\beta,\theta}(f_n)\Omega_\beta\right>
\end{equation*}
of the deformed thermal fields. Their distributional kernels are most conveniently expressed in momentum space. The structure that we obtain is very similar to the vacuum case \cite{GrosseLechner:2007}.
%------------------------------------------------------------------------
\begin{proposition}{\bf (Deformed $n$-point functions for the field $\phi_{\beta,\theta}(f)$).}\\
Let $n\in\Nl$ and $p_1,\dots,p_n\in H_m^+$. Then
\begin{equation*}
 \tilde{w}^{(n)}_{\beta,\theta}(p_1,\dots,p_n)
=\prod_{1\le k<l\le n}e^{ip_k\cdot \theta p_l}
\tilde{w}^{(n)}_{\beta}(p_1,\dots,p_n).
\end{equation*}
\end{proposition}
\begin{proof}
The Fourier transform $\tilde{\phi}_{\beta,\theta}(p)$ of $\phi_{\beta,\theta}(x)$ reads
\begin{equation*}
 \tilde{\phi}_{\beta,\theta}(p)=e^{ip\cdot \theta P_\beta}\tilde{\phi}_\beta(p)
=\tilde{\phi}_\beta(p)e^{ip\cdot \theta P_\beta},
\end{equation*}
where $\tilde{\phi}_{\beta}(p)$ is the Fourier transform of $\phi_{\beta}(x)$. The last equality follows from the covariance of the field and the antisymmetry of $\theta$. Again, by covariance there holds
\begin{equation*}
    e^{ip\cdot \theta P}\tilde{\phi}_\beta(p')
=e^{ip\cdot \theta p'}\tilde{\phi}_\beta(p')e^{ip\cdot \theta P}.
\end{equation*}
The statements follows from an iteration of this equality together with the translation invariance of $\Omega_\beta$.
\end{proof}
%------------------------------------------------------------------------

After we have described the deformation of a single operator $\phi_\beta(f)$, we now come to the deformation of the net $W\to\Aa_\beta(W)$ from (\ref{eq:ThermalNet}). This net is, in particular, wedge-local and we focus on localization in wedges, since this turns out to be stable under the deformation. The deformed net $\Aa_{\beta,\kappa},\,\kappa\in\Rl$ will be defined in terms of exponentiated deformed thermal field operators: We denote by $\phi_{\beta,\theta}(f)$ the self-adjoint closure of (\ref{eq:TFF_DeformedField}) and define 
\begin{equation*}
 V_{\beta,\theta}(f):=\exp(i\phi_{\beta,\theta}(f)),\quad f\in\hhs_m.
\end{equation*}
From the covariance of the fields (see Proposition \ref{prop:TFF_CovarianceDeformedField}) follows
\begin{equation}
\label{eq:TFF_WeylCovarianceDeformed}
 U(y,R)V_{\beta,\theta}(f)U(y,R)^{-1}=V_{\beta,R\theta R^T}(f_{(y,R)}),\quad (y,R)\in\Euc,\;f\in\hhs_m.
\end{equation}
Moreover, the following lemma holds.
%------------------------------------------------------------------------
\begin{lemma}
\label{lem:TFF_Weyl_Locality}
 Let $f,g\in\hhs_m$ with $\supp(f)\subset W_0$, $\supp(g)\subset W_0\,'$. Then
\begin{equation*}
[V_{\beta,\theta}(f),V_{\beta,-\theta}(g)]=0. 
\end{equation*}
\end{lemma}
\begin{proof}
In part $iii)$ of Proposition \ref{prop:TFF_WightmanProperties} we have shown that each $\Psi\in\F_{\beta,0}$ is an entire analytic vector for $\phi_{\beta,\theta}(f)$. Similarly, one can show that $V_{\beta,-\theta}(g)\Psi$, $\Psi\in\F_{\beta,0}$ is also an entire analytic vector for $\phi_{\beta,\theta}(f)$. Hence, for $\Psi\in\F_{\beta,0}$ we can compute the commutator with the power series
\begin{equation*}
 [V_{\beta,\theta}(f),V_{\beta,-\theta}(g)]\Psi
=\sum_{n,n'=0}^\infty\frac{i^{n+n'}}{n!n'!}
[\phi_{\beta,\theta}(f)^n,\phi_{\beta,-\theta}(g)^{n'}]\Psi.
\end{equation*}
 Now the statement follows from Theorem \ref{thm:TFF_WedgeLocality}.
\end{proof}
%------------------------------------------------------------------------

As we mentioned in Section \ref{sec:MinkowskiSpacetime}, the group $\Euc$ does not act transitively on the set of wedges $\W$. Hence it is not sufficient to prescribe the deformation of a single initial algebra to generate a net as in Section \ref{sec:CAQFT}, but we rather need to specify a collection of algebras $-$ one for each $\W_{[s,\ee]}$. Let $(s,\ee)$ be a representative of the equivalence class $[s,\ee]$ and define
\begin{equation*}
 \M^{s,\ee}_{\beta,\kappa}
:=\{V_{\beta,\Lambda_\ee(s)\theta\Lambda_\ee(s)^T}(f)
:\supp(f)\subset \Lambda_\ee(s)W_0\}''.
\end{equation*}
%------------------------------------------------------------------------
\begin{theorem}
There holds
 \begin{enumerate}
  \item[i)] $(y,R)\Lambda_\ee(s)W_0\subset \Lambda_\ee(s)W_0$ $\Longrightarrow$
$\alpha_{(y,R)}(\M^{s,\ee}_{\beta,\kappa})\subset \M^{s,\ee}_{\beta,\kappa}$.
  \item[ii)]$(y,R)\Lambda_\ee(s)W_0\subset (\Lambda_\ee(s)W_0)'$ $\Longrightarrow$ $\alpha_{(y,R)}(\M^{s,\ee}_{\beta,\kappa})\subset (\M^{s,\ee}_{\beta,\kappa})'$
 \end{enumerate}
\end{theorem}
\begin{proof}
 $i)$ By assumption we have $(\Lambda_\ee(s)^{-1}y,\Lambda_\ee(s)^{-1}R\Lambda_\ee(s))W_0\subset W_0$, which implies
\begin{equation*}
 \Lambda_\ee(s)^{-1}R\Lambda_\ee(s)\theta (\Lambda_\ee(s)^{-1}R\Lambda_\ee(s))^T
=\theta.
\end{equation*}
by \cite[Lem.3.1]{GrosseLechner:2007}.
Hence
\begin{equation}
\label{eq:SFF_IsotonyCondition}
R\Lambda_\ee(s)\theta\Lambda_\ee(s)^{T}R^T=\Lambda_\ee(s)\theta\Lambda_\ee(s)^{T}
\end{equation}
For $V_{\beta,\Lambda_\ee(s)\theta\Lambda_\ee(s)^T}(f)\in \M^{s,\ee}_{\beta,\kappa}$ there holds
\begin{equation*}
 U(y,R)V_{\beta,\Lambda_\ee(s)\theta\Lambda_\ee(s)^T}(f)U(y,R)^{-1}
=V_{\beta,R\Lambda_\ee(s)\theta\Lambda_\ee(s)^TR^T}(f_{(y,R)})
=V_{\beta,\Lambda_\ee(s)\theta\Lambda_\ee(s)^T}(f_{(y,R)})
\end{equation*}
by covariance (\ref{eq:TFF_WeylCovarianceDeformed}) and (\ref{eq:SFF_IsotonyCondition}). As
$\supp(f_{(y,R)})\subset (y,R)\Lambda_\ee(s)W_0\subset \Lambda_\ee(s)W_0$
by assumption, we find
\begin{equation*}
 U(y,R)V_{\beta,\Lambda_\ee(s)\theta\Lambda_\ee(s)^T}(f)U(y,R)^{-1}\in
\M^{s,\ee}_{\beta,\kappa}.
\end{equation*}
The assertion follows by taking the weak closure of 
\begin{equation*}
 \{U(y,R)V_{\beta,\Lambda_\ee(s)\theta\Lambda_\ee(s)^T}(f)U(y,R)^{-1}
:\supp(f)\subset \Lambda_\ee(s)W_0\}.
\end{equation*}

$ii)$ There holds $(\Lambda_\ee(s)W_0)'=\Lambda_\ee(s)W_0\,'$ and $(\Lambda_\ee(s)^{-1}y,\Lambda_\ee(s)^{-1}R\Lambda_\ee(s))W_0\subset W_0\,'$ by assumption. This implies
\begin{equation*}
 \Lambda_\ee(s)^{-1}R\Lambda_\ee(s)\theta (\Lambda_\ee(s)^{-1}R\Lambda_\ee(s))^T
=-\theta.
\end{equation*}
by \cite[Lem.3.1]{GrosseLechner:2007}.
Hence
\begin{equation}
\label{eq:SFF_LocalityCondition}
R\Lambda_\ee(s)\theta\Lambda_\ee(s)^{T}R^T=-\Lambda_\ee(s)\theta\Lambda_\ee(s)^{T}.
\end{equation}
Let $f,\hat{f}\in\hhs_m$, with $\supp(f),\supp(\hat{f})\subset \Lambda_\ee(s)W_0$, and consider 
$$
V_{\beta,\Lambda_\ee(s)\theta\Lambda_\ee(s)^T}(f),
V_{\beta,\Lambda_\ee(s)\theta\Lambda_\ee(s)^T}(\hat{f})
\in \M^{s,\ee}_{\beta,\kappa}.
$$ 
Then
\begin{equation*}
 U(y,R)V_{\beta,\Lambda_\ee(s)\theta\Lambda_\ee(s)^T}(f)U(y,R)^{-1}
=V_{\beta,R\Lambda_\ee(s)\theta\Lambda_\ee(s)^TR^T}(f_{(y,R)})
=V_{\beta,-\Lambda_\ee(s)\theta\Lambda_\ee(s)^T}(f_{(y,R)})
\end{equation*}
by covariance and (\ref{eq:SFF_LocalityCondition}). Hence
\begin{align}
 \bl[U(y,R)V_{\beta,\Lambda_\ee(s)\theta\Lambda_\ee(s)^T}(f)&U(y,R)^{-1},
V_{\beta,\Lambda_\ee(s)\theta\Lambda_\ee(s)^T}(\hat{f})\br]\nonumber\\
\label{eq:TFF_Commutator}
&=\alpha_{\Lambda_\ee(s)}\Bl(
\bl[V_{\beta,-\theta}(f_{(0,\Lambda_\ee(s)^{-1})(y,R)}))
,V_{\beta,\theta}(\hat{f}_{(0,\Lambda_\ee(s)^{-1})})\br]
\Br).
\end{align}
As $\supp(f),\supp(\hat{f})\subset \Lambda_\ee(s) W_0$ by assumption, there holds 
\begin{equation*}
\supp(\hat{f}_{(0,\Lambda_\ee(s)^{-1})})\subset W_0,
\quad 
\supp(f_{(0,\Lambda_\ee(s)^{-1})(y,R)})\subset W_0\, '.
\end{equation*}
Hence by Theorem \ref{thm:TFF_WedgeLocality} the commutator (\ref{eq:TFF_Commutator}) vanishes, {\it i.e.}
\begin{equation*}
 U(y,R)V_{\beta,\Lambda_\ee(s)\theta\Lambda_\ee(s)^T}(f)U(y,R)^{-1}
\in (\M^{s,\ee}_{\beta,\kappa})'.
\end{equation*}
The assertion follows by taking the weak closure of 
\begin{equation*}
 \{U(y,R)V_{\beta,\Lambda_\ee(s)\theta\Lambda_\ee(s)^T}(f)U(y,R)^{-1}
:\supp(f)\subset \Lambda_\ee(s)W_0\}.
\end{equation*}
\end{proof}
%------------------------------------------------------------------------
%------------------------------------------------------------------------
\begin{corollary}
 The map
\begin{equation*}
 W:=R\Lambda_\ee(s) W_0+y\longmapsto
\alpha_{(y,R)}(\M^{s,\ee}_{\beta,\kappa})=:\Aa_{\beta,\kappa}(W).
\end{equation*}
is an $\Euc$-covariant and wedge-local net of von Neumann algebra.
\end{corollary}
%------------------------------------------------------------------------

As Haag-Ruelle scattering theory is not available in representations which are associated with KMS states \cite{NarhoferRequardtThirring:1983}, we use the following more indirect method to show that the model which we have constructed is actually something new and not the old theory rewritten in a complicated way. %------------------------------------------------------------------------
\begin{theorem}
 The nets $\Aa_\beta$ and $\Aa_{\beta,\kappa}$ are unitarily inequivalent for $\kappa\neq 0$.
\end{theorem}
\begin{proof}
 Consider the wedge $W_0$ together with a rotation $r_\vtheta$ in the $(x^1,x^2)$-plane about a fixed angle $\vtheta>0$ and define
 \begin{equation*}
 \K:=r_\vtheta W_0\cap r_{-\vtheta}W_0,\quad |\vtheta|<\frac{\pi}{2}.
 \end{equation*}
 Pick some $f\in\hhs_m$ with $\supp(f)\subset\K$. Obviously there holds $\supp(f), \supp(f_{r_{-\vtheta}})\subset W_0$ and by covariance we have
\begin{equation*}
 U(r_\vtheta)\phi_{\beta,\theta}(f_{r_{-\vtheta}})U(r_\vtheta)^{-1}
=\phi_{\beta,r_\vtheta\theta r_\vtheta^T}(f).
\end{equation*}
 
 In the first step, we show that the operators $\phi_{\beta,\theta}(f)$ and $\phi_{\beta,r_\vtheta\theta r_\vtheta^T}(f)$ are affiliated with the algebra $\Aa_{\beta,\kappa}(W_0)\vee \Aa_{\beta,\kappa}(r_\vtheta W_0)$. The line of argument is essentially the same as in the proof of \cite[Thm.5.2]{Lechner:2011}, but we include it for matters of completeness. First we show that $F:=\phi_{\beta,\theta}(f)$ is affiliated with $\Aa_{\beta,\kappa}(W_0)$. Let $\Psi\in D(F)$ be a vector in the domain of $F$ and $\Psi_0\in\F_{\beta,0}$ a finite particle vector. As $F^*$ changes the particle number only by a finite amount, it follows that $\Psi_0$ and $F^*\Psi_0$ are entire analytic vectors for $G':=\phi_{\beta,-\theta}(g')$ with $g'\in\hhs_m$ and $\supp(g')\subset W_0\, '$. Since $F^*$ and $(G')^k$ commute on $\F_{\beta,0}$ for all $k\in\Nl_0$ by Theorem \ref{thm:TFF_WedgeLocality}, there follows
\begin{align*}
 \bl<\Psi_0,e^{iG'}F\Psi\br>
&=\bl<F^*e^{-iG'}\Psi_0,\Psi\br>\\
&=\sum_{k=0}^\infty\frac{(-i)^k}{k!}\bl<F^*(G')^k\Psi_0,\Psi\br>\\
&=\sum_{k=0}^\infty\frac{(-i)^k}{k!}\bl<(G')^kF^*\Psi_0,\Psi\br>\\
&=\bl<e^{-iG'}F^*\Psi_0,\Psi\br>\\
&=\bl<\Psi_0,Fe^{iG'}\Psi\br>
\end{align*}
As $\F_{\beta,0}$ is dense in $\HS_\beta$, there holds $e^{iG'}F\Psi=Fe^{iG'}\Psi$ for all $\Psi\in D(F)$. Now, this identity also holds for all operators in the $*$-algebra $\Aa$ which is algebraically generated by $\exp(i\phi_{\beta,\theta}(g'))$ with $\supp(g')\subset W_0\,'$. Moreover, any $A'\in\Aa_{\beta,\kappa}(W_0)'$ is a weak limit of a sequence $A_n'$ in $\Aa$ and since $A_n'F\Psi=FA_n'\Psi$ is stable under taking weak limits there follows $A'F\Psi=FA'\Psi$ for all $A'\in \Aa_{\beta,\kappa}(W_0)'$, {\it i.e.} $F$ is affiliated with $\Aa_{\beta,\kappa}(W_0)$. The proof that $\phi_{\beta,r_\vtheta\theta r_\vtheta^T}(f)$ is affiliated with $\Aa_{\beta,\kappa}(r_\vtheta W_0)$ works the same way. Hence $\phi_{\beta,\theta}(f)$ and $\phi_{\beta,r_\vtheta\theta r_\vtheta^T}(f)$ are affiliated with the generated von Neumann-algebra $\Aa_{\beta,\kappa}(W_0)\vee \Aa_{\beta,\kappa}(r_\vtheta W_0)$.

Now, from the translation invariance of the state $\omega_\beta$ follows
\begin{equation*}
 \phi_{\beta,\theta}(f)\Omega_\beta
=\phi_{\beta}(f)\Omega_\beta
=\phi_{\beta,r_\vtheta\theta r_\vtheta^T}(f)\Omega_\beta.
\end{equation*}
The KMS condition implies that $\Omega_\beta$ is separating for $\pi_\beta(\Aa_m)''$ (see Proposition \ref{prop:KMS_basicproperties}), so it is also separating for $\Aa_{\beta}(W_0)\vee \Aa_{\beta}(r_\vtheta W_0)$. For a contradiction, assume that $\Omega_\beta$ is also separating for the deformed algebra $\Aa_{\beta,\kappa}(W_0)\vee \Aa_{\beta,\kappa}(r_\vtheta W_0)$. Then
\begin{equation}
\label{eq:TFF_SoughtContradiction}
 \phi_{\beta,\theta}(f)=\phi_{\beta,r_\vtheta\theta r_\vtheta^T}(f),
\end{equation}
since $\phi_{\beta,\theta}(f)$ and $\phi_{\beta,r_\vtheta\theta r_\vtheta^T}(f)$ are affiliated with this algebra.
We proceed by showing that these operators are in fact not equal. Consider the Fock vector $\psi\otimes \bar{\Omega}$, $\psi\in\HS_1$ describing a one-particle, no-hole state. For the undeformed thermal creation and annihilation operators there holds
\begin{align*}
 a_\beta(\vphi)(\psi\otimes\bar{\Omega})
 &=\bl<\sqrt{1+\rho}\,\vphi,\psi\br>_m(\Omega\otimes\bar\Omega)
+\psi\otimes \overline{\sqrt{\rho}\,\vphi}\\
 a^\dagger_\beta(\vphi)(\psi\otimes\bar{\Omega})
 &=(\sqrt{1+\rho}\,\vphi\otimes_s \psi)\otimes\bar{\Omega}.
\end{align*}
Hence, for the 1,1-contribution of the deformed thermal creation and annihilation operators we obtain (see Definition \ref{def:TFF_deformed_a})
\begin{equation*}
 ([a_{\beta,\theta}(\vphi)+a^\dagger_{\beta,\theta}(\vphi)](\psi\otimes\bar{\Omega}))_{1,1}(p,q)
=(a_{\beta,\theta}(\vphi)(\psi\otimes\bar{\Omega}))_{1,1}(p,q)
=e^{iq\cdot\theta p}\psi(p)\sqrt{\rho_q}\cdot \overline{\vphi(q)},
\end{equation*}
which implies 
\begin{equation}
\label{eq:InequivalenceCrucialVector}
 ([\phi_{\beta,\theta}(f)-\phi_{\beta,r_\vtheta\theta r_\vtheta^T}(f)](\psi\otimes\bar{\Omega}))_{1,1}(p,q)
=\frac{1}{\sqrt{2}}(e^{iq\cdot\theta p}-e^{iq\cdot r_\vtheta\theta r_\vtheta^T p})\psi(p)\sqrt{\rho_q}\cdot \tilde{f}_+(-q)=0
\end{equation}
for all $f\in\hhs_m$, $\psi\in\HS_1$ and $p,q\in H_m^+$ by (\ref{eq:TFF_SoughtContradiction}). Excluding the trivial case $f,\psi\equiv 0$ we find
\begin{equation}
\label{eq:InequivalenceMatrixElements}
 q\cdot(\theta - r_\vtheta\theta r_\vtheta^T)p=2\pi k,
\end{equation}
for some $k=k(p,q)\in\Zl$. As the matrix $\theta$ is antisymmetric, there holds $k(p,p)=0$ for all $p\in H^+_m$. This and the continuity of $(p,q)\mapsto k(p,q)$ implies that $k(p,q)=0$ for all $p,q\in H^+_m$. Since the matrix elements (\ref{eq:InequivalenceMatrixElements}) coincide for all $p,q\in H_m^+$, it follows that the matrices $\theta$ and $r_\vtheta\theta r_\vtheta^T$ must coincide. However, as $\Lpo\ni\Lambda\to\Lambda\theta\Lambda^T$ is a homomorphism of homogeneous spaces \cite[Lem.3.1]{GrosseLechner:2007}, there follows $\vtheta=0$ which contradicts our initial assumption about the rotation $r_\vtheta$. Hence $\Omega_\beta$ is in fact not separating for $\Aa_{\beta,\kappa}(W_0)\vee \Aa_{\beta,\kappa}(r_\vtheta W_0)$.

If the nets $\Aa_\beta$ and $\Aa_{\beta,\kappa}$ were unitarily equivalent, there would exist a unitary $V$, such that $V^*\Aa_{\beta,\kappa}(W)V=\Aa_{\beta}(W)$ for all $W\in\W$ and $V\Omega_\beta=\Omega_\beta$. But such a unitary would clearly preserve the separating property for wedge-algebras, from which we conclude that $\Aa_\beta$ and $\Aa_{\beta,\kappa}$ are not unitarily equivalent.
\end{proof}
%------------------------------------------------------------------------

%%%%%%%%%%%%%%%%%%%%%%%%%%%%%%%%%%%%%%%%%%%%%%%%%%%%%%%%%%%%
\chapter{Quantum Field Theories on Cosmological Spacetimes}
\label{ch:Cosmological}
%%%%%%%%%%%%%%%%%%%%%%%%%%%%%%%%%%%%%%%%%%%%%%%%%%%%%%%%%%%%

\noindent
In order to apply the deformation scheme from Section \ref{sec:CAQFT} to quantum field theories on curved manifolds, we will consider in this chapter spacetimes with a sufficiently large isometry group containing two commuting Killing fields, which give rise to a representation of $\Rl^2$ as required in (\ref{eq:IntroWarpedConvolution}). This setting is wide enough to encompass a number of cosmologically relevant manifolds such as the Bianchi type models I-VII. Making use of the algebraic framework of quantum field theory, we can then formulate quantum field theories in an operator-algebraic language and study their deformations. Despite the fact that the warped convolution was invented for the deformation of Minkowski space quantum field theories, it turns out that all reference to the particular structure of flat spacetime, such as Poincar\'e transformations and a Poincar\'e invariant vacuum state, can be avoided.

We are interested in understanding to what extent the familiar structure of quantum field theories on curved spacetimes is preserved under such deformations, and investigate in particular covariance and localization properties. Concerning locality, it is known that in warped models on Minkowski space, point-like localization is weakened to localization in wedges \cite{GrosseLechner:2007,BuchholzSummers:2008,GrosseLechner:2008,BuchholzLechnerSummers:2010}. Because of their intimate relation to the Poincar\'e symmetry of Minkowski spacetime, it is not obvious what a good replacement for such a collection of regions is in the presence of non-vanishing curvature. In fact, different definitions are possible, and wedges on special manifolds have been studied by many authors in the literature \cite{Kay:1985,BorchersBuchholz:1999,Rehren:2000,BuchholzMundSummers:2001,GuidoLongoRobertsVerch:2001,BuchholzSummers:2004-2,LauridsenRibeiro:2007,Strich:2008,Borchers:2009}.
\\
\\
This chapter is structured as follows. In Section \ref{sec:geometry} we show that on those four-dimensional curved spacetimes which allow for the application of the deformation methods from \cite{BuchholzLechnerSummers:2010}, and thus carry two commuting Killing fields, there also exists a family of wedges with causal properties analogous to the Minkowski space wedges. Because of the prominent role wedges play in many areas of Minkowski space quantum field theory \cite{BisognanoWichmann:1975, Borchers:1992, Borchers:2000, BuchholzDreyerFlorigSummers:2000, BrunettiGuidoLongo:2002}, this geometric and manifestly covariant construction is also of interest independently of its relation to deformations.

In Section \ref{sec:deformation}, we then consider quantum field theories on curved spacetimes, and deform them by warped convolution. We first show that these deformations can be carried through in a model-independent, operator-algebraic framework, and that the emerging models share many structural properties with deformations of field theories on flat spacetime (see Section \ref{sec:generaldeformations}). In particular, deformed quantum fields are localized in the wedges of the considered spacetime. These and further aspects of deformed quantum field theories are discussed in the concrete example of a Dirac field in Section \ref{sec:dirac}.

%===========================================================
\section{Geometric setup}\label{sec:geometry}
%===========================================================

To prepare the ground for our discussion of deformations of quantum field theories on curved backgrounds, we introduce in this section a suitable class of spacetimes and study their geometrical properties. In particular, we show how the concept of {\em wedges}, known from Minkowski space, generalizes to these manifolds. Recall in preparation that a wedge in four-dimensional Minkowski space is a region which is bounded by two non-parallel characteristic hyperplanes \cite{ThomasWichmann:1997}, or, equivalently, a region which is a connected component of the causal complement of a two-dimensional spacelike plane. The latter definition has a natural analogue in the curved setting. Making use of this observation, we construct corresponding wedge regions in Section \ref{sec:edges+wedges}, and analyse their covariance, causality and inclusion properties. At the end of that section, we compare our notion of wedges to other definitions which have been made in the literature \cite{BorchersBuchholz:1999, BuchholzMundSummers:2001, BuchholzSummers:2004-2, LauridsenRibeiro:2007, Borchers:2009}, and point out the similarities and differences.

In Section \ref{sec:examples}, the abstract analysis of wedge regions is complemented by a number of concrete examples of spacetimes fulfilling our assumptions.

%===========================================================
\subsection{Edges and wedges in curved spacetimes}\label{sec:edges+wedges}
%===========================================================

Let $(M,\gST)$ be a globally hyperbolic spacetime manifold\footnote{Our conventions concerning curved spacetimes are collected in Appendix \ref{ch:NotationConventions}}.
To avoid pathological geometric situations such as closed causal curves, and also to define a full-fledged Cauchy problem for a free field theory whose dynamics is determined by a second order hyperbolic partial differential equation, we will restrict ourselves to globally hyperbolic spacetimes. So in particular, $M$ is orientable and time-orientable, and we fix both orientations. The (open) causal complement of a set $\OO\subset M$ is defined as
\label{causalcomplement}
\begin{align}
 \OO ' := M \backslash \Big[\overline{J^+(\OO)}\cup \overline{J^-(\OO)}\Big]\,,
\end{align}
where $J^\pm(\OO)$ is the causal future respectively past of $\OO$ in $M$ \cite[Section 8.1]{Wald:1984}. While this setting is standard in quantum field theory on curved backgrounds, we will make additional assumptions regarding the structure of the isometry group $\iso(M,\gST)$ of $(M,\gST)$, motivated by our desire to define wedges in $M$ which resemble those in Minkowski space.

Our most important assumption on the structure of $(M,\gST)$ is that it admits two linearly independent, spacelike, complete, commuting smooth Killing fields $\xi_1,\xi_2$, which will later be essential in the context of deformed quantum field theories. We refer here and in the following always to pointwise linear independence, which entails in particular that these vector fields have no zeros. Denoting the flows of $\xi_1,\xi_2$ by $\flow_{\xi_1},\flow_{\xi_2}$, the orbit of a point $p\in M$ is a smooth two-dimensional spacelike embedded submanifold of $M$,
\begin{align}\label{edge1}
 E := \{\flow_{\xi_1,s_1}(\flow_{\xi_2,s_2}(p))\in M \,:\, s_1,s_2\in\Rl\}\,,
\end{align}
where $s_1,s_2$ are the flow parameters of $\xi_1,\xi_2$.

Since $M$ is globally hyperbolic, it is isometric to a smooth product manifold $\Rl\times\Sigma$, where $\Sigma$ is a smooth, three-dimensional embedded Cauchy hypersurface. It is known that the metric splits according to $\gST=\beta d\mathcal{T}^2- h$ with a temporal function $\mathcal{T}:\Rl\times\Sigma\ra\Rl$ and a positive function $\beta\in C^\infty(\Rl\times\Sigma,(0,\infty))$, while $h$ induces a smooth Riemannian metric on $\Sigma$ \cite[Thm.2.1]{BernalSanchez:2005}.  We assume that, with $E$ as in \eqref{edge1}, the Cauchy surface $\Sigma$ is smoothly homeomorphic to a product manifold $I\times E$, where $I$ is an open interval or the full real line. Thus $M\cong \Rl\times I\times E$, and we require in addition that there exists a smooth embedding $\iota:\Rl\times I\ra M$. By our assumption on the topology of $I$, it follows that $(\Rl\times I,\iota^*\gST)$ is a globally hyperbolic spacetime without null focal points, a feature that we will need in the subsequent construction of wedge regions.

%===========================================================
\begin{definition}\label{admissible}
A spacetime $(M,\gST)$ is called {\em admissible} if it admits two linearly independent, spacelike, complete, commuting, smooth Killing fields $\xi_1,\xi_2$ and the corresponding split $M\cong \Rl\times I\times E$, with $E$ defined in \eqref{edge1}, has the properties described above.

The set of all ordered pairs $\xi:=(\xi_1,\xi_2)$ satisfying these conditions for a given admissible spacetime $(M,\gST)$ is denoted $\Xi(M,\gST)$. The elements of $\Xi(M,\gST)$ will be referred to as Killing pairs.
\end{definition}
%===========================================================

For the remainder of this section, we will work with an arbitrary but fixed admissible spacetime $(M,\gST)$, and usually drop the $(M,\gST)$-dependence of various objects in our notation, {\it e.g.} write $\Xi$ instead of $\Xi(M,\gST)$ for the set of Killing pairs, and $\iso$ in place of $\iso(M,\gST)$ for the isometry group. Concrete examples of admissible spacetimes, such as Friedmann-Robertson-Walker-, Kasner- and Bianchi-spacetimes, will be discussed in Section \ref{sec:examples}.

The flow of a Killing pair $\xi\in\Xi$ is written as
\begin{align}
 \flow_{\xi,s}:=\flow_{\xi_1,s_1}\circ \flow_{\xi_2,s_2}=\flow_{\xi_2,s_2}\circ \flow_{\xi_1,s_1}\,,\qquad \xi=(\xi_1,\xi_2)\in\Xi,\;\;s=(s_1,s_2)\in\Rl^2,
\end{align}
where $s_1,s_2\in\Rl$ are the parameters of the (complete) flows $\flow_{\xi_1}, \flow_{\xi_2}$ of $\xi_1,\xi_2$. By construction, each $\flow_\xi$ is an isometric $\Rl^2$-action by diffeomorphisms on $(M,\gST)$, {\it i.e.} $\varphi_{\xi,s}\in \iso$ and $\varphi_{\xi,s}\varphi_{\xi,u}=\varphi_{\xi,s+u}$ for all $s,u\in\Rl^2$.

On the set $\Xi$, the isometry group $\iso$ and the general linear group $\GL(2,\Rl)$ act in a natural manner.
%===========================================================
\begin{lemma}\label{lemma:group-xi}
 Let $h\in \iso$, $N\in \GL(2,\Rl)$, and define, $\xi=(\xi_1,\xi_2)\in\Xi$,
\begin{align}
 h_*\xi &:= (h_*\xi_1,\,h_*\xi_2)\,,\\
 (N\xi)(p) &:= N(\xi_1(p),\xi_2(p))\,,\qquad p\in M\,.
\end{align}
These operations are commuting group actions of \,$\iso$ and $\GL(2,\Rl)$ on $\Xi$, respectively. The $\GL(2,\Rl)$-action transforms the flow of $\xi\in\Xi$ according to, $s\in\Rl^2$,
\begin{align}\label{N-flow}
 \varphi_{N\xi,s} &= \varphi_{\xi,N^Ts}\,.
\end{align}
If $h_*\xi=N\xi$ for some $\xi\in\Xi$, $h\in\iso$, $N\in\GL(2,\Rl)$, then $\det N=\pm1$.
\end{lemma}
%===========================================================
\begin{proof}
Due to the standard properties of isometries, $\iso$ acts on the Lie algebra of Killing fields by the push-forward  isomorphisms $\xi_1\mapsto h_*\xi_1$ \cite{O'Neill:1983}. Therefore, for any $(\xi_1,\xi_2)\in\Xi$, also the vector fields $h_*\xi_1, h_*\xi_2$ are spacelike, complete, commuting, linearly independent, smooth Killing fields. The demanded properties of the splitting $M\cong \Rl\times I\times E$ directly carry over to the corresponding split with respect to $h_*\xi$. So $h_*$ maps $\Xi$ onto $\Xi$, and since $h_*(k_*\xi_1)=(hk)_*\xi_1$ for $h,k\in\iso$,  we have an action of $\iso$.

The second map, $\xi\mapsto N\xi$, amounts to taking linear combinations of the Killing fields $\xi_1,\xi_2$. The relation \eqref{N-flow} holds because $\xi_1,\xi_2$ commute and are complete, which entails that the respective flows can be constructed via the exponential map. Since $\det N\neq0$, the two components of $N\xi$ are still linearly independent, and since $E$ \eqref{edge1} is invariant under $\xi\mapsto N\xi$, the splitting $M\cong \Rl\times I\times E$ is the same for $\xi$ and $N\xi$. Hence $N\xi\in\Xi$, {\it i.e.} $\xi\mapsto N\xi$ is a $\GL(2,\Rl)$-action on $\Xi$, and since the push-forward is linear, it is clear that the two actions commute.

To prove the last statement, we consider the submanifold $E$ \eqref{edge1} together with its induced metric. Since the Killing fields $\xi_1,\xi_2$ are tangent to $E$, their flows are isometries of $E$. Since $h_*\xi=N\xi$ and $E$ is two-dimensional, it follows that $N$ acts as an isometry on the tangent space $T_pE$, $p\in E$. But as $E$ is spacelike and two-dimensional, we can assume without loss of generality that the metric of $T_pE$ is the Euclidean metric, and therefore has the two-dimensional Euclidean group ${\rm E}(2)$ as its isometry group. Thus $N\in\GL(2,\Rl)\cap{\rm E}(2)={\rm O}(2)$, {\it i.e.} $\det N=\pm1$.
\end{proof}
%===========================================================

The $\GL(2,\Rl)$-transformation given by the flip matrix $\Pi:=\left(0\;\;1\atop1\;\;0\right)$ will play a special role later on. We therefore reserve the name {\em inverted Killing pair} of $\xi=(\xi_1,\xi_2)\in\Xi$ for
\begin{align}\label{def:xiprime}
 \xi'&:=\Pi\xi=(\xi_2,\xi_1)\,.
\end{align}
Note that since we consider ordered tuples $(\xi_1,\xi_2)$, the Killing pairs $\xi$ and $\xi'$ are not identical. Clearly, the map $\xi\mapsto\xi'$ is an involution on $\Xi$, {\it i.e.} $(\xi')'=\xi$.
\\
\\
After these preparations, we turn to the construction of wedge regions in $M$, and begin by specifying their {\em edges}.
%===========================================================
\begin{definition}\label{edgedef}
An edge is a subset of $M$ which has the form
\begin{align}\label{def:edge}
 E_{\xi,p}:=\{\flow_{\xi,s}(p)\in M\,:\,s\in\Rl^2\}
\end{align}
for some $\xi\in\Xi$, $p\in M$. Any spacelike vector $n_{\xi,p}\in T_pM$ which completes the gradient of the chosen temporal function and the Killing vectors $\xi_1(p),\xi_2(p)$ to a positively oriented basis $(\nabla\mathcal{T}(p),\xi_1(p),\xi_2(p),n_{\xi,p})$ of $T_pM$ is called an oriented normal of $E_{\xi,p}$.
\end{definition}
%===========================================================

It is clear from this definition that each edge is a two-dimensional, spacelike, smooth submanifold of $M$. Our definition of admissible spacetimes $M\cong \Rl\times I\times E$ explicitly restricts the topology of $I$, but not of the edge \eqref{edge1}, which can be homeomorphic to a plane, cylinder, or torus.

Note also that the description of the edge $E_{\xi,p}$ in terms of $\xi$ and $p$ is somewhat redundant: Replacing the Killing fields $\xi_1,\xi_2$ by linear combinations $\xiti:=N\xi$, $N\in\GL(2,\Rl)$, or replacing $p$ by $\pti:=\varphi_{\xi,u}(p)$ with some $u\in\Rl^2$, results in the same manifold $E_{\xiti,\pti}=E_{\xi,p}$.
\\\\
Before we define wedges as connected components of causal complements of edges, we have to prove the following key lemma, from which the relevant properties of wedges follow. For its proof, it might be helpful to visualize the geometrical situation as sketched in Figure \ref{fig:wedge}.

%===========================================================
\begin{figure}[htbp]
  \centering
    \includegraphics[width=0.5\textwidth]{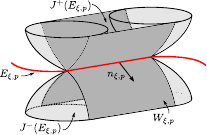}\\
  \caption{\em Three-dimensional sketch of the wedge $W_{\xi,p}$ and its edge $E_{\xi,p}$}
  \label{fig:wedge}
\end{figure}
%===========================================================

%===========================================================
\begin{lemma}\label{edgeprop}
 The causal complement $E_{\xi,p}'$ of an edge $E_{\xi,p}$ is the disjoint union of two connected components, which are causal complements of each other.
\end{lemma}
%===========================================================
\begin{proof}
We first show that any point $q\in E'_{\xi,p}$ is connected to the base point $p$ by a smooth, spacelike curve. Since $M$ is globally hyperbolic, there exist Cauchy surfaces $\Sigma_p,\Sigma_q$ passing through $p$ and $q$, respectively. We pick two compact subsets $K_q\subset \Sigma_q$, containing $q$, and $K_p\subset \Sigma_p$, containing $p$. If $K_p, K_q$ are chosen sufficiently small, their union $K_p\cup K_q$ is an acausal, compact, codimension one submanifold of $M$. It thus fulfills the hypothesis of Thm. 1.1 in \cite{BernalSanchez:2006}, which guarantees that there exists a spacelike Cauchy surface $\Sigma$ containing the said union. In particular, there exists a smooth, spacelike curve $\gamma$ connecting $p=\gamma(0)$ and $q=\gamma(1)$. Picking spacelike vectors $v\in T_p\Sigma$ and $w\in T_q\Sigma$, we have the freedom of choosing $\gamma$ in such a way that $\dot{\gamma}(0)=v$ and $\dot{\gamma}(1)=w$. If $v$ and $w$ are chosen linearly independent from $\xi_1(p),\xi_2(p)$ and $\xi_1(q),\xi_2(q)$, respectively, these vectors are oriented normals of $E_{\xi,p}$ respectively $E_{\xi,q}$, and we can select $\gamma$ such that it intersects the edge $E_{\xi,p}$ only in $p$.

Let us define the region
\newpage

\noindent
\begin{align}
\label{wedgechar}
W_{\xi,p}
&:=
\{q\in E_{\xi,p}':\;\exists\, \gamma\in C^1([0,1],M)\text{ with }\gamma(0)=p, \gamma(1)=q, E_{\xi,p}\cap\gamma=\{p\},\nonumber\\
&\qquad
\dot{\gamma}(0) \text{ is an oriented normal of } E_{\xi,p},\,\dot{\gamma}(1) \text{ is an oriented normal of } E_{\xi,q}
\}\,,
\end{align}
and, exchanging $\xi$ with the inverted Killing pair $\xi'$, we correspondingly define the region $W_{\xi',p}$. It is clear from the above argument that $W_{\xi,p}\cup W_{\xi',p}=E_{\xi,p}'$, and that we can prescribe arbitrary normals $n,m$ of $E_{\xi,p}$, $E_{\xi,q}$ as initial respectively final tangent vectors of the curve $\gamma$ connecting $p$ to $q\in W_{\xi,p}$.

The proof of the lemma consists in establishing that $W_{\xi,p}$ and $W_{\xi',p}$ are disjoint, and causal complements of each other. To prove disjointness of $W_{\xi,p}, W_{\xi',p}$, assume there exists a point $q\in W_{\xi,p}\cap W_{\xi',p}$. Then $q$ can be connected with the base point $p$ by two spacelike curves, whose tangent vectors satisfy the conditions in \eqref{wedgechar} with $\xi$ respectively $\xi'$. By joining these two curves, we have identified a continuous loop $\la$ in $E_{\xi,p}'$. As an oriented normal, the tangent vector $\dot{\la}(0)$ at $p$ is linearly independent of $\xi_1(p), \xi_2(p)$, so that $\la$ intersects $E_{\xi,p}$ only in $p$.

Recall that according to Definition \ref{admissible}, $M$ splits as the product $M\cong \Rl\times I\times E_{\xi,p}$, with an open interval $I$ which is smoothly embedded in $M$. Hence we can consider the projection $\pi(\la)$ of the loop $\la$ onto $I$, which is a closed interval $\pi(\la)\subset I$ because the simple connectedness of $I$ rules out the possibility that $\pi(\la)$ forms a loop, and on account of the linear independence of $\{\xi_1(p),\xi_2(p),n_{\xi,p}\}$, the projection cannot be just a single point. Yet, as $\la$ is a loop, there exists $p'\in\la$ such that $\pi(p')=\pi(p)$. We also know that $\pi^{-1}(\{\pi(p)\})=\Rl\times\{\pi(p)\}\times E_{\xi,p}$ is contained in $J^+(E_{\xi,p})\cup E_{\xi,p}\cup J^-(E_{\xi,p})$ and, since $p$ and $p'$ are causally separated, the only possibility left is that they both lie on the same edge. Yet, per construction, we know that the loop intersects the edge only once at $p$ and, thus, $p$ and $p'$ must coincide, which is the sought contradiction.

To verify the claim about causal complements, assume there exist points $q\in W_{\xi,p}$, $q'\in W_{\xi',p}$ and a causal curve $\gamma$ connecting them, $\gamma(0)=q$, $\gamma(1)=q'$. By definition of the causal complement, it is clear that $\gamma$ does not intersect $E_{\xi,p}$. In view of our restriction on the topology of $M$, it follows that $\gamma$ intersects either $J^+(E_{\xi,p})$ or $J^-(E_{\xi,p})$. These two cases are completely analogous, and we consider the latter one, where there exists a point $q''\in \gamma\cap J^-(E_{\xi,p})$. In this situation, we have a causal curve connecting $q\in W_{\xi,p}$ with $q''\in J^-(E_{\xi,p})$, and since $q\notin J^-(q'')\subset J^-(E_{\xi,p})$, it follows that $\gamma$ must be past directed. As the time orientation of $\gamma$ is the same for the whole curve, it follows that also the part of $\gamma$ connecting $q''$ and $q'$ is past directed. Hence $q'\in J^-(q'')\subset J^-(E_{\xi,p})$, which is a contradiction to $q'\in W_{\xi',p}$. Thus $W_{\xi',p}\subset {W_{\xi,p}}'$.

To show that $W_{\xi',p}$ coincides with ${W_{\xi,p}}'$, let $q\in {W_{\xi,p}}'=\overline{W_{\xi,p}}'\subset E_{\xi,p}' = W_{\xi,p}\sqcup W_{\xi',p}$. Yet $q\in W_{\xi,p}$ is not possible since $q\in {W_{\xi,p}}'$ and $W_{\xi,p}$ is open. So $q\in W_{\xi',p}$, {\em i.e.} we have shown ${W_{\xi,p}}'\subset W_{\xi',p}$, and the claimed identity $W_{\xi',p} = {W_{\xi,p}}'$ follows.
\end{proof}
%===========================================================
\newpage 

\noindent
\vspace{-0.8cm}
\begin{wrapfigure}{r}{46mm}
 \centering
{\includegraphics{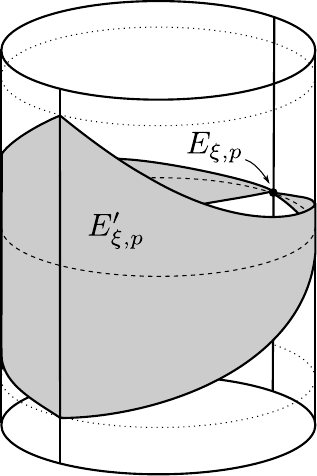}}\\
{\em \small $E_{\xi,p}'$ in a Lorentz cylinder}
\end{wrapfigure}

Lemma \ref{edgeprop} does not hold if the topological requirements on $M$ are dropped. As an example, consider a cylinder universe $\Rl\times S^1\times\Rl^2$, the product of the Lorentz cylinder $\Rl\times S^1$ \cite{O'Neill:1983} and the Euclidean plane $\Rl^2$. The translations in the last factor $\Rl^2$ define spacelike, complete, commuting, linearly independent Killing fields $\xi$.  Yet the causal complement of the edge $E_{\xi,p}=\{0\}\times\{1\}\times\Rl^2$ has only a single connected component, which has empty causal complement. In this situation, wedges lose many of the useful properties which we establish below for admissible spacetimes.
\\\\
In view of Lemma \ref{edgeprop}, wedges in $M$ can be defined as follows.

%===========================================================
\label{wedgecosmological}
\begin{definition}{\bf (Wedges)}\\
A wedge is a subset of $M$ which is a connected component of the causal complement of an edge in $M$. Given $\xi\in\Xi$, $p\in M$, we denote by $W_{\xi,p}$ the component of $E_{\xi,p}'$ which intersects the curves $\gamma(t) := \exp_p(t\,n_{\xi,p})$, $t>0$, for any oriented normal $n_{\xi,p}$ of $E_{\xi,p}$. The family of all wedges is denoted
\begin{align}
  \W
:=
\{W_{\xi,p}\,:\, \xi\in\Xi,\,p\in M\}\,.
\end{align}
\end{definition}
%===========================================================

As explained in the proof of Lemma \ref{edgeprop}, the condition that the curve $\Rl^+\ni t\mapsto \exp_p(t\,n_{\xi,p})$ intersects a connected component of $E_{\xi,p}'$ is independent of the chosen normal $n_{\xi,p}$, and each such curve intersects precisely one of the two components of $E_{\xi,p}'$.

Some properties of wedges which immediately follow from the construction carried out in the proof of Lemma \ref{edgeprop} are listed in the following proposition.

%===========================================================
\begin{proposition}{\bf (Properties of wedges)}\label{prop:wedge-properties}
\\
 Let $W=W_{\xi,p}$ be a wedge. Then
\begin{enumerate}
 \item $W$ is causally complete, {\it i.e.} $W''=W$, and hence globally hyperbolic.
\item The causal complement of a wedge is given by inverting its Killing pair,
\begin{align}\label{cc-inverted}
 (W_{\xi,p})' = W_{\xi',p}\,.
\end{align}
 \item A wedge is invariant under the Killing flow generating its edge,
\begin{align}
\flow_{\xi,s}(W_{\xi,p})=W_{\xi,p}\,,\qquad s\in\Rl^2\,.
\end{align}
\end{enumerate}
\end{proposition}
%===========================================================
\begin{proof}
{\em a)} By Lemma \ref{edgeprop}, $W$ is the causal complement of another wedge $V$, and therefore causally complete: $W''=V'''=V'=W$. Since $M$ is globally hyperbolic, this implies that $W$ is globally hyperbolic, too \cite[Prop.12.5]{Keyl:1996}.

{\em b)} This statement has already been checked in the proof of Lemma \ref{edgeprop}.

{\em c)} By definition of the edge $E_{\xi,p}$ \eqref{def:edge}, we have $\flow_{\xi,s}(E_{\xi,p})= E_{\xi,p}$ for any $s\in\Rl^2$, and since the $\flow_{\xi,s}$ are isometries, it follows that $\flow_{\xi,s}(E_{\xi,p}')= E_{\xi,p}'$. Continuity of the flow implies that also the two connected components of this set are invariant.
\end{proof}
%===========================================================

%===========================================================
\begin{corollary}{\bf (Properties of the family of wedge regions)}\\
 The family $\W$ of wedge regions is invariant under the isometry group $\iso$ and under taking causal complements. For $h\in \iso$, it holds
\begin{align}\label{h-W}
 h(W_{\xi,p})=W_{h_*\xi,h(p)}\,.
\end{align}
\end{corollary}
%===========================================================
\begin{proof}
Since isometries preserve the causal structure of a spacetime, we only need to look at the action of  isometries on edges. We find
\begin{equation}
hE_{\xi,p}
=\{h\circ\flow_{\xi,s}\circ h^{-1}(h(p)):s\in\Rl^2\}
=\{\flow_{h_*\xi,s}(h(p)):s\in\Rl^2\}
=E_{h_*\xi,h(p)}
\end{equation}
by using the well-known fact that conjugation of flows by isometries amounts to the push-forward by the isometry of the associated vector field. Since $h_*\xi\in\Xi$ for any $\xi\in\Xi$, $h\in \iso$ (Lemma \ref{lemma:group-xi}), the family $\W$ is invariant under the action of the isometry group. Closedness of $\W$ under causal complementation is clear from Prop. \ref{prop:wedge-properties} b).
\end{proof}
%===========================================================

In contrast to the situation in flat spacetime, the isometry group $\iso$ does {\em not} act transitively on $\W(M,\gST)$ for generic admissible $M$, and there is no isometry mapping a given wedge onto its causal complement. This can be seen explicitly in the examples discussed in Section \ref{sec:examples}. To keep track of this structure of $\W(M,\gST)$, we decompose $\Xi(M,\gST)$ into orbits under the $\iso$- and $\GL(2,\Rl)$-actions.

%===========================================================
\begin{definition}\label{def:equivalence}
Two Killing pairs $\xi,\xiti\in\Xi$ are equivalent, written $\xi\sim\xiti$, if there exist $h\in\iso$ and $N\in\GL(2,\Rl)$ such that $\xiti = N h_* \xi$.
\end{definition}
%===========================================================
As $\xi\mapsto N\xi$ and $\xi\mapsto h_*\xi$ are commuting group actions, $\sim$ is an equivalence relation. According to Lemma \ref{lemma:group-xi} and Prop. \ref{prop:wedge-properties} b), c), acting with $N\in\GL(2,\Rl)$ on $\xi$ either leaves $W_{N\xi,p}=W_{\xi,p}$ invariant (if $\det N>0$) or exchanges this wedge with its causal complement, $W_{N\xi,p}=W_{\xi,p}'$ (if $\det N<0$). Therefore the ``coherent''\footnote{See \cite{BuchholzSummers:2007} for a related notion on Minkowski space.} subfamilies arising in the decomposition of the family of all wedges along the equivalence classes $[\xi]\in\Xi\slash\!\!\sim$,
\begin{align}\label{W-decomposition}
 \W = \bigsqcup_{[\xi]}\W_{[\xi]}\,,
\qquad
\W_{[\xi]}
:=
\{W_{\xiti,p}\,:\,\xiti\sim\xi,\,p\in M\}\,,
\end{align}
take the form
\begin{align}\label{def:Wxi}
\W_{[\xi]}
=
\{W_{h_*\xi,p},\,W_{h_*\xi,p}'\,:\,h\in\iso,\,p\in M\}
\,.
\end{align}
In particular, each subfamily $\W_{[\xi]}$ is invariant under the action of the isometry group and causal complementation.

In our later applications to quantum field theory, it will be important to have control over causal configurations $W_1\subset W_2'$ and inclusions $W_1\subset W_2$ of wedges $W_1,W_2\in\W$. Since $\W$ is closed under taking causal complements, it is sufficient to consider inclusions. Note that the following proposition states in particular that inclusions can only occur between wedges from the same coherent subfamily $\W_{[\xi]}$.

%===========================================================
\begin{proposition}{\bf (Inclusions of wedges).}\label{prop:inclusions}\\
Let $\xi,\xiti\in\Xi$, $p,\pti\in M$. The wedges $W_{\xi,p}$ and $W_{\xiti,\pti}$ form an inclusion, $W_{\xi,p}\subset W_{\xiti,\pti}$, if and only if $p\in \overline{W_{\xiti,\pti}}$ and there exists $N\in\GL(2,\Rl)$ with $\det N>0$, such that $\xiti=N\xi$.
\end{proposition}
%===========================================================
\begin{proof}
($\Leftarrow$) Let us assume that $\xiti=N\xi$ holds for some $N\in\GL(2,\Rl)$ with $\det N>0$,  and $p\in \overline{W_{\xiti,\pti}}$. In this case, the Killing fields in $\xiti$ are linear combinations of those in $\xi$, and consequently, the edges $E_{\xi,p}$ and $E_{\xiti,\pti}$ intersect if and only if they coincide, {\it i.e.} if $\pti\in E_{\xi,p}$. If the edges coincide, we clearly have $W_{\xiti,\pti}=W_{\xi,p}$. If they do not coincide, it follows from $p\in \overline{W_{\xiti,\pti}}$ that $E_{\xi,p}$ and $E_{\xiti,\pti}$ are either spacelike separated or they can be connected by a null geodesic.

Consider now the case that $E_{\xi,p}$ and $E_{\xiti,\pti}$ are spacelike separated, {\it i.e.} $p\in W_{\xiti,\pti}$. Pick a point $q\in W_{\xi,p}$ and recall that $W_{\xi,p}$ can be characterized by equation (\ref{wedgechar}). Since  $p\in W_{\xiti,\pti}$ and $q\in W_{\xi,p}$, there exist curves $\gamma_p$ and $\gamma_q$, which connect the pairs of points $(\pti,p)$ and $(p,q)$, respectively, and comply with the conditions in (\ref{wedgechar}). By joining $\gamma_{p}$ and $\gamma_q$ we obtain a curve which connects $\pti$ and $q$. The tangent vectors $\dot{\gamma}_p(1)$ and $\dot{\gamma}_q(0)$ are oriented normals of $E_{\xi, p}$ and we choose $\gamma_p$ and $\gamma_q$ in such a way that these tangent vectors coincide. Due to the properties of $\gamma_p$ and $\gamma_q$, the joint curve also complies with the conditions in (\ref{wedgechar}), from which we conclude $q\in W_{\xiti,\pti}$, and thus $W_{\xi,p}\subset W_{\xiti,\pti}$.

Consider now the case that $E_{\xiti,\pti}$ and $E_{\xi,p}$ are connected by null geodesics, {\it i.e.} $p\in \partial\overline{W_{\xiti,\pti}}$. Let $r$ be the point in $E_{\xi,p}$ which can be connected by a null geodesic with $\pti$ and pick a point $q\in W_{\xi,p}$. The intersection $J^-(r)\cap \partial\overline{W_{\xi,p}}$ yields another null curve, say $\mu$, and the intersection $\mu\cap J^-(q)=:p'$ is non-empty since $r$ and $q$ are spacelike separated and $q\in W_{\xi,p}$. The null curve $\mu$ is chosen future directed and parametrized in such a way that $\mu(0)=p'$ and $\mu(1)=r$. By taking $\varepsilon\in(0,1)$ we find $q\in W_{\xi,\mu(\varepsilon)}$ and $\mu(\varepsilon)\in W_{\xiti,\pti}$ which entails $q\in W_{\xiti,\pti}$.

($\Rightarrow$) Let us assume that we have an inclusion of wedges $W_{\xi,p} \subset W_{\xiti,\pti}$. Then clearly $p\in\overline{W_{\xiti,\pti}}$. Since $M$ is four-dimensional and $\xi_1,\xi_2,\xiti_1,\xiti_2$ are all spacelike, they cannot be linearly independent. Let us first assume that three of them are linearly independent, and without loss of generality, let $\xi=(\xi_1,\xi_2)$ and $\xiti=(\xi_2, \xi_3)$ with three linearly independent spacelike Killing fields $\xi_1,\xi_2,\xi_3$. Picking points $q\in E_{\xi,p}$, $\tilde{q}\in E_{\xiti,\pti}$ these can be written as $q=(t,x_1,x_2,x_3)$ and $\tilde{q}=(\tilde{t},\tilde{x}_1,\tilde{x}_2,\tilde{x}_3)$ in the global coordinate system of flow parameters constructed from $\xi_1,\xi_2,\xi_3$ and the gradient of the temporal function.

For suitable flow parameters $s_1,s_2,s_3$, we have $\flow_{\xi_1,s_1}(q)=(t,\tilde{x}_1,x_2,x_3)=:q'\in E_{\xi,p}$ and $\flow_{(\xi_2,\xi_3),(s_2,s_3)}(\tilde{q})=(\tilde{t},\tilde{x}_1,x_2,x_3)=:\tilde{q}'\in E_{\xiti,\pti}$. Clearly, the points $q'$ and $\tilde{q}'$ are  connected by a timelike curve, {\it e.g.} the curve whose tangent vector field is given by the gradient of the temporal function. But a timelike curve connecting the edges of $W_{\xi,p}, W_{\xiti,\pti}$ is a contradiction to these wedges forming an inclusion. So no three of the vector fields $\xi_1,\xi_2,\xiti_1,\xiti_2$ can be linearly independent.

Hence $\xiti= N\xi$ with an invertible matrix $N$. It remains to establish the correct sign of $\det N$, and to this end, we assume $\det N<0$. Then we have $(W_{\xi,p})'=W_{\xi',p}\subset W_{\xiti,\pti}$, by  (Prop. \ref{prop:wedge-properties} b)) and the ($\Leftarrow$) statement in this proof, since $\xiti$ and $\xi'$ are related by a positive determinant transformation and $p\in \overline{W_{\xiti,\pti}}$. This yields that both, $W_{\xi,p}$ and its causal complement, must be contained in $W_{\xiti,\pti}$, a contradiction. Hence $\det N>0$, and the proof is finished.
\end{proof}
%===========================================================

Having derived the structural properties of the set $\W$ of wedges needed later, we now compare our wedge regions to the Minkowski wedges and to other definitions proposed in the literature.

The flat Minkowski spacetime $(\Rl^4,\eta)$ clearly belongs to the class of admissible spacetimes, with translations along spacelike directions and rotations in the standard time zero Cauchy surface as its complete spacelike Killing fields. However, as Killing pairs consist of non-vanishing vector fields, and each rotation leaves its rotation axis invariant, the set $\Xi(\Rl^4,\eta)$ consists precisely of all pairs $(\xi_1,\xi_2)$ such that the flows $\flow_{\xi_1}$, $\flow_{\xi_2}$ are translations along two linearly independent spacelike directions. Hence the set of all edges in Minkowski space coincides with the set of all two-dimensional spacelike planes. Consequently, each wedge $W\in\W(\Rl^4,\eta)$ is bounded by two non-parallel characteristic three-dimensional planes. This is precisely the family of wedges usually considered in Minkowski space\footnote{Note that we would get a ``too large'' family of wedges in Minkowski space if we would drop the requirement that the vector fields generating edges are Killing. However, the assumption that edges are generated by {\em commuting} Killing fields is motivated by the application to deformations of quantum field theories, and one could generalize our framework to spacetimes with edges generated by complete, linearly independent smooth Killing fields.} (see, for example, \cite{ThomasWichmann:1997}).

Besides the features we established above in the general admissible setting, the family of Minkowski wedges has the following well-known properties:

\begin{enumerate}
 \item Each wedge $W\in\W(\Rl^4,\eta)$ is the causal completion of the world line of a uniformly accelerated observer.
 \item Each wedge $W\in\W(\Rl^4,\eta)$ is the union of a family of double cones whose tips lie on two fixed lightrays.
\item The isometry group (the Poincar\'e group) acts transitively on $\W(\Rl^4,\eta)$.
 \item $\W(\Rl^4,\eta)$ is {\em causally separating}\label{causal-sep} in the sense that given
any two spacelike separated double cones $\OO_1,\OO_2\subset\Rl^4$, then there exists a wedge $W$ such that $\OO_1\subset
W\subset\OO_2'$ \cite{ThomasWichmann:1997}. $\W(\Rl^4,\eta)$ is a subbase for the topology of $\Rl^4$.
\end{enumerate}

All these properties a)--d) do {\em not} hold for the class $\W(M,\gST)$ of wedges on a {\em general} admissible spacetime, but some hold for certain subclasses, as can be seen from the explicit examples in the subsequent section.

There exist a number of different constructions for wedges in curved spacetimes in the literature, mostly for special manifolds. On de Sitter respectively anti de Sitter space Borchers and Buchholz \cite{BorchersBuchholz:1999} respectively Buchholz and Summers \cite{BuchholzSummers:2004-2} construct wedges by taking property a) as their defining feature, see also the generalization by Strich \cite{Strich:2008}. In the de Sitter case, this definition is equivalent to our definition of a wedge as a connected component of the causal complement of an edge \cite{BuchholzMundSummers:2001}. But as two-dimensional spheres, the de Sitter edges do not admit two linearly independent commuting Killing fields. Apart from this difference due to our restriction to commuting, linearly independent, Killing fields, the de Sitter wedges can be constructed  in the same way as presented here. Thanks to the maximal symmetry of the de Sitter and anti de Sitter spacetimes, the respective isometry groups act transitively on the corresponding wedge families c), and causally separate in the sense of d).

A definition related to the previous examples has been given by Lauridsen-Ribeiro for wedges in asymptotically anti de Sitter spacetimes (see Def. 1.5 in \cite{LauridsenRibeiro:2007}). Note that these spacetimes are not admissible in our sense since anti de Sitter space is not globally hyperbolic.

Property b) has recently been taken by Borchers \cite{Borchers:2009} as a definition of wedges in a quite general class of curved spacetimes which is closely related to the structure of double cones. In that setting, wedges do not exhibit in general all of the features we derived in our framework, and can for example have compact closure.

Wedges in a class of Friedmann-Robertson-Walker spacetimes with spherical spatial sections have been constructed with the help of conformal embeddings into de Sitter space \cite{BuchholzMundSummers:2001}. This construction also yields wedges defined as connected components of causal complements of edges. Here a) does not, but c) and d) do hold, see also our discussion of Friedmann-Robertson-Walker spacetimes with flat spatial sections in the next section.

The idea of constructing wedges as connected components of causal complements of specific two-dimensional submanifolds has also been used in the context of globally hyperbolic spacetimes with a bifurcate Killing horizon \cite{GuidoLongoRobertsVerch:2001}, building on earlier work in \cite{Kay:1985}. Here the edge is given as the fixed point manifold of the Killing flow associated with the horizon.

%===========================================================
 \subsection{Concrete examples}\label{sec:examples}
%===========================================================

In the previous section we provided a complete but abstract characterization of the geometric structures of the class of spacetimes we are interested in. This analysis is now complemented by presenting a number of explicit examples of admissible spacetimes.

The easiest way to construct an admissible spacetime is to take the warped product \cite[Chap.7]{O'Neill:1983} of an edge with another manifold. Let $(E,\gST_E)$ be a two-dimensional Riemannian manifold endowed with two complete, commuting, linearly independent, smooth Killing fields, and let $(X,\gST_X)$ be a two-dimensional, globally hyperbolic spacetime diffeomorphic to $\Rl\times I$, with $I$ an open interval or the full real line. Then, given a positive smooth function $f$ on $X$, consider the {\em warped product}
$M:=X\times_f E$, {\em i.e.} the product manifold $X\times E$ endowed with the metric tensor field
\begin{align*}
\gST:=\pi_X^*(\gST_X)+(f\circ\pi_X) \cdot \pi_E^*(\gST_E),
\end{align*}
where $\pi_X:M\ra X$ and $\pi_E:M\ra E$ are the projections on $X$ and $E$. It readily follows that $(M,\gST)$ is admissible in the sense of Definition \ref{admissible}.
\\\\
The following proposition describes an explicit class of admissible spacetimes in terms of their metrics.
%=================================================
\begin{proposition}\label{prop:metric}
Let $(M,\gST)$ be a spacetime diffeomorphic to $\Rl\times I\times \Rl^2$, where $I\subseteq\Rl$ is open and simply connected, endowed with a global coordinate system $(t,x,y,z)$ according to which the metric reads
\begin{equation}\label{metric}
 \gST=e^{2f_0}dt^2-e^{2f_1}dx^2-e^{2f_2}dy^2-e^{2f_3}(dz-q\,dy)^2.
\end{equation}
Here $t$ runs over the whole $\Rl$, $f_i,q\in C^\infty(M)$ for $i=0,...,3$ and $f_i,q$ do not depend on $y$ and $z$. Then $(M,\gST)$ is an admissible spacetime in the sense of Definition \ref{admissible}.
\end{proposition}
%=================================================
\begin{proof}
Per direct inspection of \eqref{metric}, $M$ is isometric to $\Rl\times\Sigma$ with $\Sigma\cong I\times\Rl^2$ with $\gST=\beta\, dt^2-h_{ij}dx^i dx^j$, where $\beta$ is smooth and positive, and $h$ is a metric which depends smoothly on $t$. Furthermore, on the hypersurfaces at constant $t$, $\det h=e^{2(f_1+f_2+f_3)}>0$ and $h$ is block-diagonal. If we consider the sub-matrix with $i,j=y,z$, this has a positive determinant and a positive trace. Hence we can conclude that the induced metric on $\Sigma$ is Riemannian, or, in other words, $\Sigma$ is a spacelike, smooth, three-dimensional Riemannian hypersurface. Therefore we can apply Theorem 1.1 in \cite{BernalSanchez:2005} to conclude that $M$ is globally hyperbolic.

Since the metric coefficients are independent from $y$ and $z$, the vector fields $\xi_1=\partial_y$ and $\xi_2=\partial_z$ are smooth Killing fields which commute and, as they lie tangent to the Riemannian hypersurfaces at constant time, they are also spacelike. Furthermore, since per definition of spacetime, $M$ and thus also $\Sigma$ is connected, we can invoke the Hopf-Rinow-Theorem \cite[Chap.5, Thm.21]{O'Neill:1983} to conclude that $\Sigma$ is complete and, thus, all its Killing fields are complete. As $I$ is simply connected by assumption, it follows that $(M,\gST)$ is admissible.
\end{proof}
%=================================================

Under an additional assumption, also a partial converse of Proposition \ref{prop:metric} is true. Namely, let $(M,\gST)$ be a globally hyperbolic spacetime with two complete, spacelike, commuting, smooth Killing fields, and pick a local coordinate system $(t,x,y,z)$, where $y$ and $z$ are the flow parameters of the Killing fields. Then, if the reflection map $r:M\ra M$, $r(t,x,y,z)=(t,x,-y,-z)$, is an isometry, the metric is locally of the form \eqref{metric}, as was proven in \cite{Chandrasekhar:1983, ChandrasekharFerrari:1984}. The reflection $r$ is used to guarantee the vanishing of the unwanted off-diagonal metric coefficients, namely those associated to ``$dx\,dy$" and ``$dx\,dz$". Notice that the cited papers allow only to establish a result on the local structure of $M$ and no a priori condition is imposed on the topology of $I$, in distinction to Proposition \ref{prop:metric}.
\\\\
Some of the metrics \eqref{metric} are used in cosmology. For the description of a spatially homogeneous but in general anisotropic universe $M\cong J\times\Rl^3$ where $J\subseteq\Rl$ (see chapter 5 in \cite{Wald:1984} and \cite{FullingParkerHu:1974}), one puts $f_0=q=0$ in \eqref{metric} and takes $f_1,f_2,f_3$ to depend only on $t$. This yields the metric of {\em Kasner spacetimes} respectively {\em Bianchi I models}\footnote{
The Bianchi models I--IX \cite{Ellis:2006} arise from the classification of three-dimensional real Lie algebras, thought of as Lie subalgebras of the Lie algebra of Killing fields. Only the cases  Bianchi I--VII, in which the three-dimensional Lie algebra contains $\Rl^2$ as a subalgebra, are of direct interest here, since only in these cases Killing pairs exist.}
 \begin{equation}\label{Kasner}
 \gST=dt^2-e^{2f_1}dx^2-e^{2f_2}dy^2-e^{2f_3}dz^2\,.
 \end{equation}
Clearly here the isometry group contains three smooth Killing fields, locally given by $\partial_x, \partial_y, \partial_z$, which are everywhere linearly independent, complete and commuting. In particular,   $(\partial_x,\partial_y)$, $(\partial_x,\partial_z)$ and $(\partial_y,\partial_z)$ are Killing pairs.
\\\\
A case of great physical relevance arises when specializing the metric further by taking all the functions $f_i$ in \eqref{Kasner} to coincide. In this case, the metric takes the so-called  {\em Friedmann-Robertson-Walker} form
 \begin{equation}\label{FRW}
 \gST
= dt^2-a(t)^2\,[dx^2+dy^2+dz^2]
=a(t(\tau))^2\,\left[d\tau^2-dx^2-dy^2-dz^2\right]
\,.
 \end{equation}
Here the {\em scale factor} $a(t) := e^{f_1(t)}$ is defined on some interval $J\subseteq\Rl$, and in the second equality, we have introduced the {\em conformal time} $\tau$, which is implicitly defined by $d\tau = a^{-1}(t)dt$. Notice that, as in the Bianchi I model, the manifold is $M\cong J\times\Rl^3$, {\em i.e.} the variable $t$ does not need to range over the whole real axis. (This does not affect the property of global hyperbolicity.)

By inspection of \eqref{FRW}, it is clear that the isometry group of this spacetime contains the three-dimensional Euclidean group ${\rm E}(3) = \Rl^3\rtimes{\rm O}(3)$. Disregarding the Minkowski case, where $J=\Rl$ and $a$ is constant, the isometry group in fact coincides with ${\rm E}(3)$. Edges in such a Friedmann-Robertson-Walker universe are of the form $\{\tau\}\times S$, where $S$ is a two-dimensional plane in $\Rl^3$ and $t(\tau)\in J$. Here $\W$ consists of a single coherent family, and the $\iso$-orbits in $\W$ are labelled by the time parameter $\tau$ for the corresponding edges. Also note that the family of Friedmann-Robertson-Walker wedges is causally separating in the sense discussed on page \pageref{causal-sep}.

The second form of the metric in \eqref{FRW} is manifestly a conformal rescaling of the flat Minkowski metric. Interpreting the coordinates $(\tau,x,y,z)$ as coordinates of a point in $\Rl^4$ therefore gives rise to a conformal embedding $\iota:M\hookrightarrow\Rl^4$.

%===========================================================
\begin{wrapfigure}{r}{60mm}
 \centering
{\includegraphics{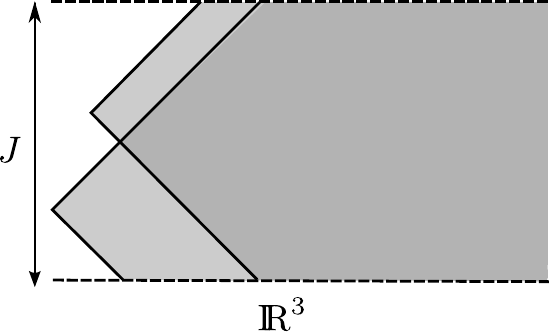}}\\
{\em Two wedges in FRW spacetime}
\end{wrapfigure}
%===========================================================
\noindent
In this situation, it is interesting to note that the set of all images $\iota(E)$ of edges $E$ in the Friedmann-Robertson-Walker spacetime coincides with the set of all Minkowski space edges which lie completely in $\iota(M)=J\times\Rl^3$, provided that $J$ does not coincide with $\Rl$. These are just the edges parallel to the standard Cauchy surfaces of constant $\tau$ in $\Rl^4$.
So Friedmann-Robertson-Walker edges can also be characterized in terms of Minkowski space edges and the conformal embedding $\iota$, analogous to the construction of wedges in Friedmann-Robertson-Walker spacetimes with spherical spatial sections in \cite{BuchholzMundSummers:2001}.

%===============================================
\section{Quantum field theories on admissible spacetimes}\label{sec:deformation}
%===============================================

Having discussed the relevant geometric structures, we now fix an admissible spacetime $(M,\gST)$ and discuss warped convolution deformations of quantum field theories on it. For models on flat Minkowski space, it is known that this deformation procedure weakens point-like localization to localization in wedges \cite{BuchholzLechnerSummers:2010}, and we will show here that the same holds true for admissible curved spacetimes. For a convenient description of this weakened form of localization, and for a straightforward application of the warped convolution technique, we will work in the framework of local quantum physics (see Section \ref{sec:AQFT}).

In this setting, a model theory is defined by a net of field algebras, and here we consider algebras $\fF(W)$ of quantum fields supported in wedges $W\in\W(M,\gST)$. The main idea underlying the deformation is to apply the formalism developed in \cite{BuchholzSummers:2008,BuchholzLechnerSummers:2010}, but with the global translation symmetries of Minkowski space replaced by the Killing flow $\flow_\xi$ corresponding to the wedge $W=W_{\xi,p}$ under consideration. In the case of Minkowski spacetime, these deformations reduce to the familiar structure of a noncommutative Minkowski space with commuting time.

The details of the model under consideration will not be important in Section \ref{sec:generaldeformations}, since our construction relies only on a few structural properties satisfied in any well-behaved quantum field theory. In Section \ref{sec:dirac}, the deformed Dirac quantum field  is presented as a particular example.

%===============================================
\subsection{Deformations of nets with Killing symmetries}\label{sec:generaldeformations}
%===============================================

Proceeding to the standard mathematical formalism \cite{Haag:1996,Araki:1999}, we consider a C$^*$-algebra $\fF$, whose elements are interpreted as (bounded functions of) quantum fields on the spacetime $M$. The field algebra $\fF$ has a local structure, and in the present context, we focus on localization in wedges $W\in\W$, since this form of localization turns out to be stable under the deformation. Therefore, corresponding to each wedge $W\in\W$, we consider the C$^*$-subalgebra $\fF(W)\subset\fF$ of fields supported in $W$. Furthermore, we assume there exists a strongly continuous action $\alpha$ of the isometry group $\iso$ of $(M,\gST)$ on $\fF$, and a Bose/Fermi automorphism $\gamma$ whose square is the identity automorphism, and which commutes with $\alpha$. This automorphism will be used to separate the Bose/Fermi parts of fields $F\in\fF$; in the model theory of a free Dirac field discussed later, it can be chosen as a rotation by $2\pi$ in the Dirac bundle.

To allow for a straightforward application of the results of \cite{BuchholzLechnerSummers:2010}, we will also assume in the following that the field algebra is concretely realized on a separable Hilbert space $\Hil$, which carries a unitary representation $U$ of $\iso$ implementing the action $\alpha$, {\em i.e.}
\begin{align*}
U(h)FU(h)^{-1} = \alpha_h(F)
\,,\qquad
h\in\iso,\,F\in\fF\,.
\end{align*}
We emphasize that despite working on a Hilbert space, we do not select a state, since we do not make any assumptions regarding $U$-invariant vectors in $\Hil$ or the spectrum of subgroups of the representation $U$.\footnote{Note that every C$^*$-dynamical system $(\A,\alpha,G)$, where $\A\subset \mathcal{B(H)}$ is a concrete C$^*$-algebra on a separable Hilbert space $\mathcal{H}$ and $\alpha:G\ra \Aut(\A)$ is a strongly continuous representation of a locally compact group $G$, has a covariant representation \cite[Prop.7.4.7, Lem.7.4.9]{Pedersen:1979}, build out of the left-regular representation on the Hilbert space $L^2(G)\otimes \mathcal{H}$.} The subsequent analysis will be carried out in a C$^*$-setting, without using the weak closures of the field algebras $\fF(W)$ in $\B(\Hil)$.

For convenience, we also require the Bose/Fermi automorphism $\gamma$ to be unitarily implemented on $\Hil$, {\em i.e.} there exists a unitary $V=V^*=V^{-1}\in\B(\Hil)$ such that $\gamma(F)=VFV$. We will also use the associated unitary twist operator
\begin{align}\label{def:Z}
 Z := \frac{1}{\sqrt{2}}(1-iV)
\,.
\end{align}
Clearly, the unitarily implemented $\alpha$ and $\gamma$ can be continued to all of $\B(\Hil)$. By a slight abuse of notation, these extensions will be denoted by the same symbols.

In terms of the data $\{\fF(W)\}_{W\in\W}, \alpha, \gamma$, the structural properties of a quantum field theory on $M$ can be summarized as follows \cite{Haag:1996,Araki:1999}.

%===============================================
\begin{enumerate}
\label{fieldnetcosmological}
 \item {\em Isotony:} $\fF(W)\subset\fF(\tilde{W})$ whenever $W\subset\tilde{W}$.
 \item {\em Covariance} under $\iso$:
\begin{align}
 \alpha_h(\fF(W)) = \fF(h W)\,,\qquad h\in \iso,\;W\in\W\,.\label{covariance}
\end{align}
\item {\em Twisted Locality:} With the unitary $Z$ \eqref{def:Z},  there holds
\begin{align}\label{twisted-locality}
[ZFZ^*,\,G]=0,\qquad  F\in\fF(W), G\in\fF(W'),\;W\in\W\,.
\end{align}
\end{enumerate}
%===============================================

The twisted locality condition \eqref{twisted-locality} is equivalent to normal commutation relations between the Bose/Fermi parts $F_\pm:=\frac{1}{2}(F\pm \gamma(F))$ of fields in spacelike separated wedges, $[F_+,G_\pm]=[F_\pm,G_+]=\{F_-,G_-\}=0$ for $F\in\fF(W), G\in\fF(W')$ \cite{DoplicherHaagRoberts:1969}.

The covariance requirement \eqref{covariance} entails that for any Killing pair $\xi\in\Xi$, the algebra $\fF$ carries a corresponding $\Rl^2$-action $\tau_\xi$, defined by
\begin{align*}
  \tau_{\xi,s}:=\alpha_{\flow_{\xi,s}}=\text{ad}\,U_\xi(s)\,,\quad s\in\Rl^2\,,
\end{align*}
where $U_\xi(s)$ is shorthand for $U(\flow_{\xi,s})$. Since a wedge of the form $W_{\xi,p}$ with some $p\in M$ is invariant under the flows $\varphi_{N\xi,s}$ for any $N\in\GL(2,\Rl)$ (see Proposition \ref{prop:wedge-properties} c) and Lemma \ref{lemma:group-xi}), we have in view of covariance
\begin{align*}
 \tau_{N\xi,s}(\fF(W_{\xi,p}))
=
\fF(W_{\xi,p})
\,,\qquad
N\in\GL(2,\Rl),\; s\in\Rl^2\,.
\end{align*}

In this setting, all structural elements necessary for the application of warped convolution deformations \cite{BuchholzLechnerSummers:2010} are present, and we will use this technique to define a deformed net $W\longmapsto\fF(W)_\kappa$ of C$^*$-algebras on $M$, depending on a deformation parameter $\kappa\in\Rl$. For $\kappa=0$, we will recover the original theory, $\fF(W)_0=\fF(W)$, and for each $\kappa\in\Rl$, the three basic properties a)--c) listed above will remain valid. To achieve this, the elements of $\fF(W)$ will be deformed with the help of the Killing flow leaving $W$ invariant. We begin by recalling some definitions and results from \cite{BuchholzLechnerSummers:2010}, adapted to the situation at hand.
\\\\
Similar to the Weyl product appearing in the quantization of classical systems, the warped convolution deformation is defined in terms of oscillatory integrals of $\fF$-valued functions, and we have to introduce the appropriate smooth elements first. The action $\alpha$ is a strongly continuous representation of the Lie group $\iso$, which acts automorphically and thus isometrically on the C$^*$-algebra $\fF$. In view of these properties, the smooth elements $\fF^\infty:=\{F\in\fF\,:\,\iso\ni h\mapsto\alpha_h(F)\;\text{is } \|\cdot\|_\fF\text{-smooth}\}$ form a norm-dense $*$-subalgebra $\fF^\infty\subset\fF$ (see, for example, \cite{Taylor:1986}). However, the subalgebras $\fF(W_{\xi,p})\subset\fF$ are in general only invariant under the $\Rl^2$-action $\tau_\xi$, and we therefore also introduce a weakened form of smoothness. An operator $F\in\fF$ will be called $\xi${\em -smooth} if
\begin{align}
 \Rl^2\ni s\mapsto \tau_{\xi,s}(F)\in\fF
\end{align}
is smooth in the norm topology of $\fF$. On the Hilbert space level, we have a dense domain\label{smoothvectors} $\Hil^\infty:=\{\Psi\in\Hil\,:\,\iso\ni h\mapsto U(h)\Psi\;\text{is } \|\cdot\|_\Hil\text{-smooth}\}$ of smooth vectors in $\Hil$.

As further ingredients for the definition of the oscillatory integrals, we pick a smooth, compactly supported ``cutoff'' function $\chi\in C_0^\infty(\Rl^2\times\Rl^2)$ with $\chi(0,0)=1$, and the standard antisymmetric $2\times2$ matrix
\begin{align}\label{def:Q}
 \theta:=
\left(
\begin{array}{cc}
 0&1\\-1&0
\end{array}
\right)
\,.
\end{align}
With these data, we associate to a $\xi$-smooth $F\in\fF$ the deformed operator ({\em warped convolution}) \cite{BuchholzLechnerSummers:2010}
\begin{align}\label{def:Ala}
F_{\xi,\kappa}
:=
\frac{1}{4\pi^2}
\lim_{\eps\ra0}
\int ds\,ds'\,e^{-iss'}\chi(\eps s,\eps s')\,U_\xi(\kappa\theta s)FU_\xi(s'-\kappa\theta s)
\,,
\end{align}
where $\kappa$ is a real parameter, and $ss'$ denotes the standard Euclidean inner product of $s,s'\in\Rl^2$. The above limit exists in the strong operator topology of $\B(\Hil)$ on the dense subspace $\Hil^\infty$, and is independent of the chosen cutoff function $\chi$ within the specified class. The thus (densely) defined operator $F_{\xi,\kappa}$ can be shown to extend to a bounded $\xi$-smooth operator on all of $\Hil$, which we denote by the same symbol \cite{BuchholzLechnerSummers:2010}. As can be seen from the above formula, setting $\kappa=0$ yields the undeformed operator $F_{\xi,0}=F$, for any $\xi\in\Xi$.

The deformation $F\to F_{\xi,\kappa}$ is closely related to Rieffel's deformation of C$^*$-algebras \cite{Rieffel:1992}, where one introduces the deformed product
\begin{align}\label{rieffel-product}
 F\times_{\xi,\kappa} G :=
\frac{1}{4\pi^2}
\lim_{\eps\ra0}
\int ds\,ds'\,e^{-iss'}\chi(\eps s,\eps s')\,\tau_{\xi,\kappa\theta s}(F)\tau_{\xi,s'}(G)
\,.
\end{align}
This limit exists in the norm topology of $\fF$ for any $\xi$-smooth $F,G\in\fF$, and $F\times_{\xi,\kappa} G$ is $\xi$-smooth as well.

As is well known, this procedure applies in particular to the deformation of classical theories in terms of star products. As field algebra, one would then take a suitable commutative $*$-algebra of functions on $M$, endowed with the usual pointwise operations. The isometry group acts on this algebra automorphically by pullback, and in particular, the flow $\flow_{\xi}$ of any Killing pair $\xi\in\Xi$ induces automorphisms. The Rieffel product therefore defines a star product on the subalgebra of smooth elements $f,g$ for this action,
\begin{align}
(f\star_{\xi,\kappa} g)(p)
=
\frac{1}{4\pi^2}
\lim_{\eps\ra0}
\int d^2s\,d^2s' e^{-iss'}\,\chi(\eps s,\eps s')\, f(\flow_{\xi,\kappa\theta  s}(p))g(\flow_{\xi,s'}(p))
\,.
\end{align}
The function algebra endowed with this star product can be interpreted as a noncommutative version of the manifold $M$, similar to the flat case \cite{GayralGraciaBondiaIochumSchuckerVarilly:2004}. Note that since we are using a two-dimensional spacelike flow on a four-dimensional spacetime, the deformation corresponds to a noncommutative Minkowski space with ``commuting time'' in the flat case.
\\\\
The properties of the deformation map $F\to F_{\xi,\kappa}$ which will be relevant here are the following.
%=======================================================
\begin{lemma}{\bf \cite{BuchholzLechnerSummers:2010}:}\label{thm:deformationproperties}
\\
Let $\xi\in\Xi$, $\kappa\in\Rl$, and consider $\xi$-smooth operators
$F,G\in\fF$. Then
 \begin{enumerate}
  \item ${F_{\xi,\kappa}}^*={F^*}_{\xi,\kappa}$.
  \item $F_{\xi,\kappa} G_{\xi,\kappa} = (F\times_{\xi,\kappa} G)_{\xi,\kappa}$.
  \item If\,\footnote{In \cite{BuchholzSummers:2008,BuchholzLechnerSummers:2010}, this statement is shown to hold under the weaker assumption that the commutator $[\tau_{\xi,s}(F),\,G]$ vanishes only for all $s\in S+S$, where $S$ is the joint spectrum of the generators of the $\Rl^2$-representation $U_\xi$ implementing $\tau_\xi$. But since usually $S=\Rl^2$ in the present setting, we refer here only to the weaker statement, where $S+S\subset\Rl^2$ has been replaced by $\Rl^2$.} $[\tau_{\xi,s}(F),\,G]=0$ for all $s\in\Rl^2$, then $[F_{\xi,\kappa}, G_{\xi,-\kappa}]=0$.
  \item If a unitary $Y\in\B(\Hil)$ commutes with $U_\xi(s)$, $s\in\Rl^2$, then $YF_{\xi,\kappa}Y^{-1} = (YFY^{-1})_{\xi,\kappa}$, and $YF_{\xi,\kappa}Y^{-1}$
 is $\xi$-smooth.
 \end{enumerate}
\end{lemma}
%=======================================================

Since we are dealing here with a field algebra obeying twisted locality, we also point out that statement c) of the above lemma carries over to the twisted local case.

%===========================================
\begin{lemma}\label{lemma:twist}
Let $\xi\in\Xi$ and $F,G\in\fF$ be $\xi$-smooth such that $[Z \tau_{\xi,s}(F)Z^*,\,G]=0$. Then
\begin{align}
 [ZF_{\xi,\kappa}Z^*, G_{\xi,-\kappa}]=0\,.
\end{align}
\end{lemma}
%===========================================
\begin{proof}
The Bose/Fermi operator $V$ commutes with the representation of the isometry group, and thus the same holds true for the twist $Z$ \eqref{def:Z}. So
in view of Lemma \ref{thm:deformationproperties} d), the assumption implies that $ZFZ^*$ is $\xi$-smooth, and $[\tau_{\xi,s}(ZFZ^*),\,G]=0$ for all $s\in\Rl^2$. In view of Lemma \ref{thm:deformationproperties} c), we thus have $[(ZFZ^*)_{\xi,\kappa}, G_{\xi,-\kappa}]=0$.  But as $Z$ and $U_\xi(s)$ commute, $(ZFZ^*)_{\xi,\kappa}=ZF_{\xi,\kappa}Z^*$, and the claim follows.
\end{proof}
%===========================================

The results summarized in Lemma \ref{thm:deformationproperties} and Lemma \ref{lemma:twist} will be essential for establishing the isotony and twisted locality properties of the deformed quantum field theory. To also control the covariance properties relating different Killing pairs, we need an additional lemma, closely related to \cite[Prop.2.9]{BuchholzLechnerSummers:2010}.

%===========================================
\begin{lemma}\label{lemma:deformed-operators}
 Let $\xi\in\Xi$, $\kappa\in\Rl$, and $F\in\fF$ be $\xi$-smooth.
\begin{enumerate}
 \item Let $h\in\iso$. Then $\alpha_h(F)$ is $h_*\xi$-smooth, and
\begin{align}\label{axi-covariance}
 \alpha_h(F_{\xi,\kappa})
=
\alpha_h(F)_{h_*\xi,\kappa}
\,.
\end{align}
\item For $N\in\GL(2,\Rl)$, it holds
\begin{align}
 F_{N\xi,\kappa} &=F_{\xi,\det N\cdot \kappa}\,. \label{ANxi}
\end{align}
In particular,
\begin{align}\label{Axi-prime}
F_{\xi',\kappa} &= F_{\xi,-\kappa}\,.
\end{align}
\end{enumerate}
\end{lemma}
%===========================================
\begin{proof}
a) The flow of $\xi$ transforms under $h$ according to $ h\flow_{\xi,s} = \flow_{h_*\xi,s}h$, so that $\alpha_h(\tau_{\xi,s}(F))=\tau_{h_*\xi,s}(\alpha_h(F))$. Since $F$ is $\xi$-smooth, and $\alpha_h$ is isometric, the smoothness of $s\mapsto\tau_{h_*\xi,s}(\alpha_h(F))$ follows. Using the strong convergence of the oscillatory integrals \eqref{def:Ala}, we compute on a smooth vector $\Psi\in\Hil^\infty$
\begin{align*}
 \alpha_h(F_{\xi,\kappa})\Psi
&=
\frac{1}{4\pi^2}
\lim_{\eps\ra0}
 \int ds\,ds'\, e^{-iss'}\, \chi(\eps s,\eps s')\,U(h \flow_{\xi,\kappa\theta s}h^{-1})\alpha_h(F)U(h \flow_{\xi,s'-\kappa\theta s}h^{-1})\Psi
\\
&=
\frac{1}{4\pi^2}
\lim_{\eps\ra0}
 \int ds\,ds'\, e^{-iss'}\, \chi(\eps s,\eps s')\,U(\flow_{h_*\xi,\kappa\theta s})\alpha_h(F)U(\flow_{h_*\xi,s'-\kappa\theta s})\Psi
\\
&=
\alpha_h(F)_{h_*\xi,\kappa}\Psi
\,,
\end{align*}
which entails \eqref{axi-covariance} since $\Hil^\infty\subset\Hil$ is dense.

b) In view of the transformation law $\varphi_{N\xi,s}=\varphi_{\xi,N^Ts}$ \eqref{N-flow}, we get, $\Psi\in\Hil^\infty$,
\begin{align*}
&F_{N\xi,\kappa}\Psi
=
\frac{1}{4\pi^2}
\lim_{\eps\ra0}
 \int ds\,ds'\, e^{-iss'}\, \chi(\eps s,\eps s')\,U(\flow_{N\xi,\kappa\theta s}) F U(\flow_{N\xi,s'-\kappa\theta s})\Psi
\\
&=
\frac{1}{4\pi^2|\det N|}
\lim_{\eps\ra0}
 \int ds\,ds'\, e^{-i(N^{-1}s,s')}\, \chi(\eps s,\eps (N^T)^{-1}s')\,
U_\xi(\kappa N^T\theta s) F U_\xi(s'-\kappa N^T\theta s)\Psi
\\
&=
\frac{1}{4\pi^2}
\lim_{\eps\ra0}
 \int ds\,ds'\, e^{-iss'}\, \chi(\eps N s,\eps (N^T)^{-1}s')\,
U_\xi(\kappa N^T\theta Ns) F U_\xi(s'-\kappa N^T\theta N s)\Psi
\\
&=
F_{\xi,\det N\cdot \kappa}\Psi
\,.
\end{align*}
In the last line, we used the fact that the value of the oscillatory integral does not depend on the choice of cutoff function $\chi$ or $\chi_N(s,s'):=\chi(Ns,(N^T)^{-1}s')$, and the equation $N^T\theta N=\det N\cdot \theta$, which holds for any $(2\times2)$-matrix $N$.

This proves \eqref{ANxi}, and since $\xi'=\Pi\xi$, with the flip matrix $\Pi=\left(0\;\;1\atop1\;\;0\right)$ which has $\det\Pi=-1$, also \eqref{Axi-prime} follows.
\end{proof}
%===========================================

Having established these properties of individual deformed operators, we now set out to deform the net $W\mapsto\fF(W)$ of wedge algebras. In contrast to the Minkowski space setting \cite{BuchholzLechnerSummers:2010}, we are here in a situation where the set $\Xi$ of all Killing pairs is not a single orbit of one reference pair under the isometry group. Whereas the deformation of a net of wedge algebras on Minkowski space amounts to deforming a single algebra associated with a fixed reference wedge (causal Borchers triple), we have to specify here more data, related to the coherent subfamilies $\W_{[\xi]}$ in the decomposition $\W=\bigsqcup_{[\xi]}\W_{[\xi]}$ of $\W$ \eqref{W-decomposition}. For each equivalence class $[\xi]$, we choose a representative $\xi$. In case there exists only a single equivalence class, this simply amounts to fixing a reference wedge together with a length scale for the Killing flow. With this choice  of representatives $\xi\in[\xi]$ made, we introduce the sets, $p\in M$,
\begin{align}
\fF(W_{\xi,p})_\kappa
&:=
\overline{\{F_{\xi,\kappa}\,:\, F\in\fF(W_{\xi,p}) \;\,\xi\text{-smooth } \}}^{\|\cdot\|}
\,,
\label{def:awla}
\\
\fF(W_{\xi',p})_\kappa
&:=
\overline{\{F_{\xi',\kappa}\,:\, F\in\fF(W_{\xi,p}') \;\,\xi'\text{-smooth } \}}^{\|\cdot\|}
\,.
\label{def:awla-2}
\end{align}
Note that the deformed operators in $\fF(W_{\xi',p})_\kappa$ have the form $F_{\xi',\kappa}=F_{\xi,-\kappa}$ \eqref{Axi-prime}, {\em i.e.} the sign of the deformation parameter depends on the choice of reference Killing pair.

The definitions (\ref{def:awla}) and (\ref{def:awla-2}) are extended to arbitrary wedges by setting
\begin{align}\label{def:AhWla}
 \fF(hW_{\xi,p})_\kappa := \alpha_h(\fF(W_{\xi,p})_\kappa)
\,,\qquad
 \fF(hW_{\xi,p}')_\kappa := \alpha_h(\fF(W_{\xi,p}')_\kappa)\,.
\end{align}
Recall that as $h$, $p$ and $[\xi]$ vary over $\iso$, $M$ and $\Xi\slash\!\!\sim$, respectively, this defines $\fF(W)_\kappa$ for all $W\in\W$ ({\it cf.} \eqref{def:Wxi}). It has to be proven that this assignment is well-defined, {\em e.g.} that \eqref{def:AhWla} is independent of the way the wedge $hW_{\xi,p}=\hti W_{\xi,\pti}$ is represented. This will be done below. However, note that the definition of $\fF(W)_\kappa$ {\em does} depend on our choice of representatives $\xi\in[\xi]$, since rescaling $\xi$ amounts to rescaling the deformation parameter (Lemma \ref{lemma:deformed-operators} b)).
\\
\\
Before establishing the main properties of the assignment $W\to\fF(W)_\kappa$, we check that the sets (\ref{def:awla}) and (\ref{def:awla-2}) are C$^*$-algebras. As the C$^*$-algebra $\fF(W_{\xi,p})$ is $\tau_{\xi}$-invariant and $\tau_{\xi}$ acts strongly continuously, the $\xi$-smooth operators in $\fF(W_{\xi,p})$ which appear in the definition \eqref{def:awla} form a norm-dense $*$-subalgebra. Now the deformation $F \mapsto F_{\xi,\kappa}$ is evidently linear and commutes with taking adjoints (Lemma \ref{thm:deformationproperties} a)); so the sets $\fF(W_{\xi,p})_\kappa$ are $*$-invariant norm-closed subspaces of $\B(\Hil)$. To check that these spaces are also closed under taking products, we again use the invariance of $\fF(W_{\xi,p})$ under $\tau_\xi$: By inspection of the Rieffel product \eqref{rieffel-product}, it follows that for any two $\xi$-smooth $F,G\in\fF(W_{\xi,p})$,  also the product $F\times_{\xi,\kappa} G$ lies in this algebra (and is $\xi$-smooth, see \cite{Rieffel:1992}). Hence the multiplication formula from Lemma \ref{thm:deformationproperties} b) entails that the above defined $\fF(W)_\kappa$ are actually C$^*$-algebras in $\B(\Hil)$.
\\
\\
The map $W\mapsto\fF(W)_\kappa$ defines the wedge-local field algebras of the deformed quantum field theory. Their basic properties are collected in the following theorem.

%==============================================
\begin{theorem}\label{thm:deformed-net}
The above constructed map $W\longmapsto\fF(W)_\kappa$, $W\in\W$, is a well-defined, isotonous, twisted wedge-local, $\iso$-covariant net of $C^*$-algebras on $\Hil$, {\it i.e.} $W,\Wti\in\W$,
\begin{align}
& \fF(W)_\kappa \subset \fF(\Wti)_\kappa \hspace*{19mm}\text{for } W\subset\Wti\,,\label{isotony2}
\\
& [Z F_\kappa Z^*,\, G_\kappa] = 0\hspace*{19mm}\text{for } F_\kappa\in\fF(W)_\kappa, G_\kappa\in\fF(W')_\kappa\,,\label{locality2}
\\
& \alpha_h(\fF(W)_\kappa) = \fF(hW)_\kappa\,,\hspace*{8mm} h\in \iso\,.\label{covariance2}
\end{align}
For $\kappa=0$, this net coincides with the original net, $\fF(W)_0=\fF(W)$, $W\in\W$.
\end{theorem}
%==============================================
\begin{proof}
It is important to note from the beginning that all claimed properties relate only wedges in the same coherent subfamily $\W_{[\xi]}$. This can be seen from the form \eqref{def:Wxi} of $\W_{[\xi]}$, which is manifestly invariant under isometries and causal complementation, and the structure of the inclusions (Proposition \ref{prop:inclusions}). So in the following proof, it is sufficient to consider a fixed but arbitrary equivalence class $[\xi]$, with selected representative $\xi$.

We begin with establishing the isotony of the deformed net, and therefore consider inclusions of wedges of the form $hW_{\xi,p}, hW_{\xi,p}'$, with $h\in\iso$, $p\in M$ arbitrary, and $\xi\in\Xi$ fixed. Starting with the inclusions  $hW_{\xi,p}\subseteq\hti W_{\xi,\pti}$, we note that according to \eqref{h-W} and Prop. \ref{prop:inclusions}, there exists $N\in\GL(2,\Rl)$ with positive determinant such that $h_*\xi=N\hti_*\xi$. Equivalently, $(\hti^{-1}h)_*\xi=N\xi$, which by Lemma \ref{lemma:group-xi} implies $\det N=1$. By definition, a generic $\xi$-smooth element of $\fF(hW_{\xi,p})_\kappa$ is of the form $\alpha_h(F_{\xi,\kappa})=\alpha_h(F)_{h_*\xi,\kappa}$ with some $\xi$-smooth $F\in\fF(W_{\xi,p})$. But according to the above observation, this can be rewritten as
\begin{align}\label{calc1}
\alpha_h(F_{\xi,\kappa})
=
\alpha_h(F)_{h_*\xi,\kappa}
=
\alpha_h(F)_{N\hti_*\xi,\kappa}
=
\alpha_h(F)_{\hti_*\xi,\kappa}
\,,
\end{align}
where in the last equation we used $\det N=1$ and Lemma \ref{lemma:deformed-operators} b). Taking into account that $hW_{\xi,p}\subseteq\hti W_{\xi,\pti}$, and that the undeformed net is covariant and isotonous, we have $\alpha_h(F)\in\fF(hW_{\xi,p})\subset\fF(\hti W_{\xi,p})$, and so the very right hand side of \eqref{calc1} is an element of $\fF(\hti W_{\xi,p})_\kappa$. Going to the norm closures, the inclusion $\fF(h W_{\xi,p})_\kappa \subset \fF(\hti W_{\xi,p})_\kappa$ of C$^*$-algebras follows.

Analogously, an inclusion of causal complements, $hW_{\xi,p}'\subseteq\hti W_{\xi,\pti}'$, leads to the inclusion of C$^*$-algebras $\fF(h W_{\xi,p}')_\kappa \subset \fF(\hti W_{\xi,p}')_\kappa$, the only difference to the previous argument consisting in an exchange of $h$,$\hti$ and $p,\pti$.

To complete the investigation of inclusions of wedges in $\W_{[\xi]}$, we must also consider the case $hW_{\xi,p}\subseteq \hti W_{\xi,\pti}'=W_{\hti_*\xi',\pti}$. By the same reasoning as before, there exists a matrix $N$ with $\det N=1$ such that $(\hti^{-1}h)_*\xi=N\xi'=N\Pi\xi$ with the flip matrix $\Pi$. So  $N':=N\Pi$ has determinant $\det N'=-1$, and $h_*\xi = N'\hti_*\xi'$. Using \eqref{Axi-prime}, we find for $\xi$-smooth $F\in\fF(hW_{\xi,p})$,
\begin{align}
\alpha_h(F_{\xi,\kappa})
=
\alpha_h(F)_{h_*\xi,\kappa}
=
\alpha_h(F)_{N'\hti_*\xi',\kappa}
=
\alpha_h(F)_{\hti_*\xi',-\kappa}
\,.
\end{align}
By isotony and covariance of the undeformed net, this deformed operator is an element of $\fF(\hti W_{\xi,\pti}')_\kappa$ \eqref{def:awla-2}, and taking the norm closure in \eqref{def:awla-2} yields $\fF(hW_{\xi,p})_\kappa \subset\fF(\hti W_{\xi,\pti}')_\kappa$. So the isotony \eqref{isotony2} of the net is established. This implies in particular that the net $\fF_\kappa$ is well-defined, since in case $hW_{\xi,p}$ equals $\hti W_{\xi,\pti}$ or its causal complement, the same arguments yield the equality of $\fF(hW_{\xi,p})_\kappa$ and $\fF(\hti W_{\xi,\pti})_\kappa$ respectively $\fF(\hti W_{\xi,\pti}')_\kappa$.

The covariance of $W\mapsto\fF(W)_\kappa$ is evident from the definition. To check twisted locality, it is thus sufficient to consider the pair of wedges $W_{\xi,p}$, $W_{\xi,p}'$. In view of the definition of the C$^*$-algebras $\fF(W)_\kappa$ \eqref{def:awla} as norm closures of algebras of deformed smooth operators, it suffices to show that any $\xi$-smooth $F\in\fF(W_{\xi,p})$, $G\in\fF(W_{\xi',\pti})$ fulfill the commutation relation
\begin{align}\label{comm-thm}
 [Z F_{\xi,\kappa}Z^*,\,G_{\xi',\kappa}]=0\,.
\end{align}
But $G_{\xi',\kappa}=G_{\xi,-\kappa}$ \eqref{Axi-prime}, and since the undeformed net is twisted local and covariant, we have $[\tau_{\xi,s}(F),\,G]=0$, for all $s\in\Rl^2$, which implies $[Z F_{\xi,\kappa}Z^*,\,G_{\xi,-\kappa}]=0$ by Lemma \ref{lemma:twist}.

The fact that setting $\kappa=0$ reproduces the undeformed net is a straightforward consequence of $F_{\xi,0}=F$ for any $\xi$-smooth operator, $\xi\in\Xi$.
\end{proof}
%===============================================

Theorem \ref{thm:deformed-net} is our main result concerning the structure of deformed quantum field theories on admissible spacetimes: It states that the same covariance and localization properties as on flat spacetime can be maintained in the curved setting. Whereas the action of the isometry group and the chosen representation space of $\fF$ are the same for all values of the deformation parameter $\kappa$, the concrete C$^*$-algebras $\fF(W)_\kappa$ depend in a non-trivial and continuous way on $\kappa$: For a fixed wedge $W$, the collection $\{\fF(W)_\kappa\,:\,\kappa\in\Rl\}$ forms a continuous field of C$^*$-algebras \cite{Dixmier:1977}; this follows from Rieffel's results \cite{Rieffel:1992} and the fact that $\fF(W)_\kappa$ forms a faithful representation of Rieffel's deformed C$^*$-algebra $(\fF(W),\times_\kappa)$ \cite{BuchholzLechnerSummers:2010}.

For deformed nets on Minkowski space, there also exist proofs showing that the {\em net} $W\mapsto\fF(W)_\kappa$ depends on $\kappa$, for example by working in a vacuum representation and calculating the corresponding collision operators \cite{BuchholzSummers:2008}. There one finds as a striking effect of the deformation that the interaction depends on $\kappa$, {\em i.e.} that deformations of interaction-free models have non-trivial S-matrices. However, on generic curved spacetimes, a distinguished invariant state like the vacuum state with its positive energy representation of the translations does not exist. Consequently, the result concerning scattering theory cannot be reproduced here. Instead we will establish the non-equivalence of the undeformed net $W\mapsto \fF(W)$ and the deformed net $W\mapsto \fF(W)_\kappa$, $\kappa\neq0$, in a concrete example model in Section \ref{sec:dirac}.

As mentioned earlier, the family of wedge regions $\W(M,\gST)$ is causally separating in a subclass of admissible spacetimes, including the Friedmann-Robertson-Walker universes. In this  case, the extension of the net $\fF_\kappa$ to double cones or similar regions $\OO\subset M$ via
\begin{align}
 \fF(\OO)_\kappa := \bigcap_{W\supset\OO}\fF(W)_\kappa
\end{align}
is still twisted local. These algebras contain all operators localized in the region $\OO$ in the deformed theory. On other spacetimes $(M,\gST)$, such an extension is possible only for special regions, intersections of wedges, whose shape and size depend on the structure of $\W(M,\gST)$.

 Because of the relation of warped convolution to noncommutative spaces, where sharp localization is impossible, it is expected that $\fF(\OO)$ contains only multiples of the identity if $\OO$ has compact closure. We will study this question in the context of the deformed Dirac field in Section \ref{sec:dirac}.
\\\\
We conclude this section with a remark concerning the relation between the field and observable net structure of deformed quantum field theories. The field net $\fF$ is composed of Bose and Fermi fields, and therefore contains observable as well as unobservable quantities. The former give rise to the {\em observable net} $\frA$ which consists of the subalgebras invariant under the grading automorphism $\gamma$. In terms of the projection $v(F):=\frac{1}{2}(F+\gamma(F))$, the observable wedge algebras are
\begin{align}\label{def:AW}
 \frA(W) := \{F\in\fF(W)\,:\, F=\gamma(F)\} = v(\fF(W))
\,,\qquad
W\in\W\,,
\end{align}
so that $\frA(W)$, $\frA(\Wti)$ commute (without twist) if $W$ and $\Wti$ are spacelike separated.

Since the observables are the physically relevant objects, we could have considered a deformation $\frA(W)\to\frA(W)_\kappa$ of the observable wedge algebras along the same lines as we did for the field algebras. This approach would have resulted precisely in the $\gamma$-invariant subalgebras of the deformed field algebras $\fF(W)_\kappa$, {\em i.e.} the diagram
$$
\begin{CD}
\fF(W)   @>\text{deformation\;}>>  \fF(W)_\kappa\\
@VvVV  @VVvV\\
\frA(W)  @>\text{deformation\;}>>  \frA(W)_\kappa
\end{CD}
$$
commutes. This claim can quickly be verified by noting that the projection $v$ commutes with the deformation map $F\mapsto F_{\xi,\kappa}$.

%===============================================
\subsection{The Dirac field and its deformation}\label{sec:dirac}
%===============================================

After the model-independent description of deformed quantum field theories carried out in the previous section, we now consider the theory of a free Dirac quantum field as a concrete example model. We first briefly recall the notion of Dirac (co)spinors and the classical Dirac equation, following largely \cite{DappiaggiHackPinamonti:2009, Sanders:2008} and partly \cite{Dimock:1982,FewsterVerch:2002}, where the proofs of all the statements below are presented and an extensive description of the relevant concepts is available. Afterwards, we consider the quantum Dirac field using Araki's self-dual CAR-algebra formulation \cite{Araki:1971}.
\\
\\
As before, we work on a fixed but arbitrary admissible spacetime $(M,\gST)$ in the sense of Definition \ref{admissible} and we fix its orientation. Therefore, as a four-dimensional, time oriented and oriented, globally hyperbolic spacetime, $M$ admits a {\em spin structure} $(SM,\rho)$, consisting of a principle bundle $SM$ over $M$ with structure group $\SL(2,\Cl)$, and a smooth bundle homomorphism $\rho$ projecting $SM$ onto the frame bundle $FM$, which is a principal bundle over $M$ with $\SO(1,3)_0$ as structure group. The map $\rho$ preserves base points and is equivariant in the sense that it intertwines the natural right actions $R$ of $\SL(2,\Cl)$ on $SM$ and of $\SO(1,3)_0$ on $FM$, respectively,
\begin{align}
\rho\circ R_{\widetilde\Lambda} = R_\Lambda\circ\rho,\qquad\Lambda\in \SO(1,3)_0\,,
\end{align}
with the covering homomorphism $\Lambda\mapsto\widetilde\Lambda$ from $\SL(2,\Cl)$ to $\SO(1,3)_0$.

Although each spacetime of the type considered here has a spin structure \cite[Thm.2.1, Lem.2.1]{DappiaggiHackPinamonti:2009}, this is only unique if the underlying manifold is simply connected \cite{Geroch:1968,Geroch:1970}, {\em i.e.} if all edges are simply connected in the case of an admissible spacetime. In the following, it is understood that a fixed choice of spin structure has been made.

The main object we shall be interested in is the {\em Dirac bundle}, that is the associated vector bundle
\begin{align}
 DM:= SM\times_T\Cl^4
\end{align}
with the representation $T:=D^{(\frac{1}{2},0)}\oplus D^{(0,\frac{1}{2})}$ of $\SL(2,\Cl)$ on $\Cl^4$. {\em Dirac spinors} $\psi$ are smooth global sections of $DM$, and the space they span will be denoted $\mathcal{E}(DM)$. The dual bundle $D^*M$ is called the {\em dual Dirac bundle}, and its smooth global sections $\psi'\in\mathcal{E}(D^*M)$ are referred to as {\em Dirac cospinors}.

For the formulation of the Dirac equation, we need two more ingredients. The first are the so-called {\em gamma-matrices}, which are the coefficients of a global tensor $\gamma\in\mathcal{E}(T^*M\otimes DM\otimes D^*M)$ such that $\gamma=\gamma_{aB}^Ae^a\otimes E_A\otimes E^B$. Here $E_A$ and $E^B$ with $A,B=1,...,4$ are four global sections of $DM$ and $D^*M$ respectively, such that $(E_A,E^B)=\delta_A^B$, with $(.\,,.)$ the natural pairing between dual elements. Notice that $E_A$ descends also from a global section $E$ of $SM$ since we can define $E_A(x):=[(E(x),z_A)]$ where $z_A$ is the standard basis of $\Cl^4$. At the same time, out of $E$, we can construct $e_a$, with $a=0,...,3$, as a set of four global sections of $TM$ once we define $e:=\rho\circ E$ as a global section of the frame bundle, which, in turn, can be read as a vector bundle over $TM$ with $\Rl^4$ as typical fibre. The set of all $e_a$ is often referred to as the non-holonomic basis of the base manifold. In this case upper indices are defined via the natural pairing over $\Rl^4$, that is $(e^b,e_a)=\delta_a^b$. Furthermore we choose the gamma-matrices to be of the following form:
\begin{align}
\gamma_0=\left(\begin{array}{cc}
1_2 & 0 \\
0 & -1_2
\end{array}\right),\qquad \gamma_k=\left(\begin{array}{cc}
0 & \sigma_k \\
-\sigma_k & 0
\end{array}\right),
\qquad k=1,2,3,
\end{align}
where the $\sigma_k$ are the Pauli matrices, and $1_n$ denotes the $n\times n$ identity matrix. These matrices fulfill the anticommutation relation $\left\{\gamma_a,\gamma_b\right\}=2\eta_{ab}1_4,$ with the flat Minkowski metric $\eta$. They therefore depend on the sign convention in the metric signature, and differ from those introduced in \cite{DappiaggiHackPinamonti:2009}, where a different convention was used.

The last ingredient we need to specify is the {\em covariant derivative} (spin connection) on the space of smooth sections of the Dirac bundle, that is
\begin{align}
\nabla:\mathcal{E}(DM)\ra\mathcal{E}(T^*M\otimes DM),
\end{align}
whose action on the sections $E_A$ is given as $\nabla E_A=\sigma_{aA}^B e^a E_B$. The connection coefficients can be expressed as $\sigma^B_{aA}=\frac{1}{4}\Gamma^b_{ad}\gamma_{bC}^B\gamma^{dC}_A$, where $\Gamma^b_{ad}$ are the coefficients arising from the Levi-Civita connection expressed in terms of non-holonomic basis \cite[Lem.2.2]{DappiaggiHackPinamonti:2009}.
\\\\
We can now introduce the Dirac equation for spinors $\psi\in\mathcal{E}(DM)$ and cospinors  $\psi'\in\mathcal{E}(D^*M)$ as
\begin{align}\label{Diraceq}
D\psi  &:=(-i\gamma^\mu\nabla_\mu+m1)\psi=0\\
D'\psi'&:=(+i\gamma^\mu\nabla_\mu+m1)\psi'=0\,,
\end{align}
where $m\geq0$ is a constant while $1$ is the identity on the respective space.

The Dirac equation has unique advanced and retarded fundamental solutions \cite{DappiaggiHackPinamonti:2009}: Denoting the smooth and compactly supported sections of the Dirac bundle by $\mathcal{D}(DM)$, there exist two continuous linear maps
$$S^\pm:\mathcal{D}(DM)\ra\mathcal{E}(DM),$$
such that $S^\pm D=D S^\pm = 1$ and $\supp(S^\pm f)\subseteq J^\pm (\supp(f))$ for all $f\in\mathcal{D}(DM)$. In the case of cospinors, we shall instead talk about $S_*^\pm:\mathcal{D}(D^*M)\ra\mathcal{E}(D^*M)$ and they have the same properties of $S^\pm$, except that $S_*^\pm D'=D'S_*^\pm=1$. In analogy with the theory of real scalar fields, the {\em causal propagators} for Dirac spinors and cospinors are defined as $S := S^+-S^-$ and $S_*:=S^+_*-S^-_*$, respectively.

For the formulation of a quantized Dirac field, it is advantageous to collect spinors and cospinors in a single object. We therefore introduce the space
\begin{align}
 \DD := \DD(DM\oplus D^*M)\,,
\end{align}
on which we have the conjugation
\begin{align}
 \Gamma(f_1\oplus f_2) := f_1^*\beta \oplus \beta^{-1}f_2^*\,,
\end{align}
defined in terms of the adjoint $f\mapsto f^*$ on $\Cl^4$ and the Dirac conjugation matrix $\beta$. This matrix is the unique selfadjoint element of $\SL(4,\Cl)$ with the properties that $\gamma^*_a=-\beta\gamma_a\beta^{-1}$, $a=0,...,3$, and $i\beta n^a\gamma_a$ is a positive definite matrix, for any timelike future-directed vector $n$.

Due to these properties, the sesquilinear form on $\DD$ defined as
\begin{equation}\label{sesquilinear}
(f_1 \oplus f_2,\, g_1\oplus g_2)
:=
- i\langle f_1^*\beta, Sg_1\rangle
+ i\langle S_* g_2, \beta^{-1}f^*_2\rangle,
\end{equation}
where $\langle.\,,.\rangle$ is the global pairing between $\cal{E}(DM)$ and $\DD(D^*M)$ or between $\cal{E}(D^*M)$ and $\DD(DM)$, is positive semi-definite. Therefore the quotient
\begin{align}
 \K := \DD/(\ker S\oplus \ker S_*)\,.
\end{align}
has the structure of a pre-Hilbert space, and we denote the corresponding scalar product and norm by $\langle\,.\,,\,.\,\rangle_S$ and $\|f\|_S:=\langle f,f\rangle_S^{1/2}$. The conjugation $\Gamma$ descends to the quotient $\K$, and we denote its action on $\K$ by the same symbol. Moreover, $\Gamma$ is compatible with the sesquilinear form \eqref{sesquilinear} in such a way that it extends to an antiunitary involution $\Gamma=\Gamma^*=\Gamma^{-1}$ on the Hilbert space $\overline{\K}$ \cite[Lem.4.2.4]{Sanders:2008}.

Regarding covariance, the isometry group $\iso$ of $(M,\gST)$ naturally acts on the sections in $\DD$ by pullback. In view of the geometrical nature of the causal propagator, this action descends to the quotient $\K$ and extends to a unitary representation $u$ of $\iso$ on the Hilbert space $\overline{\K}$.

Given the pre-Hilbert space $\K$, the conjugation $\Gamma$, and the representation $u$ as above, the quantized Dirac field can be conveniently described as follows \cite{Araki:1971}.  We consider the C$^*$-algebra $\CAR(\K,\Gamma)$, that is, the unique unital C$^*$-algebra generated by the symbols $B(f)$, $f\in\K$, such that, for all $f,g\in\K$,
\begin{enumerate}
\item $f \longmapsto B(f)$ is complex linear,
\item $B(f)^* = B(\Gamma f)$,
\item $\{B(f),\,B(g)\} = \langle\Gamma f,g\rangle_S\cdot 1$.
\end{enumerate}
The field equation is implicit here since we took the quotient with respect to the kernels of $S,S_*$. The standard picture of spinors and cospinors can be recovered via the identifications $\psi(g)=B(0\oplus g)$ and $\psi^\dagger(f)=B(f\oplus 0)$.

As is well known, the Dirac field satisfies the standard assumptions of quantum field theory, and we briefly point out how this model fits into the general framework used in Section \ref{sec:generaldeformations}. The global field algebra $\fF:=\CAR(\K,\Gamma)$ carries a natural $\iso$-action $\alpha$ by Bogoliubov transformations: Since the unitaries $u(h)$, $h\in\iso$, commute with the conjugation $\Gamma$, the maps
\begin{align*}
 \alpha_h(B(f)) := B(u(h)f)\,,\qquad f\in\K,
\end{align*}
extend to automorphisms of $\fF$. Similarly, the grading automorphism $\gamma$ is fixed by
\begin{align*}
 \gamma(B(f)) := -B(f)\,,
\end{align*}
and clearly commutes with $\alpha_h$. The field algebra is faithfully represented on the Fermi Fock space $\Hil$ over $\overline{\K}$, where the field operators take the form
\begin{align}
 B(f) = \frac{1}{\sqrt{2}}(a^\dagger(f)+a(\Gamma f))\,,
\end{align}
with the usual Fermi Fock space creation/annihilation operators $a^{\#}(f)$, $f\in\K$. In this representation, the second quantization $U$ of $u$ implements the action $\alpha$. The Bose/Fermi-grading can be implemented by $V=(-1)^N$, where $N\in\B(\Hil)$ is the Fock space number operator \cite{Foit:1983}.

Regarding the regularity assumptions on the symmetries, recall that the anticommutation relations of the CAR algebra imply that $\|a(f)\|=\|f\|_S$, and thus $f\mapsto B(f)$ is a linear continuous map from $\K$ to $\fF$. As $s\mapsto u_\xi(s)f$ is smooth in the norm topology of $\K$ for any $\xi\in\Xi$, $f\in\K$, this implies that the field operators $B(f)$ transform smoothly under the action $\alpha$. Furthermore, the unitarity of $u$ yields strong continuity of $\alpha$ on all of $\fF$, as required in Section \ref{sec:generaldeformations}.

The field algebra $\fF(W)\subset\fF$ associated with a wedge $W\in\W$ is defined as the unital C$^*$-algebra generated by all $B(f)$, $f\in\K(W)$, where $\K(W)$ is the set of (equivalence classes of) smooth and compactly supported sections of $DM\oplus D^*M$ with support in $W$. Since $\langle \Gamma f,g\rangle_S=0$ for $f\in\K(W)$, $g\in\K(W')$, we have $\{B(f),B(g)\}=0$ for $f\in\K(W)$, $g\in\K(W')$, which implies the twisted locality condition \eqref{twisted-locality}. Isotony is clear from the definition and covariance under the isometry group follows from $u(h)\K(W)=\K(hW)$.
\\\\
The model of the Dirac field therefore fits into the framework of Section \ref{sec:generaldeformations}, and the warped convolution deformation defines a one-parameter family of deformed nets $\fF_\kappa$. Besides the properties which were established in the general setting in Section \ref{sec:generaldeformations}, we can here consider the explicit deformed field $B(f)_{\xi,\kappa}$. A nice characterization of these operators can be given in terms of their $n$-point functions associated with a quasifree state $\om$ on $\fF$.

Let $\om$ be an $\iso$-invariant quasifree state on $\fF$, and let ($\Hil^\om,\pi^\om,\Om^\om)$ denote the associated GNS triple. As a consequence of invariance of $\om$, the GNS space $\Hil^\om$ carries a unitary representation $U^\om$ of $\iso$ which leaves $\Om^\om$ invariant. In this situation, the warping map
\begin{align}\label{warp-om}
F_{\xi,\kappa} \mapsto F_{\xi,\kappa}^\om
:=
\frac{1}{4\pi^2}\lim_{\eps\ra0}
\int ds\,ds'\,
e^{-iss'}\,
\chi(\eps s,\eps s')\,
U^\om_\xi(\kappa\theta  s)\pi^\om(F)U_\xi^\om(s'-\kappa\theta s)
\,,
\end{align}
defined for $\xi$-smooth $F\in\fF$ as before, extends continuously to a representation of the Rieffel-deformed C$^*$-algebra $(\fF,\times_{\xi,\kappa})$ on $\Hil^\om$ \cite[Thm.2.8]{BuchholzLechnerSummers:2010}. Moreover, the $U^\om$-invariance of $\Om^\om$ implies
\begin{align}\label{FxiOm}
F^\om_{\xi,\kappa}\Om^\om
=
\pi^\om(F)\Om^\om
\,,\qquad \xi\in\Xi,\kappa\in\Rl,\,F\in\fF\;\;\xi\text{-smooth}.
\end{align}
Since the CAR-algebra is simple, all its representations are faithful \cite{BratteliRobinson:1997}. We will therefore identify $\fF$ with its representation $\pi^\om(\fF)$ in the following, and drop the $\om$-dependence of $\Hil^\om,\Om^\om,U^\om$ and the warped convolutions $F^\om_{\xi,\kappa}$ from our notation.

To characterize the deformed field operators $B(f)_{\xi,\kappa}$, we will consider the $n$-point functions
\begin{align*}
\om_n(f_1,...\,,f_n)
:=
\om(B(f_1)\cdots B(f_n))
=
\langle\Om,\,B(f_1)\cdots B(f_n)\Om\rangle
\,,\qquad
f_1,...\,,f_n\in\K\,,
\end{align*}
and the corresponding deformed expectation values of the deformed fields, called {\em deformed $n$-point functions},
\begin{align*}
\om^{\xi,\kappa}_n(f_1,...\,,f_n)
:=
\langle\Om,\,B(f_1)_{\xi,\kappa}\cdots B(f_n)_{\xi,\kappa}\Om\rangle
\,,\qquad
f_1,...\,,f_n\in\K\,.
\end{align*}
Of particular interest are the {\em quasifree} states, where $\om_n$ vanishes if $n$ is odd, and  $\om_n$ is a linear combination of products of two-point functions $\om_2$ if $n$ is even. In particular, the undeformed four-point function of a quasifree state reads
\begin{align}\label{4pt-null}
\!\!\!\!\!\!
\om_4(f_1,f_2,f_3,f_4)
=
\om_2(f_1,f_2)\om_2(f_3,f_4)
+\om_2(f_1,f_4)\om_2(f_2,f_3)
-\om_2(f_1,f_3)\om_2(f_2,f_4)
\,.
\!\!\!\!
\end{align}

%==============================================
\begin{proposition}\label{prop:n-pt}
 The deformed $n$-point functions of a quasifree and $\iso$-invariant state vanish for odd $n$. The lowest deformed even $n$-point functions are: $f_1,...,f_4\in\K$,
\begin{align}
 \om_2^{\xi,\kappa}(f_1,f_2) &= \om_2(f_1,f_2)
\,,\label{2pt}
\\
\omega_4^{\xi,\kappa}(f_1,f_2,f_3,f_4)
&=
\om_2(f_1,f_2)\om_2(f_3,f_4)
+\om_2(f_1,f_4)\om_2(f_2,f_3)
\label{4pt}
\\
&
- \frac{1}{4\pi^2}
\lim_{\varepsilon\ra 0}
\int ds\,ds'\,
e^{-iss'}
\chi(\eps s,\eps s')
\,
\om_2(f_1,u_\xi(s)f_3)\cdot \om_2(f_2,u_\xi(2\kappa\theta  s')f_4)
\,.
\nonumber
\end{align}
\end{proposition}
%==============================================
\begin{proof}
The covariant transformation behaviour of the Dirac field, $U_\xi(s)B(f)U_\xi(s)^{-1}=B(u_\xi(s)f)$, the invariance of $\Om$ and the form \eqref{warp-om} of the warped convolution imply that any deformed $n$-point function can be written as an integral over undeformed $n$-point functions with transformed arguments. As the latter functions vanish for odd $n$, we also have $\om^{\xi,\kappa}_n=0$ for odd $n$.

Taking into account \eqref{FxiOm}, we obtain for the deformed two-point function
\begin{align*}
\om^{\xi,\kappa}_2(f_1,f_2)
&=
\langle\Om,\,B(f_1)_{\xi,\kappa}B(f_2)_{\xi,\kappa}\Om\rangle
\\
&=
\langle (B(f_1)^*)_{\xi,\kappa}\Om,\,B(f_2)_{\xi,\kappa}\Om\rangle
=
\langle B(f_1)^*\Om,\,B(f_2)\Om\rangle
=
\om_2(f_1,f_2)\,,
\end{align*}
proving \eqref{2pt}.

To compute the four-point function, we use again $B(f)_{\xi,\kappa}\Om=B(f)\Om$ and $U_\xi(s)\Om=\Om$. Inserting the definition of the warped convolution \eqref{def:Ala}, and using the transformation law $U_\xi(s)B(f)U_\xi(s)^{-1}=B(u_\xi(s)f)$ and the shorthand $f(s):=u_\xi(s)f$, we obtain
\begin{align*}
&\om^{\xi,\kappa}_4(f_1,f_2,f_3,f_4)
=
\langle\Om,\,B(f_1)B(f_2)_{\xi,\kappa}B(f_3)_{\xi,\kappa}B(f_4)\Om\rangle
\\
&=
(2\pi)^{-4}
\lim_{\eps_1,\eps_2\ra0}
\int d\bs\,
e^{-i(s_1s_1'+s_2s_2')}
\chi_\eps(\bs)
\,
\om_4(f_1,f_2(\kappa\theta s_1),f_3(\kappa\theta  s_2+s_1'),f_4(s_1'+s_2'))\,,
\end{align*}
where $d\bs=ds_1\,ds_1'\,ds_2\,ds_2'$ and $\chi_\eps(\bs)=\chi(\eps_1 s_1,\eps_1 s_1')
\chi(\eps_2 s_2,\eps_2 s_2')$. After the substitutions $s_2'\to s_2'-s_1'$ and $s_1'\to s_1'-\kappa\theta  s_2$, the integrations over $s_2,s_2'$ and the limit $\eps_2\ra0$ can be carried out. The result is
\begin{align*}
(2\pi)^{-2}
\lim_{\eps_1\ra0}
\int ds_1\,ds'_1\,
e^{-is_1s_1'}
\hat{\chi}(\eps_1 s_1,\eps_1 s_1')
\,
\om_4(f_1,f_2(\kappa\theta s_1),f_3(s_1'),f_4(s_1'-\kappa\theta  s_1))\,,
\end{align*}
with a smooth, compactly supported cutoff function $\hat{\chi}$ with $\hat{\chi}(0,0)=1$.

We now use the fact that $\om$ is quasi-free and write $\om_4$ as a sum of products of two-point functions \eqref{4pt-null}. Considering the term where $f_1,f_2$  and $f_3,f_4$ are contracted, in the second factor $\om_2(f_3(s'_1),f_4(s_1'-\kappa\theta s_1))$ the $s_1'$-dependence drops out because $\om$ is invariant under isometries. So the integral over $s_1'$ can be performed, and yields a factor $\delta(s_1)$ in the limit $\eps_1\ra0$. Hence all $\kappa$-dependence drops out in this term, as claimed in \eqref{4pt}.

Similarly, in the term where $f_1,f_4$  and $f_2,f_3$ are contracted, all integrations disappear after using the invariance of $\om$ and making the substitution $s_1'\to s_1'+\kappa\theta s_1$. Also this term does not depend on $\kappa$.

Finally, we compute the term containing the contractions $f_1,f_3$ and $f_2,f_4$,
\begin{align*}
&(2\pi)^{-2}
\lim_{\eps_1\ra0}
\int ds_1\,ds'_1\,
e^{-is_1s_1'}
\hat{\chi}(\eps_1 s_1,\eps_1 s_1')
\,
\om_2(f_1,f_3(s_1'))\cdot \om_2(f_2(\kappa\theta s_1),f_4(s_1'-\kappa\theta  s_1))
\\
&=
(2\pi)^{-2}
\lim_{\eps_1\ra0}
\int ds_1\,ds'_1\,
e^{-is_1s_1'}
\hat{\chi}(\eps_1 s_1,\eps_1 s_1')
\,
\om_2(f_1,f_3(s_1'))\cdot \om_2(f_2,f_4(s_1'-2\kappa\theta  s_1))
\\
&=
(2\pi)^{-2}
\lim_{\eps_1\ra0}
\int ds_1\,ds'_1\,
e^{-is_1s_1'}
\tilde{\chi}(\eps_1 s_1,\eps_1 s_1')
\,
\om_2(f_1,f_3(s_1'))\cdot \om_2(f_2,f_4(2\kappa\theta  s_1))
\,.
\end{align*}
In the last step, we substituted $s_1\to s_1+\frac{1}{2\kappa}\theta^{-1}s_1'$, used the antisymmetry of $\theta$ and absorbed the new variables in $\hat{\chi}$ in a redefinition of this function. Since the oscillatory integrals are independent of the particular choice of cutoff function, comparison with \eqref{4pt} shows that the proof is finished.
\end{proof}
%==============================================

The structure of the deformed $n$-point functions build from a quasifree state $\om$ is quite different from the undeformed case. In particular, the two-point function is undeformed, but the four-point function depends on the deformation parameter. For even $n>4$, a structure similar to the $n$-point functions on noncommutative Minkowski space \cite{GrosseLechner:2008} is expected, which are all $\kappa$-dependent. These features clearly show that the deformed field $B(f)_{\xi,\kappa}$, $\kappa\neq0$, differs from the undeformed field. Considering the commutation relations of the deformed field operators, it is also straightforward to check that the deformed field is not unitarily equivalent to the undeformed one.

However, this structure does not yet imply that the deformed and undeformed Dirac quantum field theories are inequivalent. For there could exist a unitary $V$ on $\Hil$ satisfying $VU(h)V^*=U(h)$, $h\in\iso$, $V\Om=\Om$, which does not interpolate the deformed and undeformed fields, but the C$^*$-algebras according to $V\fF(W)_\kappa V^*=\fF(W)$, $W\in\W$. If such a unitary exists, the two theories would be physically indistinguishable.

On flat spacetime, an indirect way of ruling out the existence of such an intertwiner $V$, and thus establishing the non-equivalence of deformed and undeformed theories, is to compute their S-matrices and show that these depend on $\kappa$ in a non-trivial manner. However, on curved spacetimes, collision theory is not available and we will therefore follow the more direct non-equivalence proof of \cite[Lem.4.6]{BuchholzLechnerSummers:2010}, adapted to our setting. This proof aims at showing that the local observable content of warped theories is restricted in comparison to the undeformed setting, as one would expect because of the connection to noncommutative spacetime. However, the argument requires a certain amount of symmetry, and we therefore restrict here to the case of a Friedmann-Robertson-Walker spacetime $M$.

As discussed in Section \ref{sec:examples}, $M$ can then be viewed as $J\times\Rl^3\subset\Rl^4$ via a conformal embedding, where $J\subset\Rl$ is an interval. Recall that in this case, we have the Euclidean group E$(3)$ contained in $\iso(M,\gST)$, and can work in global coordinates $(\tau,x,y,z)$, where $\tau\in J$ and $x,y,z\in\Rl$ are the flow parameters of Killing fields. As reference Killing pair, we pick $\zeta:=(\partial_y,\partial_z)$, and as reference wedge, the ``right wedge'' $W^0:=W_{\zeta,0}=\{(\tau,x,y,z)\,:\,\tau\in J,\,x>|\tau|\}$.

In this geometric context, consider the rotation $r^\varphi$ about angle $\varphi$ in the $x$-$y$-plane, and the cone
\begin{align}\label{def:cone}
 \C := r^\varphi W^0 \cap r^{-\varphi} W^0\,,
\end{align}
with some fixed angle $|\varphi|<\frac{\pi}{2}$. Clearly $\C\subset W^0$, and the reflected cone $j_x\C$, where $j_x(t,x,y,z)=(t,-x,y,z)$, lies spacelike to $W^0$ and $r^\varphi W^0$.

Moreover, we will work in the GNS-representation of a particular state $\om$ on $\fF$ for the subsequent proposition, which besides the properties mentioned above also has the Reeh-Schlieder property. That is, the von Neumann algebra $\fF(\C)''\subset\B(\Hil^\om)$ has $\Om^\om$ as a cyclic vector.

Since the Dirac field theory is a locally covariant quantum field theory satisfying the time slice axiom \cite{Sanders:2010}, the existence of such states can be deduced by spacetime deformation arguments \cite[Thm.4.1]{Sanders:2009-1}. As $M$ and the Minkowski space have unique spin structures, and $M$ can be deformed to Minkowski spacetime in such a way that its E$(3)$ symmetry is preserved, the state obtained from deforming the Poincar\'e invariant vacuum state on $\Rl^4$ is still invariant under the action of the Euclidean group.

In the GNS representation of such a Reeh-Schlieder state on a Friedmann-Robertson-Walker spacetime, we find the following non-equivalence result.

%==============================================
\begin{proposition}\label{prop:inequivalence}
Consider the net $\fF_\kappa$ generated by the deformed Dirac field on a Fried\-mann-Robertson-Walker spacetime with flat spatial sections in the GNS representation of a quasifree invariant state with the Reeh-Schlieder property. Then the implementing vector $\Om$ is cyclic for the field algebra ${\fF(\C)_\kappa}''$ associated with the cone \eqref{def:cone} if and only if $\kappa=0$. In particular, the nets $\fF_0$ and $\fF_\kappa$ are inequivalent for $\kappa\neq0$.
\end{proposition}
%==============================================
\begin{proof}
Let $f\in\K(\C)$, so that $f, u(r^{-\varphi})f\in\K(W^0)$, and both field operators, $B(f)_{\zeta,\kappa}$ and $B(u(r^{-\varphi})f)_{\zeta,\kappa}$, are contained in $\fF(W^0)_\kappa$. Taking into account the covariance of the deformed net, it follows that $U(r^\varphi)B(u(r^{-\varphi})f)_{\zeta,\kappa} U(r^{-\varphi}) = B(f)_{r^\varphi_*\zeta,\kappa}$ lies in $\fF(r^\varphi W^0)_\kappa$.

Now the cone $\C$ is defined in such a way that the two wedges $W^0$ and $r^\varphi W^0$ lie spacelike to $j_x\C$. Let us assume that $\Om$ is cyclic for ${\fF(\C)_\kappa}''$, which by the unitarity of $U(j_x)$ is equivalent to $\Om$ being cyclic for ${\fF(j_x\C)_\kappa}''$. Hence $\Om$ is separating for the commutant ${\fF(j_x\C)_\kappa}'$, which by locality contains $\fF(W^0)_\kappa$ and $\fF(r^\varphi W^0)_\kappa$. But in view of \eqref{FxiOm}, $B(f)_{\zeta,\kappa}$ and $B(f)_{r^\varphi_*\zeta,\kappa}$ coincide on $\Om$,
\begin{align*}
B(f)_{r^\varphi_*\zeta,\kappa}\Om
=
B(f)\Om
=
B(f)_{\zeta,\kappa}\Om
\,,
\end{align*}
so that the separation property implies $B(f)_{\zeta,\kappa}=B(f)_{r^\varphi_*\zeta,\kappa}$.

To produce a contradiction, we now show that these two operators are actually not equal. To this end, we consider a difference of four-point functions \eqref{4pt}, with smooth vectors $f_1,f_2:=f,f_3:=f,f_4$, and Killing pairs $\zeta$ respectively $r^\varphi_*\zeta$. With the abbreviations $w^\varphi_{ij}(s):=\om_2(f_i,u_{r^\varphi_*\zeta}(s)f_j)$, we obtain
\begin{align*}
&\hspace*{-4cm}
\langle\Om,B(f_1)\left(B(f)_{\zeta,\kappa}B(f)_{\zeta,\kappa}
-
B(f)_{r^\varphi_*\zeta,\kappa}B(f)_{r^\varphi_*\zeta,\kappa}\right)B(f_4)\Om\rangle
\\
&=
\om^{\zeta,\kappa}_4(f_1,f,f,f_4)
-
\om^{r^\varphi_*\zeta,\kappa}_4(f_1,f,f,f_4)
\\
&=
(w^0_{13}\star_\kappa w^0_{24})(0)
-
(w^\varphi_{13}\star_\kappa w^\varphi_{24})(0)\,,
\end{align*}
where $\star_\kappa$ denotes the Weyl-Moyal star product on smooth bounded functions on $\Rl^2$, with the standard Poisson bracket given by the matrix \eqref{def:Q} in the basis $\{\zeta_1,\zeta_2\}$. Now the asymptotic expansion of this expression for $\kappa\ra0$ gives in first order the difference of Poisson brackets \cite{EstradaGracia-BondiaVarilly:1989}
\begin{align*}
\{w^0_{13},\,w^0_{24}\}(0)-\{w^\varphi_{13},\,w^\varphi_{24}\}(0)
&=
\langle f_1, P^\zeta_1 f\rangle\langle f, P^\zeta_2 f_4\rangle
-
\langle f_1, P^\zeta_2 f\rangle\langle f, P^\zeta_1 f_4\rangle
\\
&\qquad -
\langle f_1, P^{r^\varphi_*\zeta}_1 f\rangle\langle f, P^{r^\varphi_*\zeta}_2 f_4\rangle
+
\langle f_1, P^{r^\varphi_*\zeta}_2 f\rangle\langle f, P^{r^\varphi_*\zeta}_1 f_4\rangle
\,,
\end{align*}
where all scalar products are in $\K$ and $P^{r^\varphi_*\zeta}_1,P^{r^\varphi_*\zeta}_2$ denote the generators of $s\mapsto u_{r^\varphi_*\zeta}(s)$. By considering $f_4$ orthogonal to $P_1^\zeta f$ and $P_1^{r^\varphi_*\zeta}f$, we see that for $B(f)_{\zeta,\kappa}=B(f)_{r^\varphi_*\zeta,\kappa}$ it is necessary that $\langle f_1,(P^\zeta_j-P^{r^\varphi_*\zeta}_j)f\rangle=0$. But varying $f,f_1$ within the specified limitations gives dense subspaces in $\K$, {\em i.e.} we must have $P^\zeta_j=P^{r^\varphi_*\zeta}_j$. This implies that translations in a spacelike direction are represented trivially on the Dirac field, which is not compatible with its locality and covariance properties.

So we conclude that the deformed field operator $B(f)_{r^\varphi_*\zeta,\kappa}$ is not independent of $\varphi$ for $\kappa\neq0$, and hence the cyclicity assumption is not valid for $\kappa\neq0$. Since on the other hand $\Om$ is cyclic for ${\fF(\C)_0}''$ by the Reeh-Schlieder property of $\om$, and a unitary $V$ leaving $\Om$ invariant and mapping $\fF(\C)_0$ onto $\fF(\C)_\kappa$ would preserve this property, we have established that the nets $\fF_0$ and $\fF_\kappa$, $\kappa\neq0$, are not equivalent.
\end{proof}
%==============================================

%%%%%%%%%%%%%%%%%%%%%%%%%%%%%%%%%%%%%%%%%%%%%%%%%%%%%%%%%%%%
\chapter{Quantum Field Theories on de Sitter Spacetime}
\label{ch:deSitter}
%%%%%%%%%%%%%%%%%%%%%%%%%%%%%%%%%%%%%%%%%%%%%%%%%%%%%%%%%%%%

\noindent
In this chapter we apply the warped convolution deformation procedure to quantum field theories with global $\Uone$ gauge symmetry on four-dimensional de Sitter spacetime. This spacetime does not belong to the class considered in chapter \ref{ch:Cosmological}, since it does not admit two linearly independent commuting Killing vector fields. We use a combination of external and internal symmetries, consisting of a one-parameter group of boosts associated with a wedge and the gauge symmetry, as an $\Rl^2$-action to define the deformation. The resulting theory is wedge-local and unitarily inequivalent to the undeformed one for a class of theories, including the free charged Dirac field.
\\
\\
This chapter structured as follows. In Section \ref{sec:deSitterSpacetime} the basic notions concerning the geometry and causal structure of de Sitter spacetime are recalled and we discuss the de Sitter group together with its universal covering. After that, the covariance and inclusion properties of wedges in de Sitter space are studied. 

In Section \ref{sec:DeformationsQFT} we consider quantum field theories with global gauge symmetry within the algebraic setting (field nets) and we show how to reconstruct a wedge-local field net from an inclusion of two C$^*$-algebras, which are in a suitable relative position to a wedge. Then the warped convolution deformation is applied to a field net with global $\Uone$ gauge symmetry  and the properties of the resulting theory are studied. 

In Section \ref{subsec:CARnets} a particular class of field nets is investigated in more detail, namely, nets of $\CAR$-algebras. For these theories the deformed operators can be computed explicitly. The fixed-points of the deformation map are determined and it is shown that the deformed and undeformed field nets are non-isomorphic. In the end we comment on warped convolutions in terms of purely external and internal symmetries using other Abelian subgroups of the de Sitter and gauge group.

\newpage
%===========================================================
\section{de Sitter spacetime} 
\label{sec:deSitterSpacetime}
%===========================================================
%===========================================================
\subsection{Geometry and causal structure} 
\label{subsec:deSitterGeometryCausalStructure}
%===========================================================
The de Sitter spacetime $(M,\gdS)$ is a vacuum solution of Einstein's equation with positive cosmological constant. 
It is maximally symmetric, so it admits 10 Killing vector fields, which is the maximum number for a spacetime of dimension four. It is also globally hyperbolic, so the Cauchy problem for partial differential equations of hyperbolic type, such as the Klein-Gordon and Dirac equation, is well-posed. Furthermore, it is a special case of the Friedmann-Robertson-Walker spacetimes which describe a spatially homogeneous and isotropic universe and it plays a prominent role in many inflationary scenarios for the early universe \cite{Linde09}.

Most conveniently it can be represented as the embedded submanifold
\begin{equation*}
	M=\bl\{x\in\Rl^5: x\cdot x=-1\br\} 
\end{equation*}
of five-dimensional Minkowski space $(\Rl^5,\eta)$. The signature of $\eta$ is $(1,-1,-1,-1,-1)$ and the de Sitter radius is fixed to one. The metric $\gdS$ on $M$ is the induced metric from the ambient space, {\it i.e.} $\gdS=\iota^*\eta$, where $\iota:M\hookrightarrow \Rl^5$ is the embedding map. We use the ambient space notation to parametrize the de Sitter hyperboloid, so we write $x=(x^0,x^1,\vec{x})$, $\vec{x}=(x^2,x^3,x^4)$ for points in $M$, subject to the relation $(x^0)^2-\sum_{k=1}^4(x^k)^2=-1$, where $\{x^\mu:\mu=0,\dots, 4\}$ is a Cartesian coordinate system of $\Rl^5$.

Since the metric on $M$ is the induced metric from the ambient Minkowski space, the causal structure is also inherited. Hence points in $(M,\gdS)$ are called timelike, spacelike or null related, if they are so as points in $(\Rl^5,\eta)$, respectively. We fix a time orientation in $(M,\gdS)$, such that $(1,0,0,0,0)\in\Rl^5$ is future-directed. The interior of the causal complement of a spacetime region $\OO\subset M$ is denoted by $\OO'$.

For the generators of the Clifford algebra which is associated with the quadratic form $x\cdot x$ on the vector space $\Rl^5$ we use the representation \cite{Gazeau07}
\begin{equation*}
\gamma_0=
\begin{pmatrix}
 1 & 0\\
 0 & -1
\end{pmatrix},\qquad
\gamma_1=
\begin{pmatrix}
 0 & 1\\
 -1 & 0
\end{pmatrix},\qquad
\gamma_k=
\begin{pmatrix}
 0 & e_k\\
 e_k & 0
\end{pmatrix},\, k=2,3,4,
\end{equation*}
where $e_k=(-1)^k\sigma_{k-1}$ and $\sigma_{k-1}$ are the Pauli matrices. We have $\{\gamma_\mu,\gamma_\nu\}=2\eta_{\mu\nu}\cdot 1$ and $\{1,e_1,e_2,e_3\}$ is a basis of the quaternions $\Quat$. Note that this representation is in fact not faithful, since $i\gamma_0\cdots\gamma_4=1$. Similar to the case of four-dimensional Minkowski space, where points are parametrized by hermitian $2\times 2$ matrices, we parametrize points on the de Sitter hyperboloid by $2\times 2$ quaternionic matrices. This parametrization is useful for the discussion of the covering group of the de Sitter group later on. Define
\begin{equation*}
 M\ni (x^0,x^1,\vec{x})=x\longmapsto
 \undertilde{x}:=\sum_{\mu=0}^4x^\mu\gamma_\mu=
\begin{pmatrix}
 x^0 & -q \\
 \overline{q} & -x^0
\end{pmatrix}
\in\Mat(2,\Quat),
\end{equation*}
where $\overline{q}=(x^1,-\vec{x})$ is the quaternionic conjugate of $q=(x^1,\vec{x})$. Conversely, every $2\times 2$ matrix of the above form determines a point in de Sitter space via
\begin{equation*}
x^\mu=\frac{1}{4}\Tr(\gamma_\mu \undertilde{x}),\quad \mu=0,\dots,4.
\end{equation*}
The map $x\mapsto \undertilde{x}$ defines an isomorphism between $M$ and $\mathrm{H}(2,\Quat)\gamma_0\subset\Mat(2,\Quat)$, where $\mathrm{H}(2,\Quat)$ are the hermitian $2\times 2$ matrices over $\Quat$. Furthermore, there holds 
$(x\cdot x)1= \undertilde{x}^*\gamma_0\undertilde{x}\gamma_0$, where $\undertilde{x}^*$ is the transpose of the quaternionic conjugate of $\undertilde{x}$.

%===========================================================
\subsection{The de Sitter group and its covering}
\label{subsec:DeSitterGroup}
%===========================================================
The isometry group of de Sitter spacetime $(M,\gdS)$ is 
\begin{equation*}
 \mathrm{O}(1,4)=\{\Lambda\in\Mat(5,\Rl):\Lambda^T\eta\Lambda=\eta\}
\end{equation*}
and its action on $M$ is given by the action of the Lorentz group in the ambient Minkowski space.
This group is a ten-dimensional, non-compact, non-connected and real Lie group which has four connected components. The connected component which contains the identity is denoted by $\L_0:=\SO(1,4)_0$. This group is called de Sitter group (proper orthochronous Lorentz group) and its elements preserve the orientation and time orientation of $(M,\gdS)$. 

Since we also want to treat quantum fields with half-integer spin we consider the two-fold (and universal) covering of $\L_0$, which is the spin group $\widetilde{\L}_0:=\Spin(1,4)$. Hence there exists a short exact sequence of group homomorphisms
\begin{equation*}
 1\lra\ker(\pi)=\{\pm 1\}\lra \widetilde{\L}_0\overset{\pi}{\lra} \L_0\lra 1.
\end{equation*}
There holds $\L_0\cong\widetilde{\L}_0/\{\pm 1\}$ and $\widetilde{\L}_0$ is simply connected. Note that the Lie group $\Spin(1,4)$ is isomorphic to the pseudo-symplectic group \cite{Takahashi63}
\begin{equation*}
	\Sp(1,1)=\Bl\{
\begin{pmatrix}
	a & b\\
	c & d 
\end{pmatrix}
\in \Mat(2,\Quat):\;\bar{a}b=\bar{c}d,\; |a|^2-|c|^2=1,\; |d|^2-|b|^2=1\Br\}.
\end{equation*}
Equivalently, $h\in\Sp(1,1)$ if and only if $h^*\gamma_0h=\gamma_0$. In this representation the covering homomorphism $\pi:\widetilde{\L}_0\ra \L_0$ is given by
\begin{equation*}
 (\pi(h))_{\mu\nu}=\frac{1}{4}\Tr(\gamma_\mu h\gamma_\nu h^{-1}),\quad h\in\tLo
\end{equation*}
and $\tLo$ acts on $M$ by conjugation $\undertilde{x}\mapsto h\undertilde{x}h^{-1}$.

%===========================================================
\subsection{de Sitter wedges} 
\label{subsec:wedges}
%===========================================================
Now we discuss the typical localization regions of the deformed quantum fields from Section \ref{sec:DeformationsQFT}. In \cite{BorchersBuchholz:1999} a de Sitter wedge is defined as the causal completion of the worldline of a uniformly accelerated observer in de Sitter space. Equivalently, they can be characterized as intersections of wedges in the ambient Minkowski space \cite{ThomasWichmann:1997} and the de Sitter hyperboloid.\label{wedgedeSitter} Hence we specify a reference wedge
\begin{equation*}
 W_0:=\{x\in\Rl^5: x^1>|x^0|\}\cap M
\end{equation*}
and define the family of wedges $\W$ as the set of all de Sitter transforms of $W_0$:\label{familywedgesdeSitter}
\begin{equation*}
 \W:=\{h W_0:h\in\L_0\}.
\end{equation*}
By definition, $\L_0$ acts transitively on $\W$. Each wedge $W\in\W$ has an attached edge\label{edgedeSitter} $E_W$ which is a two-sphere. We have $E_{W_0}=\{x\in M:x^0=x^1=0\}$ and $E_W=h E_{W_0}$ for $W=h W_0$. The wedge $W$ coincides with a connected component of the causal complement of the edge $E_W$ \cite{BorchersBuchholz:1999}. For the stabilizer of the wedge $W$ we write 
$\L_0(W):=\{h\in\L_0:h W=W\}$.

From the properties of wedges in $(\Rl^5,\eta)$ follows that the causal complement of a wedge is again a wedge and that every $W\in\W$ is causally complete, {\it i.e.} $W''=(W')'=W$.  
Furthermore, the family $\W$ is causally separating, so given spacelike separated double cones $\OO_1,\OO_2\subset M$, there exists a $W\in\W$ such that $\OO_1\subset W\subset \OO_2\,'$ (see \cite{ThomasWichmann:1997}). 
\\
\\
For every $W\in\W$ there exists a one-parameter group\label{onepargroupwedge} $\Gamma_W=\{\Lambda_W(t)\in\L_0:t\in\Rl\}$, such that each $\Lambda_W(t),t\in\Rl$ maps $\W$ onto $\W$ and $\Lambda_W(t) W=W$ for all $t\in\Rl$. Moreover 
\begin{equation}
\label{eq:GammaCovarianceL0}
 \Lambda_{hW}(t)=h\Lambda_W(t) h^{-1},\quad h\in\L_0,\; t\in\Rl.
\end{equation}
Associated with $\Gamma_W$ is a future-directed Killing vector field $\xi_W$ in the wedge $W$ and the worldline from which the wedge is constructed is an integral curve of (a portion of) this vector field. 
Furthermore, for every $W\in\W$ there exists a reflection $j_W\in \L_0$ which maps $\W$ onto $\W$ and satisfies 
\begin{equation}
\label{eq:WedgeReflection}
 j_W W=W',\quad j_{hW}=hj_Wh^{-1},\quad h\in\L_0.
\end{equation}
Since $\L_0$ acts transitively on $\W$ we only need to specify these maps for $W_0$. We choose
\begin{equation}
\label{eq:GammaW0}
	\Lambda_{W_0}(t):=\left( 
	\begin{array}{ccc}
		\cosh(2\pi t) & \sinh(2\pi t) & 0\\
		\sinh(2\pi t) & \cosh(2\pi t) & 0\\
		0&0& 1_3 
	\end{array}
	\right),\quad j_{W_0}(x^0,x^1,\vec{x}):=(x^0,-x^1,-\vec{x})
\end{equation}
and note that $\xi_{W_0}=x^1\del_{x^0}+x^0\del_{x^1}$ is the associated Killing vector field.
\begin{remark}
 Within the context of applications of Tomita-Takesaki modular theory in quantum field theory the standard choice for the reflection  is $(x^0,x^1,\vec{x})\mapsto (-x^0,-x^1,\vec{x})$, which is an element of the extended symmetry group $\L_+=\L_0\rtimes \Zl_2$. In this thesis we have no intention to use these techniques and the choice (\ref{eq:GammaW0}) appears to be more natural since we restrict our considerations to $\L_0$. However, all of our results can be generalized to the group $\L_+$ in a straightforward manner.
\end{remark}

The following lemma collects the basic properties of these maps.

%%%%%%%%%%%%%%%%%%%%%%%%%%%%%%%%%%%%%%%%%%%%
\begin{lemma}
\label{lem:Gamma-J}
 Let $W\in\W$ and $\Lambda_W(t)\in \Gamma_W$, $j_W$ be as above. Then
\begin{enumerate}
 \item[a)] $h\Lambda_W(t)h^{-1}=\Lambda_W(t),\; h\in\L_0(W),t\in\Rl$,
 \item[b)] $j_W\Lambda_W(t)j_W=\Lambda_W(-t),\; t\in\Rl$,
\end{enumerate}
\end{lemma}
%%%%%%%%%%%%%%%%%%%%%%%%%%%%%%%%%%%%%%%%%%%%
\begin{proof}
	a): For $h\in\L_0(W)$ holds $\Lambda_W(t)=\Lambda_{hW}(t)=h\Lambda_W(t)h^{-1}$ for all
	$t\in\Rl$ by (\ref{eq:GammaCovarianceL0}). b): This follows from $j_{W_0}\Lambda_{W_0}(t)j_{W_0}=\Lambda_{W_0}(-t)$ and (\ref{eq:GammaCovarianceL0}), (\ref{eq:WedgeReflection}).
\end{proof}
%%%%%%%%%%%%%%%%%%%%%%%%%%%%%%%%%%%%%%%%%%%%

The stabilizer of $W$ has the form $\L_0(W)=\Gamma_W\times \SO(3)$, where $\SO(3)$ are rotations in $E_W$. Hence $\Gamma_W$ coincides with the center of $\L_0(W)$. From b) follows that the Killing vector fields associated with $W$ and $W'$ differ only by temporal orientation.

The following lemma shows that the possible causal configurations of wedges are very much constrained in de Sitter space.

%%%%%%%%%%%%%%%%%%%%%%%%%%%%%%%%%%%%%%%%%%%%
\begin{lemma}
\label{lem:WedgeInclusions}
Let $W_1,W_2\in\W$ and $W_1\subset W_2$. Then $W_1=W_2$.
\end{lemma}
%%%%%%%%%%%%%%%%%%%%%%%%%%%%%%%%%%%%%%%%%%%%
\begin{proof}
The wedges $W_1,W_2$ can be written as $W_k=M\cap \widetilde{W}_k$, $k=1,2$, where $\widetilde{W}_k$ is a wedge in the ambient Minkowski space. Since the causal closure of $W_k$ in $\Rl^5$ coincides with $\widetilde{W}_k$, there follows $\widetilde{W}_1\subset \widetilde{W}_2$ from $W_1\subset W_2$. As the edges $E_{\widetilde{W}_k}$ both contain the origin, there follows $E_{\widetilde{W}_1}=E_{\widetilde{W}_2}$ and also $E_{W_1}=E_{W_2}$ since $E_{W_k}=M\cap E_{\widetilde{W}_k}$. The assertion $W_1=W_2$ follows from the assumption $W_1\subset W_2$ together with the fact that $W_k$ is a connected component of the causal complement of $E_{W_k}$.
\end{proof}
%%%%%%%%%%%%%%%%%%%%%%%%%%%%%%%%%%%%%%%%%%%%

%-------------------------------------------
\begin{remark}
All the previous statements carry over to the covering $\tLo$ in a straightforward manner. Define an action of $\tLo$ on $\W$ with the covering homomorphism
\begin{equation}
\label{eq:CoveringWedgeAction}
 h W:=\pi(h) W,\quad h\in\tLo,\, W\in\W,
\end{equation}
which is transitive, since $\L_0$ acts transitively. The one-parameter group $\Gamma_W\subset\L_0$ lifts to a {\em unique} one-parameter group $\widetilde{\Gamma}_W\subset\tLo$ and for its elements we write $\lambda_W(t),t\in\Rl$. Again, since $\tLo$ acts transitively on $\W$, we only need to specify these maps for $W_0$. We have \cite[p.368]{Takahashi63}
\begin{equation*}
 \lambda_{W_0}(t)=
\begin{pmatrix}
 \cosh(\pi t) & \sinh(\pi t)\\
 \sinh(\pi t) & \cosh(\pi t)
\end{pmatrix}.
\end{equation*}
Clearly, $\lambda_W(t)W=W$ and $\lambda_{hW}=h\lambda_{W}h^{-1}$ for all $h\in\tLo$, $W\in\W$ with respect to the action (\ref{eq:CoveringWedgeAction}). For the lift of the reflection $j_{W_0}$ we choose 
\begin{equation*}
 \bj_{W_0}:=
\begin{pmatrix}
 1 & 0\\
 0 & -1
\end{pmatrix}.
\end{equation*}                                                                                                                                            
Again, $\bj_W W=W'$ and $\bj_{hW}=h\bj_{W_0}h^{-1}$ for all $h\in\tLo$, $W\in\W$ with respect to (\ref{eq:CoveringWedgeAction}). Hence analogous statements as in Lemma \ref{lem:Gamma-J} hold for $\widetilde{\Gamma}_W$ with $\L_0$ replaced by $\tLo$, {\it i.e.} 
\begin{equation}
\label{eq:CovarianceGammaJ}
	h\lambda_W(t)h^{-1}=\lambda_W(t),\qquad \bj_W\lambda_W(t)\bj_W=\lambda_W(-t),
\end{equation}
for all $h\in\tLo(W):=\{h\in\tLo:hW=W\}$ and $t\in\Rl$.
\end{remark}
%-------------------------------------------

%===========================================================
\section{Deformations of quantum field theories on de Sitter spacetime} 
\label{sec:DeformationsQFT}
%===========================================================
%===========================================================
\subsection{Field nets} 
\label{subsec:FieldNetsOnDS}
%===========================================================
We use the algebraic formulation of quantum field theory on curved spacetimes \cite{Dimock:1980} adapted to the concrete case of de Sitter space \cite{BorchersBuchholz:1999} (see also section \ref{sec:AQFT}).
To this end, we consider a C$^*$-algebra $\fF$ (field algebra) whose elements are physically interpreted as (bounded functions of) quantum fields on $M$. We equip $\fF$ with a local structure and focus on localization in wedges, since this turns out to be stable under the deformation. Hence we associate to each $W\in\W$ a C$^*$-subalgebra $\fF(W)\subset \fF$. Due to the trivial inclusion properties of wedges in de Sitter space (see Lemma \ref{lem:WedgeInclusions}) the usual isotony condition reduces to well-definedness of $W\mapsto\fF(W)$.\label{fieldnetdesitter}

We assume that there exists a strongly continuous representation $\alpha$ of $\tLo$ by automorphisms on $\fF$, such that
\begin{itemize}
 \item[1)] (De Sitter Covariance): for all $h\in\tLo$, $W\in\W$ holds 
\begin{equation*}
\label{eq:DeSitterCovariance}
	\alpha_h(\fF(W))=\fF(hW). 
\end{equation*}
\end{itemize}

\noindent
Furthermore, we assume that there is a Lie group $\G$ (global gauge group) and a strongly continuous representation $\sigma$ of $\G$ by automorphisms on $\fF$, such that
\begin{itemize}
 \item[2)] (Gauge Invariance): for all $g\in \G$, $h\in\tLo$, $W\in\W$ holds 
\begin{equation}
\label{eq:GaugeInvariance}
	\sigma_g(\fF(W))=\fF(W),\qquad \sigma_g\circ\alpha_h=\alpha_h\circ \sigma_g. 
\end{equation}
\end{itemize}

\noindent
We assume that there exists a distinguished element $g_0\in \G$ such that $\gamma:=\sigma(g_0)$ satisfies
\begin{equation}
\label{eq:BFauto} 
	\gamma^2=\id.
\end{equation}
This (grading) automorphism can be used to separate an operator $F\in \fF(W)$ into its Bose($+$) and Fermi($-$) part via $F_\pm:=(F\pm\gamma(F))/2$. 

%%%%%%%%%%%%%%%%%%%%%%%%%%%%%%%%%%%%%%%
\begin{remark}
	For convenience, we assume that the datum $(\{\fF(W):W\in\W\},\alpha,\sigma,\gamma)$ is faithfully and covariantly represented on a Hilbert space $\HS$. So to each $\fF(W)$ corresponds a norm-closed $*$-subalgebra of $\BH$ and the automorphisms $\alpha,\sigma,\gamma$ are implemented by the adjoint action of unitary operators $U,V,Y$ on $\HS$, respectively. Note that this is only a slight loss of generality since we can either use the (universal) covariant representation which exists for every C$^*$-dynamical system (see \cite{BuchholzLechnerSummers:2010, DappiaggiLechnerMorfa-Morales10} and references therein) or we work in the GNS-representation of a de Sitter- and gauge-invariant state. In the former case we assume that $\HS$ is separable, as it is the case in a variety of concrete models.
\end{remark}
%%%%%%%%%%%%%%%%%%%%%%%%%%%%%%%%%%%%%%%

We assume that the grading satisfies $Y^2=1$. With the operator $Y$ a unitary twisting map $Z$ is defined to treat the (anti)commutation relations between the Bose/Fermi parts of  fields on the same footing \cite{DoplicherHaagRoberts:1969}. Let $Z:=(1-iY)/\sqrt{2}$ and 
\begin{equation*}
 \fF(W)^t:=Z\fF(W)Z^{-1}.
\end{equation*}
The map $F\mapsto ZFZ^{-1}$ is an isomorphism of $\fF(W)$ and we have \cite{Foit:1983}
\begin{equation*}
 \fF(W)^{tt}=\fF(W),\quad \fF(W)^t{}'=\fF(W)'{}^t,\quad W\in\W,
\end{equation*}
where the commutant is understood as the relative commutant in $\fF$. Locality is now formulated in the following way
\begin{itemize}
 \item[3)] (Twisted Locality): $\fF(W)\subset \fF(W')^t{}',\;W\in\W$.
\end{itemize}
Twisted locality is equivalent to the ordinary (anti)commutation relations between the Bose/Fermi parts $F_\pm$ of fields, {\it i.e.} $[F_+,G_\pm]=[F_\pm,G_+]=\{F_-,G_-\}=0$ for $F\in\fF(W)$, $G\in\fF(W')$, $W\in\W$ (see \cite{DoplicherHaagRoberts:1969}).

For later reference we define the joint action $\tau:\tLo\times \G\ra \Aut(\fF)$ of the external and internal symmetry group on $\fF$ by
\begin{equation}
\label{eq:jointautomorphism}
 \tau_{h,g}:=\alpha_h\circ \sigma_g,\quad h\in\tLo,\; g\in \G.
\end{equation}
The unitary which implements this action is $U(h)V(g)=:\U(h,g)$.

%%%%%%%%%%%%%%%%%%%%%%%%%%%%%%%%%%%%%%%
\begin{remark}
A datum $(\{\fF(W):W\in\W\},\alpha,\sigma,\gamma)$ which satisfies conditions 1) $-$ 3) is referred to as a \emph{wedge-local field net}. We simply write $\fF$ to denote it, if no confusion can arise. Examples are nets of $\CAR$-algebras with gauge symmetry, such as the free charged Dirac field (see section \ref{subsec:CARnets}).
\end{remark}
%%%%%%%%%%%%%%%%%%%%%%%%%%%%%%%%%%%%%%%

%%%%%%%%%%%%%%%%%%%%%%%%%%%%%%%%%%%%%%%
\begin{remark}
 Given a field net, the net of observables is defined as
\begin{equation*}
	\Aa(W):=\{F\in\fF(W):\sigma_g(F)=F,\; g\in \G\},
\end{equation*}
so observables form the gauge-invariant part of the field net. 
\end{remark}
%%%%%%%%%%%%%%%%%%%%%%%%%%%%%%%%%%%%%%%

Similar to the construction of wedge-local nets from wedge triples (see section \ref{sec:CAQFT}), it is possible to define a wedge-local field net in terms of an inclusion of just two C$^*$-algebras which are in a suitable relative position to $W_0$. This point of view will be advantageous for the warped convolution later on, since the deformation of a wedge-local field net amounts to deforming the relative position of one algebra in the other. Following \cite{BuchholzLechnerSummers:2010} we make the following definition.
%%%%%%%%%%%%%%%%%%%%%%%%%%%%%%%%%%%%%%%
\begin{definition}
\label{def:CausalBorchersSystem}
 A {\em causal Borchers system} $(\fF_0,\fF,\alpha,\sigma,\gamma)$, relative to $W_0$, consists of
 \begin{itemize}
  \item an inclusion $\fF_0\subset \fF$ of concrete C$^*$-algebras
  \item commuting representations $\alpha:\tLo\ra \Aut(\fF)$ and $\sigma:\G\ra \Aut(\fF)$ which are unitarily implemented
  \item an automorphism $\gamma$ on $\fF$ which commutes with $\alpha$ and $\sigma$ and satisfies $\gamma^2=\id$
 \end{itemize}
 such that
 \begin{enumerate}
  \item[a)] $\alpha_h(\fF_0)= \fF_0,\; h\in\tLo(W_0)$
  \item[b)] $\alpha_{\bj_{W_0}}(\fF_0)\subset (\fF_0)^t{}'$
  \item[c)] $\sigma_g(\fF_0)=\fF_0$, $g\in \G$.
 \end{enumerate}
\end{definition}
%%%%%%%%%%%%%%%%%%%%%%%%%%%%%%%%%%%%%%%

%%%%%%%%%%%%%%%%%%%%%%%%%%%%%%%%%%%%%%%
\begin{proposition}
\label{prop:fieldquadruple-fieldnet} 
	Let $(\fF_0,\fF,\alpha,\sigma,\gamma)$ be a causal Borchers system relative to $W_0$. Then 
	\begin{equation}
		\label{eq:WT-WLN} W:=hW_0\longmapsto \alpha_h(\fF_0)=:\fF(W), 
	\end{equation}
	defines a wedge-local field net together with $(\alpha,\sigma,\gamma)$. 
\end{proposition}
%%%%%%%%%%%%%%%%%%%%%%%%%%%%%%%%%%%%%%%
\begin{proof}
	We begin by proving well-definedness. From $h_1 W_0=h_2W_0$ follows $h_2^{-1}h_1 W_0= W_0$ and $\alpha_{h_2^{-1}h_1}(\fF_0)= \fF_0$ by assumption a). Hence $\alpha_{h_1}(\fF_0)= \alpha_{h_2}(\fF_0)$ and the assertion follows.
	
	Covariance holds by definition. 

	Twisted locality is proved in a similar way. Let $W'=hW_0$. Since $W=h\bj_{W_0}W_0$ there holds
	\begin{equation*}
	    \fF(W)=\alpha_{h\bj_{W_0}}(\fF_0)
	    \subset\alpha_h((\fF_0)^t{}')
	    =\alpha_h(\fF_0)^t{}'=\fF(W')^t{}',\quad W\in\W,
	\end{equation*}
	where we used condition b), together with the assumption that each $\alpha_h$, $h\in\tLo$ is a homomorphism which commutes with $\gamma$.
	
	The gauge invariance of the local algebras follows immediately:
	\begin{equation*}	
	\sigma_g(\fF(W))=\sigma_g(\alpha_h(\fF_0))=\alpha_h(\sigma_g(\fF_0))=\alpha_h(\fF_0)
	=\fF(W),
	\end{equation*}
	since the representations $\alpha,\sigma$ commute and by assumption c). 
\end{proof}
%%%%%%%%%%%%%%%%%%%%%%%%%%%%%%%%%%%%%%%

\noindent
Note that the converse of this proposition is trivially true. Given a wedge-local field net, then $\fF(W_0)\subset \fF$ satisfies property a) by covariance and b) by twisted locality. Property c) holds by definition.

%%%%%%%%%%%%%%%%%%%%%%%%%%%%%%%%%%%%%%%
\begin{remark}
A causal Borchers system $(\fF_0,\fF,\alpha,\sigma,\gamma)$ is closely connected to the notion of a causal Borchers triple \cite{BuchholzLechnerSummers:2010} on Minkowski spacetime (see also \cite{Lechner10} for the related notion of a wedge triple). In this setting, $\fF_0\subset\BH$ is a von Neumann algebra and $\alpha$ is the adjoint action of a unitary representation $U$ of the Poincar\'e group. In addition one assumes that the joint spectrum of the generators of the translations $U\uhr \Rl^4$ is contained in the closed forward lightcone (spectrum condition) and that $\fF_0$ admits a cyclic and separating vector (existence of a vacuum state). Gauge transformations are absent in this setting since nets of observables are considered.
\end{remark}
%%%%%%%%%%%%%%%%%%%%%%%%%%%%%%%%%%%%%%%

%%%%%%%%%%%%%%%%%%%%%%%%%%%%%%%%%%%%%%%
\begin{remark}
 For the sake of brevity we will write $\fF_0\subset \fF$ to denote a causal Borchers system relative to $W_0$.
\end{remark}
%%%%%%%%%%%%%%%%%%%%%%%%%%%%%%%%%%%%%%%

%===========================================================
\subsection{Deformations of field nets with $\Uone$ gauge symmetry}
 \label{subsec:DeformationsFieldNets}
%===========================================================
Now we apply the warped convolution deformation method to our present setting. Let $\fF_0\subset\fF$ be a causal Borchers system relative to $W_0$. The basic idea is to define a deformation $(\fF_0)_{\xi,\kappa}$ of the small algebra $\fF_0$ using a suitable $\Rl^2$-action (see below) in such a way that $(\fF_0)_{\xi,\kappa}\subset\fF$ is again a causal Borchers system. Then the inclusion $(\fF_0)_{\xi,\kappa}\subset \fF$ gives rise to another wedge-local field net by Proposition \ref{prop:fieldquadruple-fieldnet}. 
\\

\noindent
For the warped convolution we make the further assumption that the gauge group is $\G=\Uone\cong\Rl/2\pi\Zl$. The representation $\sigma$ of $\Uone$ yields a $2\pi$-periodic $\Rl$-action $F\mapsto\sigma_{\exp(is)}(F)$, $s\in\Rl$ by automorphisms on $\fF$. The warped convolution is now defined with the $\Rl^2$-action $\tau^\xi$ coming from the one-parameter group of boosts $\widetilde{\Gamma}_{W_0}\subset \tLo$ and the internal symmetry group:
\begin{equation*}
\Rl^2\ni (t,s)\longmapsto 
\tau_{\lambda_{W_0}(t),\exp(is)}=:\tau^\xi_{t,s}:\fF\lra \fF.
\end{equation*}
Note that $\widetilde{\Gamma}_{W_0}$ implicitly depends on the Killing field $\xi:=\xi_{W_0}$ which is associated with $W_0$ (see Section \ref{subsec:wedges}). We will use the notation
\begin{equation*}
 \lambda_\xi(t):=\lambda_{W_0}(t),\quad U_\xi(t):=U(\lambda_\xi(t)),\quad \U_\xi(t,s):=\U(\lambda_\xi(t),s).
\end{equation*}

Since the warped convolution is defined is terms of oscillatory integrals of operator-valued functions, we first need to specify suitable smooth elements of $\fF$ for which these integrals are well-defined. The joint action (\ref{eq:jointautomorphism}) is a strongly continuous action of the Lie group $\tLo\times \Uone$ which acts automorphically, and therefore isometrically, on $\fF$. The algebra $\fF_0$ is, in general, only invariant under the action of the subgroup $\widetilde{\L}_0(W_0)\times \Uone$. Adapted to the present setting, and following \cite{DappiaggiLechnerMorfa-Morales10}, we consider the following weakened notion of smoothness with respect to the subgroup $\widetilde{\Gamma}_{W_0}\times\Uone$. 

%%%%%%%%%%%%%%%%%%%%%%%%%%%%%%%%%%%%%%%
\begin{definition}
	An operator $F\in\fF$ is called {\it $\xi$-smooth}, if $\Rl^2\ni v\mapsto \tau^\xi_v(F)\in\fF$ is smooth in the norm topology of $\fF$. The set of all $\xi$-smooth operators in $\fF$ is denoted by $\fF^\infty_\xi$.
\end{definition}
%%%%%%%%%%%%%%%%%%%%%%%%%%%%%%%%%%%%%%%

Note that the set $\fF^\infty_\xi$ is a norm-dense $*$-subalgebra of $\fF$ (see \cite{Taylor:1986}). Another ingredient for the definition of the warped convolution is the antisymmetric (real) matrix 
\begin{equation*}
	\theta:= 
	\begin{pmatrix}
		0 & 1\\
		-1 & 0 
	\end{pmatrix}
\end{equation*}
and an arbitrary but fixed real number $\kappa$ which plays the role of a deformation parameter. 

%%%%%%%%%%%%%%%%%%%%%%%%%%%%%%%%%%%%%%%
\begin{definition}
The {\em warped convolution} of an operator $F\in \fF^\infty_\xi$ is defined as 
\begin{equation}
	\label{eq:WarpedConvolution} 
	F_{\xi,\kappa}:=\frac{1}{4\pi^2}\lim_{\eps\ra 0}\int_{\Rl^2\times \Rl^2}dv\, d v'\; e^{-ivv'}\, \chi(\eps v,\eps v')\, \tau^\xi_{\kappa\theta v}(F)\U_\xi(v').
\end{equation}
Here $vv'$ denotes the standard Euclidean inner product of $v,v'\in\Rl^2$ and $\chi\in C^\infty_0(\Rl^2\times \Rl^2)$, $\chi(0,0)=1$ is a cutoff function which is necessary to define this operator-valued integral in an oscillatory sense.  
\end{definition}
%%%%%%%%%%%%%%%%%%%%%%%%%%%%%%%%%%%%%%%

From the results in \cite{BuchholzLechnerSummers:2010} follows that the above limit exists in the strong operator topology of $\BH$ on the dense domain\label{smoothvectorsdeSitter}
\begin{equation*}
	\HS^\infty:=\{\Phi\in\HS: \tLo\times\Uone\ni (h,g) \mapsto \U(h,g)\Phi\in \HS \text{ is smooth in } \|\cdot \|_{\HS}\} 
\end{equation*}
and is independent of the chosen cutoff function $\chi$ within the specified class. The densely defined operator $F_{\xi,\kappa}$ extends to a bounded and smooth operator, which is denoted by the same symbol. For the space of all vectors which are smooth with respect to the representation $\U_\xi$ we write $\HS^\infty_\xi$.\label{xismoothvectorsdeSitter}

Furthermore, it is shown in \cite{BuchholzLechnerSummers:2010} that the warped convolution (\ref{eq:WarpedConvolution}) is closely related to Rieffel deformations of C$^*$-algebras \cite{Rieffel:1992}. In this context one defines, instead of a deformation of the algebra elements, a new product $\times_{\xi,\kappa}$ on $\fF^\infty_\xi$ by
\begin{equation*}
 F\times_{\xi,\kappa}G:=\frac{1}{4\pi^2}\lim_{\eps\ra 0}\int_{\Rl^2\times\Rl^2}d v\,d v' e^{-ivv'}
\chi(\eps v,\eps v') \tau^\xi_{\kappa \theta v}(F)\tau^\xi_{v'}(G).
\end{equation*}
This limit exists in the norm-topology of $\fF$ for all $F,G\in\fF^\infty_\xi$ and $F\times_{\xi,\kappa}G$ is again in $\fF^\infty_\xi$. The completion of $(\fF^\infty_\xi,\times_{\xi,\kappa})$ in a suitable norm yields another C$^*$-algebra \cite{Rieffel:1992}.

The following lemma collects the basic properties of the map $F\mapsto F_{\xi,\kappa}$ and shows that the warped operators form a representation of the Rieffel deformed C$^*$-algebra for a fixed deformation parameter.

%%%%%%%%%%%%%%%%%%%%%%%%%%%%%%%%%%%%%%%
\begin{lemma}[\cite{BuchholzLechnerSummers:2010,DappiaggiLechnerMorfa-Morales10}]
\label{lem:WCbasicproperties}
 Let $F,G\in\fF^\infty_\xi$ and $\kappa\in\Rl$. Then
\begin{enumerate}
 \item[a)] $(F_{\xi,\kappa})^*=(F^*)_{\xi,\kappa}$.
 \item[b)] $F_{\xi,\kappa}G_{\xi,\kappa}=(F\times_{\xi,\kappa}G)_{\xi,\kappa}$.
 \item[c)] If $[\tau^\xi_v(F),G]=0$ for all $v\in\Rl^2$, then $[F_{\xi,\kappa},G_{\xi,-\kappa}]=0$.
 \item[d)] If $[Z\tau^\xi_v(F)Z^*,G]=0$ for all $v\in\Rl^2$, then $[ZF_{\xi,\kappa}Z^*,G_{\xi,-\kappa}]=0$. 
 \item[e)] Let $X\in\BH$ be a unitary which commutes with $\U_\xi(v)$ for all $v\in\Rl^2$. Then 
$XF_{\xi,\kappa}X^{-1}=(XFX^{-1})_{\xi,\kappa}$ and $XF_{\xi,\kappa}X^{-1}$ is $\xi$-smooth.
 
\end{enumerate}
\end{lemma}
%%%%%%%%%%%%%%%%%%%%%%%%%%%%%%%%%%%%%%%
\begin{proof}
 Statements a), b), c), e) were shown in \cite{BuchholzLechnerSummers:2010}. Assertion d) is a consequence of Lemma \ref{lemma:twist} in Chapter \ref{ch:Cosmological} and the fact that $\gamma$ commutes with $\alpha$ and $\sigma$.
\end{proof}
%%%%%%%%%%%%%%%%%%%%%%%%%%%%%%%%%%%%%%%

The next lemma lists the transformation properties of warped operators under the de Sitter and gauge group.

%%%%%%%%%%%%%%%%%%%%%%%%%%%%%%%%%%%%%%%
\begin{lemma}
\label{lem:WCcovariance}
 Let $F\in\fF^\infty_\xi$, $\kappa\in\Rl$, $h\in\tLo$ and $g\in\Uone$. Then
\begin{enumerate}
 \item[a)] $\alpha_h(F)$ is $h_*\xi$-smooth and 
\begin{equation}
\label{eq:DeformedOperatorsDeSitterCovariance}
 \alpha_h(F_{\xi,\kappa})=\alpha_h(F)_{h_*\xi,\kappa},
\end{equation}
where $h_*\xi$ is the push-forward of $\xi$ with respect to $h$.
\item[b)] $\sigma_g(F)$ is $\xi$-smooth and 
\begin{equation}
\label{eq:DeformedOperatorsGaugeCovariance}
 \sigma_g(F_{\xi,\kappa})=\sigma_g(F)_{\xi,\kappa}.
\end{equation}
\end{enumerate}
\end{lemma}
%%%%%%%%%%%%%%%%%%%%%%%%%%%%%%%%%%%%%%%
\begin{proof}
 Statement a) follows from Lemma \ref{lemma:deformed-operators} in Chapter \ref{ch:Cosmological} and the fact that $\alpha$ and $\sigma$ commute. Statement b) follows from Lemma \ref{lem:WCbasicproperties} e). 
\end{proof}
%%%%%%%%%%%%%%%%%%%%%%%%%%%%%%%%%%%%%%%

Now we apply the warped convolution deformation method to a causal Borchers system $\fF_0\subset\fF$. Define
\begin{equation*}
 (\fF_0)_{\xi,\kappa}:=\overline{\{F_{\xi,\kappa}:F\in\fF_0\cap\fF^\infty_\xi\}}^{\|\cdot\|}.
\end{equation*}
The following theorem shows that the inclusion $(\fF_0)_{\xi,\kappa}\subset\fF$ gives rise to a wedge-local field net in the sense of Proposition \ref{prop:fieldquadruple-fieldnet}.

%%%%%%%%%%%%%%%%%%%%%%%%%%%%%%%%%%%%%%%
\begin{theorem}
\label{thm:main}
 Let $(\fF_0)_{\xi,\kappa}$ be as above. Then
\begin{enumerate}
 		\item[a)] $\alpha_h((\fF_0)_{\xi,\kappa})= (\fF_0)_{\xi,\kappa},\; g\in\tLo(W_0)$, 
		\item[b)] $\alpha_{\bj_{W_0}}((\fF_0)_{\xi,\kappa})\subset ((\fF_0)_{\xi,\kappa})^t{}'$,
		\item[c)] $\sigma_g((\fF_0)_{\xi,\kappa})=(\fF_0)_{\xi,\kappa},\; g\in \Uone$.
\end{enumerate}
\end{theorem}
%%%%%%%%%%%%%%%%%%%%%%%%%%%%%%%%%%%%%%%
\begin{proof}
	a): Let $F\in\fF_0\cap\fF^\infty_\xi$ and $h\in\tLo(W_0)$. From (\ref{eq:CovarianceGammaJ}) follows that $h$ commutes with each $\lambda_\xi(t),t\in\Rl$. Hence 
	\begin{equation*}
	    \alpha_h(F_{\xi,\kappa})=\alpha_h(F)_{h_*\xi,\kappa}=\alpha_h(F)_{\xi,\kappa}
	\end{equation*}
	by Lemma \ref{lem:WCcovariance} a) and $\alpha_h(F)\in \fF_0$ by property a) of the undeformed causal Borchers system. Therefore $\alpha_h(F_{\xi,\kappa})\in(\fF_0)_{\xi,\kappa}$ and by taking the norm-closure of $\{F_{\xi,\kappa}:F\in\fF_0\cap\fF^\infty_\xi\}$ the statement $\alpha_h((\fF_0)_{\xi,\kappa})=(\fF_0)_{\xi,\kappa}$ follows.

	b): From Lemma \ref{lem:WCcovariance} a) and (\ref{eq:CovarianceGammaJ}) follows
	\begin{equation}
	\label{eq:proof_locality}
	    \alpha_{{\bj}_{W_0}}(F_{\xi,\kappa})=\alpha_{{\bj}_{W_0}}(F)_{{{\bj}_{W_0}}_*\xi,\kappa}=\alpha_{{\bj}_{W_0}}(F)_{\xi,-\kappa},
	\end{equation}
	together with an elementary substitution in (\ref{eq:WarpedConvolution}). We have $\alpha_{{\bj}_{W_0}}(F)\in(\fF_0)^t{}'$ by property b) of the undeformed causal Borchers system, {\it i.e.} $[Z\alpha_{{\bj}_{W_0}}(F)Z^{-1},G]=0$ for all $G\in\fF_0$. Pick some $G\in\fF_0\cap\fF^\infty_\xi$ and consider its warped convolution $G_{\xi,\kappa}$. We have $[Z\tau_v^\xi(\alpha_{{\bj}_{W_0}}(F))Z^{-1},G]=0$ for all $v\in\Rl^2$ since $\fF_0$ is invariant under $\widetilde{\Gamma}_{W_0}\times\Uone$. Hence 
	\begin{equation*}
	  [Z\alpha_{\bj_{W_0}}(F_{\xi,\kappa})Z^{-1},G_{\xi,\kappa}]=
	  \alpha_{\bj_{W_0}}([ZF_{\xi,\kappa}Z^{-1},\alpha_{\bj_{W_0}}(G_{\xi,\kappa})])=
	  \alpha_{\bj_{W_0}}([ZF_{\xi,\kappa}Z^{-1},G_{\xi,-\kappa}])=0
	\end{equation*}
	by (\ref{eq:proof_locality}) and Lemma \ref{lem:WCbasicproperties} e). By taking the norm-closure of $\{F_{\xi,\kappa}:F\in\fF_0\cap\fF^\infty_\xi\}$ the statement $\alpha_{\bj_{W_0}}((\fF_0)_{\xi,\kappa})\subset ((\fF_0)_{\xi,\kappa})^t{}'$ follows.

	Assertion c) is a consequence of Lemma \ref{lem:WCcovariance} b) and the invariance $\fF_0$ under gauge transformations.
\end{proof}
%%%%%%%%%%%%%%%%%%%%%%%%%%%%%%%%%%%%%%%

%%%%%%%%%%%%%%%%%%%%%%%%%%%%%%%%%%%%%%%
\begin{remark}
 Note that the minus sign which appears in (\ref{eq:proof_locality}) is the main reason why the locality proof works. That this argument is also valid for the extended symmetry group $\L_+$ can be seen in the following way. The reflection $\widehat{j}_{W_0}(x^0,x^1,\vec{x})=(-x^0,-x^1,\vec{x})$ commutes with boosts in the $x^1$-direction. Again, a deformed operator transforms under the lift $\widehat{\bj}_{W_0}$ of $\hat{j}_{W_0}$ according to $\alpha_{\widehat{\bj}_{W_0}}(F_{\xi,\kappa})=F_{\xi,-\kappa}$ since $\widehat{\bj}_{W_0}$ is represented by an {\em anti}unitary operator.
\end{remark}
%%%%%%%%%%%%%%%%%%%%%%%%%%%%%%%%%%%%%%%

%===========================================================
\subsection{Example: Deformations of $\CAR$-nets} 
\label{subsec:CARnets}
%===========================================================
Now we investigate a particular class of wedge-local field nets in more detail, namely, nets of $\CAR$-algebras. The free charged Dirac field is an example thereof. After it is shown that these models fit into the framework of Section \ref{subsec:DeformationsFieldNets}, the properties of the deformed field operators and observables are studied and it is proved that the deformed and undeformed nets are non-isomorphic.

%-----------------------------------------------------------
\subsubsection{The selfdual $\CAR$-algebra}
%-----------------------------------------------------------
We use Araki's selfdual approach to the $\CAR$-algebra \cite{Araki:1971}. Let $H$ be a separable infinite-dimensional complex Hilbert space with inner product $\left<.\,,.\right>$ and let $C$ be an antiunitary involution on $H$, {\it i.e.} there holds $C^2=1$ and $\left<C f_1,C f_2\right>=\left<f_2,f_1\right>$ for all $f_1,f_2\in H$. 
On the $*$-algebra $\CAR_0(H,C)$ which is algebraically generated by symbols $B(f),f\in H$ and a unit $1$, satisfying
\begin{itemize}
 \item[$a)$] $f\mapsto B(f)$ is complex linear,
 \item[$b)$] $B(f)^*=B(C f)$,
 \item[$c)$] $\{B(f_1),B(f_2)\}=\left<C f_1,f_2\right>\cdot 1$,
\end{itemize}
there exists a unique C$^*$-norm satisfying
\begin{equation*}
	\|B(f)\|^2=\frac{1}{2}(\|f\|^2+\sqrt{\|f\|^4-|\left<f,Cf\right>|^2}).
\end{equation*}
Hence, each $B(f)$ is bounded and $f\mapsto B(f)$ is norm-continuous. The C$^*$-completion of $\CAR_0(H,C)$ is denoted by $\CAR(H,C)$. This C$^*$-algebra is simple, so all its representations are faithful or trivial \cite{Araki:1971}. 

If $u$ is a unitary on $H$ which commutes $C$, then $\alpha_u(B(f)):=B(uf)$ defines a $*$-automorphism on $\CAR(H,C)$. We refer to $u$ as Bogolyubov transformation and to $\alpha_u$ as Bogolyubov automorphism.
%-----------------------------------------------------------
\subsubsection{Quasifree representations}
%-----------------------------------------------------------
A state $\omega$ on $\CAR(H,C)$ is called quasifree, if
\begin{align*}
	\omega(B(f_1)\cdots B(f_{2n+1}))&=0\\
	\omega(B(f_1)\cdots B(f_{2n}))
	&=(-1)^{n(n-1)/2}\sum_{\epsilon}\mathrm{sgn}(\epsilon)\prod_{k=1}^n\omega(B(f_{\epsilon(k)})B(f_{\epsilon(k+n)}))
\end{align*}
holds for all $n\in\Nl$, where the sum runs over all permutations $\epsilon$ of $\{1,\dots,2n\}$ satisfying
\begin{equation*}
	\epsilon(1)<\dots <\epsilon(n),\quad \epsilon(k)<\epsilon(k+n),\quad k=1,\dots,n.
\end{equation*}
Let $S$ be a bounded linear operator on $H$ satisfying
\begin{equation*}
	\label{eq:CAR_Soperator}
	S=S^*,\quad 0\le S\le 1,\quad CSC=1-S.
\end{equation*}
In \cite[Lem.3.3]{Araki:1971} it is shown that for every such $S$ there exists a unique quasifree state $\omega_S$ satisfying 
\begin{equation*}
\label{eq:QuasifreeState}
	\omega_S(B(f)B(g))=\left<Cf,Sg\right>.
\end{equation*}
Conversely, every quasifree state on $\CAR(H,C)$ gives rise to such an operator \cite[Lem.3.2]{Araki:1971}. Hence quasifree states can be parametrized by this class of operators. 

Let $\omega_S$ be a quasifree state. For the GNS-triple associated with $(\CAR(H,C),\omega_S)$ we write $(\HS_S,\pi_S,\Omega_S)$. If a Bogolyubov transformation $u$ commutes with $S$, then the associated Bogolyubov automorphism can be implemented, {\em i.e.} there exists a unitary operator $U_S$ on $\HS_S$, such that
\begin{equation*}
	\pi_S(\alpha_u(F))=U_S\pi_S(F)U_S^{-1},\quad U_S\Omega_S=\Omega_S
\end{equation*}
holds for all $F\in\CAR(H,C)$ (see \cite[Lem.4.2]{Araki:1971}).

Fock states are a particular class of quasifree states where $S=P$ is a projection. The GNS Hilbert space $\HS_P$ is the Fermionic Fock space over $PH$
\begin{equation*}
	\HS_P=\Cl\oplus \bigoplus_{n\ge 1}\wedge^n PH,
\end{equation*}
where $\wedge^n PH$ denotes the antisymmetrization of the $n$-fold tensor product of $PH$, the cyclic vector $\Omega_P$ is the Fock vacuum in $\HS_P$ and
\begin{equation*}
 \pi_P(B(f))=a^\dagger(PC f)+a(Pf),
\end{equation*}
with the standard Fermi creation and annihilation operators $a^\#(Pf)$ on $\HS_P$.
Two representations $(\HS_P,\pi_P)$ and $(\HS_{P'},\pi_{P'})$ are unitarily equivalent, if and only if $P-P'$ is Hilbert-Schmidt \cite{ShaleStinespring64, Araki:1971}. As a consequence, a Bogolyubov transformation $u$ is implementable, if and only if $[u,P]$ is Hilbert-Schmidt.

%-----------------------------------------------------------
\subsubsection{Nets of $\CAR$-algebras}
%-----------------------------------------------------------
In order to introduce charges and global gauge transformations we double the Hilbert space $\hs:=H\oplus H$ and define the antiunitary involution
\begin{equation*}
 \cc:=
\begin{pmatrix}
 0 & C\\
 C & 0
\end{pmatrix}.
\end{equation*}
We denote elements in $\hs$ by $f_+\oplus f_-$, $f_\pm\in H$ and for the inner product in $\hs$ we write $(.\,,.)$. Applying Araki's construction to $(\hs,\cc)$ yields again a C$^*$-algebra $\CAR(\hs,\cc)$. The unitary operators
\begin{equation}
\label{eq:CAR_GaugeTransformationGeneral}
	v(s)(f_+\oplus f_-):=e^{is}f_+\oplus e^{-is}f_-,\quad s\in\Rl,\;f_\pm\in H
\end{equation}
commute with $\cc$ so there exists a representation $\sigma:\Uone\ra \Aut(\CAR(\hs,\cc))$, such that
\begin{equation}
\label{eq:CAR_gaugetransformation}
 \sigma_{s}(B(f))=B(v(s)f).
\end{equation}
We assume that there exists a unitary representation $u$ of $\tLo$ on $\hs$ which commutes with $\cc$ so that there exists a representation $\alpha:\tLo\ra\Aut(\CAR(\hs,\cc))$ satisfying
\begin{equation*}
\label{eq:CAR_covariance}
 \alpha_h(B(f))=B(u(h)f).
\end{equation*}
For the representers of the subgroup $\widetilde{\Gamma}_{W_0}$ we write $u_\xi(t):=u(\lambda_\xi(t)),\, t\in\Rl$.
\begin{remark}
 The picture in terms of spinors and cospinors is obtained by setting 
\begin{equation}
\label{eq:CAR_cospinors}
 \Psi(f_-):=B(0\oplus f_-),\quad \Psi^\dagger(f_+):=B(f_+\oplus 0).
\end{equation}
There holds $\Psi(f_-)^*=\Psi^\dagger(Cf_-)$ and  $\Psi^\dagger(f_+)^*=\Psi(Cf_+)$. From the linearity of $f\mapsto B(f)$ follows that (co)spinors transform according to
\begin{equation}
\label{eq:CAR_cospinors_gauge}
 \sigma_s(\Psi(f_-))=e^{-is}\Psi(f_-),\quad \sigma_s(\Psi^\dagger(f_+))=e^{is}\Psi^\dagger(f_+)
\end{equation}
under gauge transformations.
\end{remark}

Now we come to the net structure of the theory. Let $\hs_0\subset \hs$ be a complex linear subspace satisfying $\cc \hs_0\subset \hs_0$ and
\begin{itemize}
 \item[$i)$] $u(h)\hs_0=\hs_0,\;h\in\tLo(W_0)$,
 \item[$ii)$] $u(\bj_{W_0})\hs_0\subset (\hs_0)^\perp$,
 \item[$iii)$] $v(s)\hs_0=\hs_0,\; s\in\Rl$,
\end{itemize}
where $(\hs_0)^\perp$ is the orthogonal complement of $\hs_0$. 
\begin{remark}
In concrete models this space is explicitly given and can be constructed by different methods. In the case of the free charged Dirac field the space $\hs_0$ can be defined as the set of (Fourier-Helgason transforms of) spinor-valued testfunctions on $M$ which are localized in the wedge $W_0$ (see \cite{BartesaghiGazeauMoschellaTakook} and \cite{BrosMoschella95} for the scalar free field case) or one considers smooth sections of the Dirac bundle over $M$ modulo the kernel of the causal propagator which is associated with the Dirac equation \cite{Dimock:1982, Sanders:2010}. Since this space is constructed from testfunctions it is clear that conditions $i),ii)$ and $iii)$ are satisfied.  
\end{remark}

It is an easy exercise to show that the above conditions imply that 
\begin{equation*}
 W:=hW_0\mapsto u(h)\hs_0=:\hs(W)
\end{equation*}
is an isotonous, $\tLo$-covariant, wedge-local and gauge-invariant net of complex Hilbert spaces in the sense of \cite{BaumgaertelJurkeLledo94}. Hence it is an immediate consequence of Araki's construction that 
\begin{equation}
\label{eq:CARnet}
	W\mapsto \CAR(\hs(W),\cc)=:\fF(W)\subset\fF:=\CAR(\hs,\cc)
\end{equation}
is a wedge-local field net. Equivalently, from conditions $i)$, $ii)$ and $iii)$ follows that the inclusion $\CAR(\hs_0,\cc)=:\fF_0\subset\fF=\CAR(\hs,\cc)$ satisfies
\begin{itemize}
 \item $\alpha_h(\fF_0)=\fF_0,\; h\in\tLo(W_0)$
 \item $\alpha_{\bj_{W_0}}(\fF_0)\subset (\fF_0)^t{}'$
 \item $\sigma_s(\fF_0)=\fF_0,\; s\in\Rl$
\end{itemize}
and $hW_0\mapsto \alpha_h(\fF_0)$ defines a wedge-local field net by Proposition \ref{prop:fieldquadruple-fieldnet}) which coincides with (\ref{eq:CARnet}). 

%%%%%%%%%%%%%%%%%%%%%%%%%%%%%%%%%%%%%%%
\begin{remark}
    Observables in this net  are polynomials of $\Psi(f_-)\Psi^\dagger(f_+)$ which are manifestly gauge-invariant. The quasilocal algebra generated by them is denoted as $\Aa$. 
\end{remark}
%%%%%%%%%%%%%%%%%%%%%%%%%%%%%%%%%%%%%%%

From the above discussion it is clear that $W\mapsto \fF(W)$ complies with the general assumptions of Section \ref{subsec:FieldNetsOnDS}. As we mentioned before, the algebra $\fF$ contains a norm-dense $*$-subalgebra of smooth elements $\fF^\infty_\xi$. These can be explicitly constructed by smoothening out any element $F\in \fF$ with a smooth and compactly supported function $f\in C^\infty_0(\widetilde{\Gamma}_{W_0}\times \Uone)$ via
\begin{equation*}
 F_f:=\int_{\widetilde{\Gamma}_{W_0}\times \Uone}d(h,g)\tau^\xi_{h,g}(F)f(h,g)
\end{equation*}
where $d(h,g)$ is the left-invariant Haar measure on $\widetilde{\Gamma}_{W_0}\times \Uone$. By choosing sequences of functions $f_n$ which converges to the Dirac delta measure at the identity of $\widetilde{\Gamma}_{W_0}\times \Uone$ one sees that these elements are dense in $\fF$ in the norm topology. Since the subalgebra $\fF_0$ is invariant under the action $\tau^\xi$ it also contains a norm-dense $*$-subalgebra of smooth elements.
For the warped convolution $hW_0\mapsto \alpha_h((\fF_0)_{\xi,\kappa})$ of a net of $\CAR$-algebras we will use the shorthand notation $\fF_\kappa$.

%%%%%%%%%%%%%%%%%%%%%%%%%%%%%%%%%%%%%%%
\begin{remark}
	As we mentioned before, the free charged Dirac field provides an explicit example of a wedge-local field net of $\CAR$-algebras. For spin 1/2 fields there exists a unique de Sitter-invariant state with the Hadamard property \cite{AllenJacobson86, AllenLuetken86}. It is the analogue of the Bunch-Davies state \cite{Allen85} in the the spin 1/2 case. The Dirac field in this representation was studied in \cite{BartesaghiGazeauMoschellaTakook} and it was shown that it satisfies the so-called ``geometric KMS-condition''. By the same methods as in \cite{BorchersBuchholz:1999} one can prove that this condition implies the Reeh-Schlieder property of the state.
\end{remark}
%%%%%%%%%%%%%%%%%%%%%%%%%%%%%%%%%%%%%%%

%-----------------------------------------------------------
\subsubsection{Deformation fixed-points for observables}
%-----------------------------------------------------------
Let $\fF$ be a $\CAR$-net over $(\hs,\cc)$ in a quasifree representation $(\HS_S,\pi_S,\Omega_S)$ of a de Sitter- and gauge-invariant state. 
\begin{remark}
	As we mentioned before, all representations of the $\CAR$-algebra are faithful so will omit the $S$-dependence in our notation from now on. 
\end{remark}

\noindent
For the implementing operators we write
\begin{equation*}
	\pi(\alpha_h(B(f)))=U(h)\pi(B(f))U(h)^{-1},\qquad
	\pi(\sigma_s(B(f)))=V(s)\pi(B(f))V(s)^{-1}.
\end{equation*}
As $\alpha$ and $\sigma$ are strongly continuous, the representations $U$ and $V$ are also strongly continuous. Stone's theorem implies that the one-parameter group $\{V(s):s\in\Rl\}$ has a unique self-adjoint generator $Q$ with spectrum $\S\subset\Zl$ since $V(2\pi)=1$. Hence the representation space $\HS$ is $\S$-graded (charged sectors)
\begin{equation}
	\label{eq:QuasifreeDecomposition}
 \HS=\bigoplus_{n\in \cal{S}}\HS_n,\quad \HS_n=\{\Phi\in\HS:Q\Phi=n\Phi\}.
\end{equation}
From the transformation properties (\ref{eq:CAR_cospinors_gauge}) for (co)spinors follows that $\pi(\Psi(f_-))$ decreases charges by one and $\pi(\Psi^\dagger(f_+))$ increases charges by one, {\em i.e.} 
\begin{equation*}
    \pi(\Psi(f_-))\HS_n\subset \HS_{n-1},\quad 
    \pi(\Psi^\dagger(f_+))\HS_n\subset \HS_{n+1}
\end{equation*}
In the following we will frequently use the spectral decomposition $V(s)=\sum_{n\in \cal{S}}e^{isn}E(n)$, where $E(n)$ is the projector onto the eigenspace $\HS_n$ of $Q$. 

Before we determine the fixed-points of the deformation map for observables, we compute the warped convolution for intertwiners between charged sectors.
%%%%%%%%%%%%%%%%%%%%%%%%%%%%%%%%%%%%%%%
\begin{proposition}
	\label{prop:CARDeformation}
 Let $\pi(F)\in\BH$ be $\xi$-smooth such that $\pi(F)\HS_n\subset\HS_{n+m}$. Then
\begin{equation*}
 \pi(F)_{\xi,\kappa}=\sum_{n\in \cal{S}}U_\xi(\kappa n)\pi(F)U_\xi(-\kappa(n+m))E(n).
\end{equation*}
\end{proposition}
%%%%%%%%%%%%%%%%%%%%%%%%%%%%%%%%%%%%%%%
\begin{proof}
 Let $\Phi\in\HS^\infty_\xi$. Then
\begin{align*}
 \pi(F)_{\xi,\kappa}\Phi
&=\pi(F)_{\xi,\kappa}\sum_{n\in \cal{S}}E(n)\Phi\\
&=\sum_{n\in \cal{S}}\pi(F)_{\xi,\kappa}E(n)\Phi\\
&=\frac{1}{4\pi^2}\sum_{n\in \cal{S}}\lim_{\eps\ra 0}\int dv \int dv'e^{-ivv'}\chi(\eps v,\eps v')\U_\xi(\kappa\theta v)\pi(F)\U_\xi(-\kappa\theta v)^{-1}\U_\xi(v')E(n)\Phi\\
&=\frac{1}{4\pi^2}\sum_{n\in \cal{S}}\lim_{\eps \ra 0}\int dt ds \int dt' ds' e^{-i(tt'+ss')}
\chi_1(\eps t,\eps t')\chi_2(\eps s,\eps s')\cdot\\
&\hspace{4cm}\cdot
U_\xi(\kappa s)V(-\kappa t)\pi(F)V(-\kappa t)^{-1}U_\xi(\kappa s)^{-1}U_\xi(t')V(s')E(n)\Phi\\
&=\frac{1}{4\pi^2}\sum_{n\in \cal{S}}\lim_{\eps_1\ra 0}\int dt dt' \lim_{\eps_2\ra 0}\int ds ds' e^{-i(tt'+ss')}
\chi_1(\eps_1 t,\eps_1t')\chi_2(\eps_2 s,\eps_2s')\cdot\\
&\hspace{4cm}\cdot
U_\xi(\kappa s)V(-\kappa t)\pi(F)V(-\kappa t)^{-1}U_\xi(\kappa s)^{-1}U_\xi(t')V(s')E(n)\Phi\\
&=\frac{1}{4\pi^2}\sum_{n\in \cal{S}}\lim_{\eps_1\ra 0}\int dt dt' \lim_{\eps_2\ra 0}\int ds ds' e^{-i(tt'+ss')}
\chi_1(\eps_1 t,\eps_1t')\chi_2(\eps_2 s,\eps_2s')\cdot\\
&\hspace{4cm}\cdot
U_\xi(\kappa s)e^{-i\kappa tm}\pi(F) U_\xi(\kappa s)^{-1}U_\xi(t')e^{is'n}E(n)\Phi\\
&=\frac{1}{4\pi^2}\sum_{n\in \cal{S}}\lim_{\eps_1\ra 0}\int dt dt' \lim_{\eps_2\ra 0}\int ds ds' e^{-it(t'+\kappa m)}e^{-is'(s-n)}
\chi_1(\eps_1 t,\eps_1t')\chi_2(\eps_2 s,\eps_2s')\cdot\\
&\hspace{4cm}\cdot
U_\xi(\kappa s)\pi(F) U_\xi(\kappa s)^{-1}U_\xi(t')E(n)\Phi\\
&=\frac{1}{2\pi}\sum_{n\in \cal{S}}\lim_{\eps_1\ra 0}\int dt dt' e^{-it(t'+\kappa m)}
\chi_1(\eps_1 t,\eps_1t') U_\xi(\kappa n)\pi(F) U_\xi(\kappa n)^{-1}U_\xi(t')E(n)\Phi\\
&=\sum_{n\in \cal{S}}U_\xi(\kappa n)\pi(F) U_\xi(\kappa n)^{-1}U_\xi(-\kappa m)E(n)\Phi.
\end{align*}
In the first line we used the strong convergence of $\sum_{n\in \cal{S}}E(n)$ to the identity and the continuity of $\pi(F)_{\xi,\kappa}$ as an operator on $\HS$ for the second equality. Since the definition of the warped convolution (\ref{eq:WarpedConvolution}) does not depend on the cut-off function $\chi$ we choose $\chi(t,s,t',s')=\chi_1(t,t')\chi_2(s,s')$ with $\chi_k\in C^\infty_0(\Rl\times \Rl)$, $\chi_k(0,0)=1$, $k=1,2$.
For the fifth equality we use Fubini and regularize the integrals in the variables $s,s'$ and $t,t'$ separately by introducing cutoffs $\eps_1,\eps_2$ (see \cite{Rieffel:1992}). The behavior of $\pi(F)$ under gauge transformations and $V(s')E(n)=e^{is'n}E(n)$ is used in the sixth line. After that the $s'$-integration is performed and the Fourier transform of $\chi_2$ yields a factor $2\pi\delta(s-n)$ in the limit $\eps_2\ra 0$ since $\chi_2(0,0)=1$. Similarly we obtain a factor $2\pi\delta(t'+\kappa m)$ in the limit $\eps_1\ra 0$.
\end{proof}
%%%%%%%%%%%%%%%%%%%%%%%%%%%%%%%%%%%%%%%

%%%%%%%%%%%%%%%%%%%%%%%%%%%%%%%%%%%%%%%
\begin{remark}
    Specializing this proposition to $m=0$ yields the warped convolution for observables and $m=\pm1$ for (co)spinors.
\end{remark}
%%%%%%%%%%%%%%%%%%%%%%%%%%%%%%%%%%%%%%%

\noindent
\\
Next we determine the fixed-points of the map $\pi(A)\mapsto \pi(A)_{\xi,\kappa}$ for observables. For this purpose we need some basic facts about one-parameter unitary groups. The unitary operators $\{U_\xi(t):t\in\Rl\}$ form a strongly continuous one-parameter group and by Stone's theorem there exists a unique selfadjoint and (in general) unbounded operator $K_\xi$ (the generator the group) which is defined as
\begin{equation}
\label{eq:DefinitionGenerator}
	iK_\xi\Phi=\lim_{t\ra 0}\frac{1}{t}\bl(U_\xi(t)\Phi-\Phi\br),\quad \Phi\in D(K_\xi)
\end{equation} 
on the dense domain $D(K_\xi)=\{\Psi\in \HS:\lim_{t\ra 0}\bl(U_\xi(t)\Psi-\Psi)/t\text{ exists}\}$. 
For elements in $D(K_\xi)$ where $t\mapsto U_\xi(t)\Phi$ is smooth in $\|\cdot \|_{\HS}$ we write $D(K_\xi)^\infty$. Note that $D(K_\xi)^\infty$ is dense in $\HS$ and $\Phi\in D(K_\xi)^\infty$ if and only if $\Phi\in D((K_\xi)^l)$ for all $l\in\Nl$. If an operator commutes with $U_\xi(t)$ for all $t\in\Rl$, then $D(K_\xi)^\infty$ is invariant under its action. In particular we have
\begin{equation}
\label{eq:InvarianceDomain}
	U_\xi(t)D(K_\xi)^\infty\subset D(K_\xi)^\infty,\;t\in\Rl,\qquad E(n)D(K_\xi)^\infty\subset D(K_\xi)^\infty,\; n\in\S.
\end{equation}
Furthermore, for $F\in \fF^\infty_\xi$ there holds $\pi(F)D(K_\xi)^\infty\subset D(K_\xi)^\infty$ since $F$ is smooth with respect to boosts.

Since observables $A\in\fF$ are gauge-invariant, it follows that they are diagonal with respect to the orthogonal decomposition (\ref{eq:QuasifreeDecomposition}) of the representation space:
\begin{equation*}
	\pi(A)=\bigoplus_{n\in \S}\pi_n(A), \qquad \pi_n(A)=\pi(A)E(n):\HS_n\ra \HS_n.
\end{equation*} 		
Furthermore, since observables commute with gauge unitaries and leave charged sectors invariant, it follows that each $\pi_n:\Aa\subset\fF\ra\cal{B}(\HS_n)$ is a representation of the quasilocal algebra $\Aa$. As all representations of the $\CAR$-algebra are faithful, each $\pi_n,n\in\S$ is faithful. So if $\pi_n(A)=0$ for some $n\in \S$, then $A=0$, which implies $\pi_{m}(A)=0$ for all $m\in \S$ by linearity.

Obviously, from Proposition \ref{prop:CARDeformation} follows that $\pi(A)$ is invariant under the deformation if $\pi_n(A)=1$ for all $n\neq 0$, since $\sum_{n\in S}E(n)$ converges strongly to the identity. The following proposition provides a partial inverse to this statement.

%%%%%%%%%%%%%%%%%%%%%%%%%%%%%%%%%%%%%%%
\begin{proposition}
	Let $\Aa$ be the net of observables in a $\mathrm{CAR}$-net $\fF$ in a quasifree representation of a de Sitter- and gauge-invariant state with Reeh-Schlieder property.
	Let $A\in\Aa(W_0)$ be a $\xi$-smooth observable. If there exists an $\veps\in\Rl$ such that $\pi(A)_{\xi,\kappa}=\pi(A)$ for all $|\kappa|<\veps$, then $\pi(A)\in\Cl\cdot 1$.
\end{proposition}
%%%%%%%%%%%%%%%%%%%%%%%%%%%%%%%%%%%%%%%
\begin{proof}
	From the linearity of $\pi(A)\mapsto \pi(A)_{\xi,\kappa}$ together with the fact that each projection $E(n)$ is linear and commutes with boosts and gauge transformations follows
	\begin{equation}
	\label{eq:DecompositionObservableDeformed}
		\pi(A)_{\xi,\kappa}
		=\bigoplus_{n\in\S}\pi_n(A)_{\xi,\kappa}
		=\bigoplus_{n\in\S}\pi(A)_{\xi,\kappa}E(n).
	\end{equation}
	Consider now compact intervals $\Delta,\Delta'\subset\Rl$ and the spectral projections $\widetilde{E}_\xi(\Delta),\widetilde{E}_\xi(\Delta')$ of the generator $K_\xi$. For $\Phi,\Phi'\in D(K_\xi)^\infty$ define vectors $\Phi_\Delta:=\widetilde{E}_\xi(\Delta)\Phi$ and $\Phi'_{\Delta'}:=\widetilde{E}_\xi(\Delta')\Phi'$. As $\kappa\mapsto \pi(A)_{\xi,\kappa}$ is constant in a neighborhood of $\kappa=0$ there follows from (\ref{eq:DecompositionObservableDeformed})
	\begin{equation*}
	 0=(\Phi'_{\Delta'},\frac{d}{d\kappa}\pi(A)_{\xi,\kappa}E(n)\br|_{\kappa=0}\Phi_\Delta)_{\HS}.
	\end{equation*}
	Proposition (\ref{prop:CARDeformation}) for observables ($m=0$) implies
	\begin{equation*}
	 0
	=(\Phi'_{\Delta'},
	\frac{d}{d\kappa}U_\xi(\kappa n)\pi(A)U_\xi(-\kappa n)E(n)\br|_{\kappa=0}
	\Phi_\Delta)_{\HS}
	=in(\Phi'_{\Delta'},[K_\xi,\pi(A)]E(n)\Phi_\Delta)_{\HS}.
	\end{equation*}
	Note that $E(n)\Phi\in D(K_\xi)^\infty$ and  $\pi(A)\Phi\in D(K_\xi)^\infty$ for all $n\in \S,\Phi\in D(K_\xi)^\infty$ due to the invariance properties (\ref{eq:InvarianceDomain}) of $D(K_\xi)^\infty$. Hence
	\begin{equation*}
	 0=(\Phi'_{\Delta'},[(K_\xi)^l,\pi(A)]E(n)\Phi_\Delta)_{\HS}
	=(\Phi'_{\Delta'},[(K_\xi)^l,\pi_n(A)]\Phi_\Delta)_{\HS}
	\end{equation*}
	for all $n\neq 0$, $l\ge 0$. Since $\Phi_\Delta$, $\Delta\subset\Rl$ compact is an analytic vector for $K_\xi$, there follows
	\begin{equation*}
	 0=(\Phi'_{\Delta'},[U_\xi(t),\pi_n(A)]\Phi_\Delta)_{\HS}
	  =\sum_{l\ge 0}\frac{(it)^l}{l!}(\Phi'_{\Delta'},[(K_\xi)^l,\pi_n(A)]E(n)\Phi_\Delta)_{\HS}
	\end{equation*}
	As the linear span of $\{\Phi_\Delta:\Delta\subset\Rl\text{ compact},\,\Phi\in D(K_\xi)^\infty\}$ is dense in $\HS$ (see \cite[p.8]{Taylor:1986}), the bounded operator $[U_\xi(t),\pi_n(A)]$ vanishes on $\HS$ for $n\neq 0$. However, as
	\begin{equation*}
		0=U_\xi(t)\pi_n(A)U_\xi(t)^{-1}-\pi_n(A)=\pi_n(\alpha_{\lambda_\xi(t)}(A)-A)
	\end{equation*}
	for all $n\neq 0$ implies $\pi_m(\alpha_{\lambda_\xi(t)}(A)-A)=0$ for all $m\in \S$ it follows that $[U_\xi(t),\pi_n(A)]$ vanishes on $\HS$ for all $n\in\cal{S}$.
	
	 Since the GNS vector $\Omega$ is de Sitter invariant there holds
	\begin{equation*}
	 U_\xi(t)\pi(A)\Omega=\pi(A)U_\xi(t)\Omega=\pi(A)\Omega,
	\end{equation*}
	so the vector $\pi(A)\Omega$ is boost-invariant. In \cite[Lem.2.2]{BorchersBuchholz:1999} it is shown that boost-invariant vectors must in fact be invariant under the whole de Sitter group. Hence
	\begin{equation*}
		U(h)\pi(A)U(h)^{-1}\Omega=U(h)\pi(A)\Omega=\pi(A)\Omega,\quad h\in \tLo.
	\end{equation*}
	From the Reeh-Schlieder property of the state follows $U(h)\pi(A)U(h)^{-1}=\pi(A)$ since $\Omega$ is separating for $\Aa(W_0)$. Pick $h={\bj}_{W_0}$ and we find
	\begin{equation*}
	 \pi(A)\in \Aa(W_0)''\cap \alpha_{{\bj}_{W_0}}(\Aa(W_0)'')\subset \Aa(W_0)''\cap (\Aa(W_0)'')'.
	\end{equation*}
	by locality. Powers and St\o rmer \cite{PowersStormer} have shown that every quasifree and gauge-invariant representation of a $\CAR$-algebra is primary, so the local algebras are factors which implies $\pi(A)\in \Cl\cdot 1$.
\end{proof}
%%%%%%%%%%%%%%%%%%%%%%%%%%%%%%%%%%%%%%%

%-----------------------------------------------------------
\subsubsection{Unitary inequivalence}
%-----------------------------------------------------------
Let $\fF$ be a $\CAR$-net over $(\hs,\cc)$ in a Fock representation $(\HS_P,\pi_P,\Omega_P)$ of a de Sitter- and gauge-invariant state. An example for such a projection is $P=1\oplus 0$. It commutes with $\cc$ and all gauge transformations. Furthermore, if $u$ is a representation of $\tLo$ on $\hs$ of the form $u=u_1\oplus u_2$, where $u_1,u_2$ are representations of $\tLo$ on $H$ which commute with $C$ and are mutual adjoints of each other, then the associated state (\ref{eq:QuasifreeState}) is de Sitter- and gauge-invariant.

\begin{remark}
Again, we drop the $P$-dependence in our notation since all representations of the $\CAR$-algebra are faithful.
\end{remark}

In a Fock representation the gauge unitaries take the form $V(s)=e^{isQ}$, where $Q=N\otimes 1-1\otimes N$ is the charge operator and $N$ is the number operator on the Fermionic Fock space over $PH$. The Fock vacuum is invariant under gauge transformations. The spectrum of $Q$ is $\Zl$ and $\HS$ is $\Zl$-graded 
\begin{equation}
	\label{eq:FockSpaceDecomposition}
 \HS=\bigoplus_{n\in\Zl}\HS_n,\quad \HS_n=\{\Phi\in\HS:Q\Phi=n\Phi\}.
\end{equation}
The decomposition of $\HS$ into charged sectors $\HS_n$ and particle sectors is connected via
\begin{equation*}
 \HS_n=\bigoplus_{k-l=n}\wedge^kPH \otimes \wedge^lPH.
\end{equation*}
The grading is implemented by $Y=(-1)^\cal{N}$, where $\cal{N}$ is the number operator on $\HS$ (see \cite{Foit:1983}). In a Fock representation the (co)spinors take the form
\begin{align}
\label{eq:cospinor_Fockrep}
	\pi(\Psi(f_-))&=a^\dagger(0\oplus PCf_-)+a(Pf_-\oplus 0)\\
 \pi(\Psi^\dagger(f_+))&=a^\dagger(PCf_+\oplus 0)+a(0\oplus Pf_+).
\end{align}
A straightforward computation shows that the $\tLo$- and gauge-invariance of the state implies 
\begin{equation}
\label{eq:DeformationVanishesOnOmega}
 \pi(F)_{\xi,\kappa}\Omega=\pi(F)\Omega,\quad F\in\fF^\infty_\xi,
\end{equation}
since two unitaries drop out in (\ref{eq:WarpedConvolution}), which yields corresponding $\delta$-factors after integration.

Now we show that the deformed and undeformed nets are unitarily inequivalent. We proceed in a similar manner as in Chapter \ref{ch:Cosmological}. Consider the wedge $W_0$ and a rotation $r^\phi$ about an angle $\phi> 0$ in the $(x^1,x^2)$-plane. It is clear that the region
\begin{equation*}
 \C:=r^\phi W_0\cap r^{-\phi}W_0,\quad |\phi|<\pi/2
\end{equation*}
is a subset of $W_0$ and that the reflected region ${\bj}_{W_0}\C$ lies spacelike to $W_0$ and $r^\phi W_0$.

%%%%%%%%%%%%%%%%%%%%%%%%%%%%%%%%%%%%%%%
\begin{proposition}
 Let $\fF_\kappa$ be the warped convolution of the $\CAR$-net $\fF$ in a Fock representation of a de Sitter- and gauge-invariant state with Reeh-Schlieder property. Suppose that $u$ is a faithful\footnote{The unitary principal series representations of $\tLo$ are faithful \cite{Takahashi63}.} representation of $\tLo$ on $\hs$ which commutes with $C$ and $P$. Then the GNS vector $\Omega$ is not cyclic for $\fF(\C)_{\kappa}\,''$ for $\kappa\neq 0$. In particular, the nets $\fF$ and $\fF_\kappa$ are unitarily inequivalent for $\kappa\neq 0$.
\end{proposition}
%%%%%%%%%%%%%%%%%%%%%%%%%%%%%%%%%%%%%%%
\begin{proof}
Let $f_-\in H(\C)^\infty_\xi$. Hence $f_-,u(r^{-\phi})f_-\in H(W_0)^\infty_\xi$ and the warped operators $\pi(\Psi(f_-))_{\xi,\kappa}$, $\pi(\Psi(u(r^{-\phi})f_-))_{\xi,\kappa}$ are elements of $\fF(W_0)_{\kappa}$. From the de Sitter covariance (\ref{eq:DeformedOperatorsDeSitterCovariance}) follows
\begin{equation*}
 U(r^\phi)\pi(\Psi(u(r^{-\phi})f_-))_{\xi,\kappa}U(r^\phi)^{-1}
=\pi(\Psi(f_-))_{r^\phi_*\xi,\kappa},
\end{equation*}
which is an element of $\fF(r^\phi W_0)_{\kappa}$.
Assume now that $\Omega$ is cyclic for $\fF(\C)_{\kappa}\,''$. This is equivalent to $\Omega$ being cyclic for $\fF({\bj}_{W_0}\C)_{\kappa}\,''$ since $U({\bj}_{W_0})$ is unitary and $U({\bj}_{W_0})\Omega=\Omega$. Hence $\Omega$ is separating for $\fF({\bj}_{W_0}\C)_{\kappa}\,'$, which contains $\fF(W_0)_{\kappa}$ and $\fF(r^\phi W_0)_{\kappa}$ by locality. From (\ref{eq:DeformationVanishesOnOmega}) follows 
$\pi(\Psi(f_-))_{\xi,\kappa}\Omega
=\pi(\Psi(f_-))\Omega=\pi(\Psi(f_-))_{r^\phi_*\xi,\kappa}\Omega$, {\it i.e.}
\begin{equation}
\label{eq:UIsoughtcontradiction}
\pi(\Psi(f_-))_{\xi,\kappa}=\pi(\Psi(f_-))_{r^\phi_*\xi,\kappa}
\end{equation}
by the separating property of $\Omega$. Consider now a vector $\vphi\oplus 0\in PH\oplus PH$ of charge one in the one-particle space. Using Proposition \ref{prop:CARDeformation} for $m=-1$ we find
\begin{align*}
 \pi(\Psi(f_-))_{\xi,\kappa}(\vphi\oplus 0)
&=
\sum_{n\in\Zl}\pi(\Psi(u_\xi(\kappa n)f_-))U_\xi(\kappa)E(n)(\vphi\oplus 0)\\
&=
\Bl[a^\dagger(0\oplus PCu_\xi(\kappa)f_-))+a(Pu_\xi(\kappa)f_-\oplus 0))\Br](u_\xi(\kappa)\vphi\oplus 0)\\
&=(0\oplus PCu_\xi(\kappa)f_-)\wedge \bl(u_\xi(\kappa)\vphi\oplus 0\br)+ 
(Pu_\xi(\kappa)f_-\oplus 0,u_\xi(\kappa)\vphi\oplus 0)_{\HS}\Omega\\
&=U_\xi(\kappa)(0\oplus PCf_-)\wedge(\vphi\oplus 0)
+(Pf_-,\vphi)_{PH}\Omega.
\end{align*}
For the second equality we used that $\vphi\oplus 0$ has charge one and the explicit form (\ref{eq:cospinor_Fockrep}) of cospinors in a Fock representation. For the third equality we used the usual action of the Fermi creation and annihilation operators on Fock vectors. For the fourth equality we used that $u$ commutes with $C$ and $P$ and the fact that $U_\xi(\kappa)$ is the second quantization of $u_\xi(\kappa)$. By the same computation we find 
\begin{equation*}
 \pi(\Psi(f_-))_{r^\phi_*\xi,\kappa}(\vphi\oplus 0)
=U_{r^\phi_*\xi}(\kappa)(0\oplus PCf_-)\wedge(\vphi\oplus 0)+(Pf_-,\vphi)_{PH}\Omega.
\end{equation*}
As $\Psi^\dagger(f)_{\xi,\kappa}\Phi=\Psi^\dagger(f)_{r^\phi_*\xi,\kappa}\Phi$ for all $\Phi\in\HS^\infty_\xi$ by (\ref{eq:UIsoughtcontradiction}), there follows
\begin{equation*}
 U_\xi(\kappa)(0\oplus PCf_-)\wedge(\vphi\oplus 0)
=U_{r^\phi_*\xi}(\kappa)(0\oplus PCf_-)\wedge(\vphi\oplus 0).
\end{equation*}
Since $U$ is faithful, this implies $1=\lambda_\xi(-\kappa)r^\phi\lambda_\xi(\kappa)r^{-\phi}$ which yields $\lambda_\xi(\kappa)r^\phi=r^\phi\lambda_\xi(\kappa)$. However, for $\kappa\neq 0$ this is only true for $\phi=0$ since boosts in the $x^1$-direction do not commute with rotations in the $(x^1,x^2)$-plane and contradicts our initial assumption about the rotation $r^\phi$. 

Therefore, the operator $\Psi(f_-)_{r^\phi_*\xi,\kappa}$ does depend on $\phi$, so that the cyclicity assumption on $\Omega$ for $\fF(\C)_{\xi,\kappa}\,''$, $\kappa\neq 0$ is not valid. On the other hand, we know that $\Omega$ is cyclic for $\fF(\C)''$ by the Reeh Schlieder property of the state. A unitary which leaves $\Omega$ invariant and maps $\fF(\C)$ onto $\fF(\C)_{\xi,\kappa}$ would preserve this property, from which we conclude that the undeformed and deformed net are not unitarily equivalent.
\end{proof}
%%%%%%%%%%%%%%%%%%%%%%%%%%%%%%%%%%%%%%%

%===========================================================
\subsection{Deformations with other Abelian subgroups} 
\label{subsec:DefSubgroupsDeSitter}
%===========================================================
In the course of writing up this thesis also partial negative results were obtained, which we would like to briefly comment on. 

The warped convolution which uses a combination of boosts and internal symmetries can, in principle, also be defined for quantum field theories on Minkowski space. However, the covariance properties of the deformed operators are very different in this setting and a statement similar to Theorem \ref{thm:main} seems not to hold. The reason for this is that $\Gamma_{W_0}$ is, in contrast to the translations, not a normal subgroup of the Poincar\'e group. For a Poincar\'e group element $(a,\Lambda)$ one has
\begin{equation*}
 \alpha_{(a,\Lambda)}(F_{\xi,\kappa})=\alpha_{(a,\Lambda)}(F)_{(a,\Lambda)_*\xi,\kappa}
\end{equation*}
and 
\begin{equation}
\label{eq:PoincareConjugation}
 (a,\Lambda)(0,\Lambda(t))(a,\Lambda)^{-1}=(-\Lambda\Lambda(t)\Lambda^{-1}a+a,\Lambda\Lambda(t)\Lambda^{-1})
\end{equation}
is the flow which is associated with $(a,\Lambda)_*\xi$. Observe that the Lorentz group acts on $\Gamma_{W_0}$ merely by conjugation. For a translation $(a,1)W_0\subset W_0$ there is
\begin{equation*}
 \alpha_{(a,1)}(F_{\xi,\kappa})\neq \alpha_{(a,1)}(F)_{\xi,\kappa},
\end{equation*}
in general, since also a translational part is involved in (\ref{eq:PoincareConjugation}).

An interesting question is whether other Abelian subgroups of the de Sitter group can be used to define quantum field theories in terms of warped convolutions. A complete classification (up to conjugacy) of all subgroups of $\L_0$ in terms of subalgebras of its Lie algebra was given in \cite{PateraWinternitzZassehaus1976} (see also \cite{Shaw70, Hall04} for the $\SO(1,3)_0$ case). The Lie algebra generators $M_{\mu\nu}=-M_{\nu\mu}$,  $\mu,\nu=0,\dots,4$ of $\L_0$ satisfy the relations
\begin{equation*}
 [M_{\mu\nu},M_{\rho\sigma}]=\eta_{\mu\rho}M_{\nu\sigma}+\eta_{\nu\sigma}M_{\mu\rho}-\eta_{\nu\rho}M_{\mu\sigma}-\eta_{\mu\sigma}M_{\nu\rho}.
\end{equation*}
The matrices $M_{0k}$, $k=1,\dots,4$ generate boosts in the direction $x^k$ and the $M_{jk}$, $j,k=1,\dots,4$ generate spatial rotations in the $(x^j,x^k)$-plane. The two-dimensional Abelian subgroups of $\L_0$ are listed in table \ref{deSitterSubgroups}.
\begin{table}[t]
\begin{center}
	\begin{tabular}{ccc}
	\hline
	\hline
		 & subgroup & Lie algebra generators\\
	\hline
	$\L_1$ & $\SO(2)\times \SO(2)$ & $M_{12},M_{34}$\\
	$\L_2$ & $\mathrm{O}(1,1)\times \SO(2)$ & $M_{01},M_{23}$\\
	$\L_3 $& $\Rl^2$ & $M_{12}-M_{01},M_{23}-M_{03}$\\
	$\L_4$ & $\Rl\times \SO(2)$ & $M_{12}-M_{01},M_{34}$\\
	\hline
	\hline
	\end{tabular}	
\end{center}
\caption{\label{deSitterSubgroups}Two-dimensional Abelian subgroups of the de Sitter group.}
\end{table}
$\L_1$ consists of spatial rotations in the $(x^1,x^2)$- and $(x^3,x^4)$-plane. $\L_2$ are boosts in the $x^1$-direction and rotations in the $(x^2,x^3)$-plane. $\L_3$ corresponds to null rotations (translational part of the stabilizer group of a light ray). $\L_4$ is a combination of a null rotation and a spatial rotation. All of these groups can be used to define a warped convolution with the associated $\Rl^2$-action from the representation. 
Since we are on a curved spacetime it appears to be natural to require that the group which is used for the deformation is a subgroup of the stabilizer of a wedge. The reason is that there is not an analogue of the spectrum condition on Minkowski space available which restricts the spectral properties of the generators which are associated with isometries (the microlocal spectrum condition only gives a restriction on the singularity structure of the two-point function). Comparing the subgroup structure of $\L_0(W_0)$ with the above groups shows that only $\L_2$ is a subgroup. However, $\L_2$ violates conditions a) and b) in Definition \ref{def:CausalBorchersSystem} for certain reflections: Denote by $F_{\zeta,\kappa}$ the warped operator, where $\zeta=(M_{01},M_{23})$ is a pair of Killing vector fields (compare Chapter \ref{ch:Cosmological}) and consider the reflection $j_{12}(x^0,x^1,x^2,x^3,x^4)=(x^0,-x^1,-x^2,x^3,x^4)$ which satisfies $j_{12}W_0=(W_0)'$. The associated flow $\Lambda(t,s):=\exp(t M_{01})\exp(s M_{23}),\,  t,s\in\Rl$ transforms as
\begin{equation*}
 j_{12}\Lambda(t,s)j_{12}=\Lambda(-t,-s)
\end{equation*}
so that $\alpha_{j_{12}}(F_{\zeta,\kappa})=\alpha_{j_{12}}(F)_{{j_{12}}_*\zeta,\kappa}=F_{\zeta,\kappa}$ and condition b) is violated. Similar problems also appear if one uses a combination of boosts and translations along the edge of the wedge in the case of Minkowski spacetime. From these observations we conclude that the position of the subgroup, which is used for the deformation, within the isometry is very important and that a modification of the standard warping formula is necessary in these cases. 

A deformation with purely internal symmetries, {\em e.g.} $\Uone\times\Uone$, did not appear to be interesting, because an adaption of Proposition \ref{prop:CARDeformation} to this case yields that the deformation is trivial on the level observables and also trivial for generators $B(f)$, provided that the induced charge structure of the gauge groups is the same.

%%%%%%%%%%%%%%%%%%%%%%%%%%%%%%%%%%%%%%%%%%%%%%%%%%%%%%%%%%%%
\chapter{Summary and Outlook}
\label{ch:outlook}
%%%%%%%%%%%%%%%%%%%%%%%%%%%%%%%%%%%%%%%%%%%%%%%%%%%%%%%%%%%%
In this work we studied deformations of quantum field theories in various situations where the relativistic spectrum condition is not fulfilled. In particular, we have extended the warped convolution deformation procedure to a large class of globally hyperbolic spacetimes. 
\\
\\
For the scalar free field in a thermal representation we were able to establish wedge-locality of the deformed models and inequivalence of the corresponding nets. The inequivalence proof showed, in particular, that the thermal Fock vacuum is not separating for unions of wedge algebras, which implies that the initial KMS state is not KMS on the deformed net. Whether the deformed net admits other KMS states, {\em i.e.} if the deformed theory still has a good thermodynamic behavior, is, however, open at this point. This question is closely connected to thermodynamic properties of deformed vacuum fields \cite{Huber:2011}. 

Assuming that the group of isometries contains a two-dimensional Abelian subgroup generated by two commuting and spacelike Killing fields, we have demonstrated that many results known for Minkowski space theories carry over to the curved setting by formulating concepts like edges, wedges and translations in a geometrical language. In particular, it has been shown in a model-independent framework that the basic covariance and wedge-localization properties we started from are preserved for all values of the deformation parameter $\kappa$. As a concrete example we considered a warped Dirac field on a Friedmann-Robertson-Walker spacetime in the GNS representation of a quasifree and $\iso$-invariant state with Reeh-Schlieder property. It was shown that the deformed models depend in a non-trivial way on $\kappa$, and violate the Reeh-Schlieder property for regions smaller than wedges. 
At the current stage, it is difficult to give a clear-cut physical interpretation to the models constructed here since scattering theory is not available for quantum field theories on generic curved spacetimes. Nonetheless, it is interesting to note that in a field theoretic context the deformation leaves the two-point function invariant (Proposition \ref{prop:n-pt}), the quantity which is most frequently used for deriving observable quantum effects in cosmology (for the example of quantized cosmological perturbations, see \cite{MukhanovFeldmanBrandenberger:1990}). So, when searching for concrete scenarios where deformed quantum field theories can be matched to measurable effects, one has to look for phenomena involving the higher $n$-point functions.

By using a one-parameter group of boosts and a global $\Uone$ gauge group we deformed quantum field theories on de Sitter spacetime in such a way that wedge-localization and de Sitter covariance is preserved for all values of the deformation parameter.
This provides the first example of a warped convolution which is formulated in terms of internal symmetries. For models which arise from inclusions of $\CAR$-algebras we determined the fixed-points of the deformation map and showed that the deformed and undeformed theories are unitarily inequivalent. In section \ref{subsec:DefSubgroupsDeSitter} we gave a complete list of all two-dimensional Abelian subgroups of the de Sitter group. We argued that all of them can in principle be used to define a warped convolution, but that the deformation of a wedge triple by means of these subgroups is no longer a wedge triple. We surmise that a modification of kernel in the warping formula (\ref{eq:IntroWarpedConvolution}) is required in these cases. 
\\
\\
There exist a number of interesting directions in which this research could be extended.  

Due to the tensor product structure of the Fock space in the Araki-Woods representation there exist other deformation procedures which are different from the one we used. For example, one can think of a ``thermalization'' of the deformed vacuum fields $\phi_\theta(f)$ in \cite{GrosseLechner:2007} of the form $\phi_{\theta,\beta}(f):=\phi_\theta(\sqrt{1+\rho}\,f)\otimes \bar{1}+1\otimes \overline{\phi_\theta(\sqrt{\rho}\,f)}$. These fields are also wedge-local but their $n$-point functions are different than the ones we computed in section \ref{sec:Deformationsofthermalscalarfreefields}. Another possibility is to consider charged fields on the thermal Fock space and deformations with opposite sign on the two tensor product factors. However, none of the deformations were studied in detail due to limitations of time, but I believe they deserve further investigation. 

As far as the geometrical construction of wedge regions in curved spacetimes is concerned, we limited ourselves to edges which are generated by commuting Killing vector fields. This assumption rules out many physically interesting spacetimes, such as de Sitter, Kerr or Friedmann-Robertson-Walker spacetimes with compact spatial sections. An extension of the geometric construction of edges and wedges to such spaces seems to be straightforward and is expected to coincide with the notions which are already available. 

In the case of de Sitter spacetime it is desirable to make contact with the deformation scheme from chapter \ref{ch:Cosmological}. There the Killing flow which is associated with the edge of a wedge was used to formulate the deformation. In de Sitter space, edges have the topology of a two-sphere, which is an $\SO(3)$ orbit. A generalization of the warped convolution deformation formula could involve an integration over this group instead of $\Rl^2$. But deformations of C$^*$-algebras based on actions of this group are currently not yet available (see however \cite{Bieliavsky2002} for certain non-Abelian group action). 

A proof that intersections of wedge algebras are non-trivial (step III in section \ref{sec:CAQFT}) seems to be out of reach at the moment, since there are no sensible criteria available so far. However, in view of the picture that the deformed models can be regarded as effective theories on a noncommutative spacetime, where strictly local quantities do not exist because of space-time uncertainty relations, it is actually expected that they do not contain operators sharply localized in bounded regions for $\kappa\neq0$. A clarification of this conjecture would be desirable.

%===========================================
\begin{appendix}
%%%%%%%%%%%%%%%%%%%%%%%%%%%%%%%%%%%%%%%%%%%%%%%%%%%%%%%%%%%%
\chapter{Conventions and Notation}
\label{ch:NotationConventions}
%%%%%%%%%%%%%%%%%%%%%%%%%%%%%%%%%%%%%%%%%%%%%%%%%%%%%%%%%%%%
{\bf Units.} We work in Planck units $c=\hbar=G=k_B=1$.
\\
\\
\\
{\bf Minkowski spacetime.} We denote $n$-dimensional Minkowski spacetime by $(\Rl^n,\eta)$, where $\eta=\diag(+1,-1,\dots,-1)$ is the Minkowski metric. For the Minkowski product of vectors $x,y\in\Rl^n$ we write $x\cdot y:=x^0y^0-\sum_{k=1}^{n-1}x^ky^k$ and for partial derivatives we use the shorthand notation $\del_\mu:=\del/\del x^\mu$ with $\mu=0,\dots, n$.
Further properties of four-dimensional Minkowski space and its isometry group are discussed in section \ref{sec:MinkowskiSpacetime}. For the description of four-dimensional de Sitter spacetime we use five-dimensional Minkowski space which is treated in section \ref{sec:deSitterSpacetime}.
\\
\\
\\
{\bf Fourier transform.} For the Fourier transform $\Ff:L^2(\Rl^n,dx)\ra L^2(\Rl^n,dp)$ and its inverse $\Ff^{-1}:L^2(\Rl^n,dp)\ra L^2(\Rl^n,dx)$ on Minkowski spacetime we use the sign convention
\begin{equation*}
(\Ff f)(p)=\tilde{f}(p)=\int dx\; e^{ip\cdot x}f(x),\qquad
(\Ff^{-1} \tilde{f})(x)=f(x)=(2\pi)^{-n}\int dp\; e^{-ip\cdot x}\tilde{f}(p).
\end{equation*}
\\
\\
\\
{\bf Curved spacetimes.}\label{spacetime}\label{spacetimemetric} A spacetime manifold (or spacetime for short) will be denoted by $(M,\gST)$, where $M$ is a four-dimensional, connected, smooth manifold and $\gST$ is a smooth, Lorentzian metric with signature $(+,-,-,-)$. It is automatically guaranteed that $M$ is paracompact and second countable \cite{Geroch:1968,Geroch:1970}.

\newpage
\noindent
{\bf Frequently used symbols.}
%===========================================================
\begin{center}
\begin{longtable}{lll}
Symbol & Description & Reference\\
\hline
&\\
$\|\cdot \|$ & operator norm & $-$ \\
$\|\cdot \|_{\HS}$ & Hilbert space norm & $-$ \\
$\|\cdot\|_m$ & norm associated with $\left<.\,,.\right>_m$& page \pageref{massshellnorm}\\
$\left<.\,,.\right>_m$ & inner product on $\HS_{1}$& page \pageref{eq:SFFscalarproduct}\\
$\frak{A}$ & (wedge-)local net of observables& page \pageref{localnet}, \pageref{wedgelocalnet}\\  
$\frak{A}_m$ & Weyl algebra for the massive scalar free field& page \pageref{weylalgebraSFF}\\
$\frak{A}_\beta$ & Weyl algebra for the massive scalar free field&\\ 
&in the Araki-Woods representation& page \pageref{weylalgebraSFFthermal}\\
$a(\vphi),a^\dagger(\vphi)$ & vacuum annihilation/creation operators& page \pageref{vacuumcreationannihilationoperators}\\
$a_\beta(\vphi),a_\beta^\dagger(\vphi)$ & thermal annihilation/creation operators& page \pageref{thermalcreationannihilationoperators}\\
$a_\theta(\vphi),a_\theta^\dagger(\vphi)$ & deformed vacuum annihilation/creation operators& page \pageref{vacuumdeformedcreationannihilationoperators}\\
$a_{\beta,\theta}(\vphi),a_{\beta,\theta}^\dagger(\vphi)$ & deformed thermal annihilation/creation operators& page \pageref{thermaldeformedcreationannihilationoperators}\\
$\BH$ & bounded operators on $\HS$& $-$\\
$C^\infty(X)$ & smooth functions on $X$& $-$\\
$C^\infty_0(X)$ & smooth functions with compact support on $X$& $-$\\
$\CCR(\hhs,\sigma)$ & Weyl algebra over a real symplectic space $(\hhs,\sigma)$& page \pageref{abstractweylalgebra}\\
$d\mu_m$ & Lorentz-invariant measure on $H_m^+$& page \pageref{lorentzinvariantmeasure}\\
$D(A)$ & domain of the operator $A$& $-$\\
$E_4(3)$ & extended Euclidean group $\Rl^4\rtimes\SO(3)_0$& $-$ \\
$E_W$ & edge of the wedge $W$& page \pageref{edgeMinkoski}, \pageref{edgedeSitter}\\
$E_{\xi,p}$ & edge of the wedge $W_{\xi,p}$& page \pageref{def:edge}\\
$\cal{F}(\HS)$ & Bosonic Fock space over $\HS$& page \pageref{bosefockspace}\\
$\frak{F}$ & (wedge-)local field net& page \pageref{fieldnetcosmological}, \pageref{fieldnetdesitter}\\
$\gdS$ & spacetime metric& page \pageref{spacetimemetric}\\
$\Gamma_W$ & one-parameter group of boosts associated with $W$& page \pageref{onepargroupwedge}\\
$H_m^+$ & upper mass shell& page \pageref{masshyperboloid}\\
$\HS$ & complex Hilbert spaces& $-$\\
$\HS^\infty$ & smooth vectors in $\HS$& page \pageref{smoothvectors}, \pageref{smoothvectorsdeSitter}\\
$\HS^\infty_\xi$ & $\xi$-smooth vectors in $\HS$& page \pageref{xismoothvectorsdeSitter}\\
$\overline{\HS}$ & conjugate of $\HS$& page \pageref{conjugatehilbertspace}\\
$\HS_n$ & $n$-fold symmetric tensor product of $\HS$& $-$\\
$\wedge^n\HS$ & $n$-fold antisymmetric tensor product of $\HS$ & $-$\\
$\HS_1$ & one-particle Hilbert space for the scalar free field& page \pageref{oneparticlespaceSFF}\\
$(\hhs,\sigma)$ & real symplectic vector space for the Weyl algebra & page \pageref{realsymplecticvectorspaceWeylAlgebra}\\
$(\hhs_m,\sigma_m)$ & real symplectic vector space for the Weyl algebra&\\
& of the scalar free field of mass $m$& page \pageref{realsymplecticvectorspaceWeylAlgebraSFF}\\
$(M,\gST)$ & spacetime manifold& page \pageref{spacetime}\\
$\Iso(M,\gdS)$ & isometry group of $(M,\gdS)$& $-$\\
$J^\pm(\OO)$ & causal future/past of $\OO\subset M$& $-$\\
$\Mat(n,\mathbb{K})$ & $n\times n$ matrices over the (skew)field $\mathbb{K}$ & $-$\\
$\OO$ & spacetime region& $-$\\
$\OO'$ & causal complement of $\OO$& page \pageref{causalcomplement}\\
$\omega$ & state& page \pageref{state}\\
$\omega_0$ & vacuum state& page \pageref{vacuumstate}\\
$\omega_\beta$ & KMS state with inverse temperature $\beta$& page \pageref{KMSstate}\\
$\Omega$ & Fock vacuum in the vacuum representation& page \pageref{fockvacuumvacuum}\\
$\Omega_\beta$ & Fock vacuum in the thermal representation& page \pageref{fockvacuumthermal}\\
$\Ss(\Rl^4,\Rl)$ & real-valued Schwartz functions on $\Rl^4$& $-$\\
$\Sigma$ & Cauchy hypersurface& $-$\\
$V^+$ & forward light cone& $-$\\
$W$ & wedge& page \pageref{wedgeMinkowski}, \pageref{wedgedeSitter}\\
$W_0$ & reference wedge& page \pageref{wedgeMinkowski}, \pageref{wedgedeSitter}\\
$W_{\xi,p}$ & wedge with base point $p$ and direction $\xi$& page \pageref{wedgecosmological}\\
$\W$ & the set of all wedges& page \pageref{familywedgesMinkowski}, \pageref{wedgecosmological}, \pageref{familywedgesdeSitter}
\end{longtable}
\end{center}
%===========================================================
\include{APPENDIX}
\end{appendix}
%===========================================
%%%%%%%%%%%%%%%%%%%%%%%%%%%%%%%%%%%%%%%%%%%%
%
% Bibliography
%
%%%%%%%%%%%%%%%%%%%%%%%%%%%%%%%%%%%%%%%%%%%%
\clearpage
\newcommand{\etalchar}[1]{$^{#1}$}

%===========================================
\thispagestyle{empty}
\cleardoublepage
\thispagestyle{empty}
\noindent
{\bf\LARGE Acknowledgements}
\addcontentsline{toc}{chapter}{Acknowledgements}
\\
\\
\\
I would like to thank Prof. $\!\!$J. Yngvason for giving me the opportunity to work in this interesting field of mathematical physics and for introducing me to the realm of rigorous quantum field theory.
\\
\\
I thank Prof. $\!$H. Grosse and Prof. $\!$S. Hollands for their immediate willingness to write a report on this thesis.
\\
\\
I am deeply grateful to Dr. G. Lechner for guiding me through the course of my PhD. He was a very attentive supervisor and without his advice and encouragement this thesis would not have been possible.
\\
\\
In the course of my work I profited from discussions with many people. In particular I would like to thank M. Bischoff, J. Bros, C. Dappiaggi, W. Dybalski, T.-P. Hack, C. J\"akel, C. K\"ohler, U. Moschella, N. Pinamonti, J. Queva, K.-H. Rehren, F. Robl, H. Rumpf, A. Schenkel, J. Schlemmer , H. Steinacker, S. Waldmann and J. Zahn.
\\
\\
I thank S. Hansen, G. Lechner and J. Schlemmer for proofreading parts of this thesis.
\\
\\
This work was supported by the Marie Curie Research Training Network MRTN-CT-2006-031962 EU-NCG and by the project P22929–N16 funded by the Austrian
Science Fund (FWF).
\\
\\
Finally, I thank my mother and Mi Na for their support and love.

\newpage
%===========================================
\thispagestyle{empty}
\cleardoublepage
\thispagestyle{empty}
\noindent
{\bf\LARGE Lebenslauf}
\\
\\
\\
\\
\begin{tabular}{rl}
 Name & Morfa Morales, Eric\\
 Geburtsdatum & 05.01.1982\\
 Geburtsort & Eisenach\\
 &\\
 Seit 03.2009 & Doktoratsstudium Physik, Universit\"at Wien\\
 10.2008 & Diplom Physik, Georg-August-Universit\"at G\"ottingen\\
 & Titel der Diplomarbeit: ``On a second law of black hole mechanics\\
 & in a higher derivative theory of gravity''\\
 & Betreuer: Prof. S. Hollands\\
 10.2002 $-$ 10.2008 & Diplomstudiengang Physik, Georg-August-Universit\"at G\"ottingen\\
 06.2001 & Abitur, Ernst-Abbe-Gymnasium Eisenach  
\end{tabular}

%===========================================
\end{document}